\documentclass[12pt]{article}
\usepackage{jheppub}
\usepackage{mathtools}
\newcommand\Tr{\mathrm{Tr}\,}
\newcommand\STr{\mathrm{STr}\,}
\newcommand\CC{\mathrm{CC}}
\newcommand\CH{\mathrm{CH}}
\newcommand\HH{\mathrm{HH}}
\newcommand\HC{\mathrm{HC}}
\newcommand\BH{\mathrm{BH}}
\newcommand\BC{\mathrm{BC}}
\newcommand\bC{\mathbb{C}}

\newcommand\bZ{\mathbb{Z}}
\newcommand\cT{\mathcal{T}}

\newcommand\cB{\mathcal{B}}
\newcommand\cC{\mathcal{C}}
\newcommand\cD{\mathcal{D}}
\newcommand\cA{\mathcal{A}}

\newcommand\cO{\mathcal{O}}
\newcommand\cM{\mathcal{M}}
\newcommand\cL{\mathcal{L}}
\newcommand\cI{\mathcal{I}}
\newcommand\fo{\mathfrak{o}}
\newcommand\fc{\mathfrak{c}}

\newcommand\fh{\mathfrak{h}}

\newcommand\fp{\mathfrak{p}}
\newcommand\fL{\mathfrak{L}}
\newcommand\fP{\mathfrak{P}}
\newcommand\fD{\mathfrak{D}}
\newcommand\fM{\mathfrak{M}}
\newcommand\Sym{\text{Sym}}

\newcommand{\BV}{\mathrm{BV}}
\newcommand{\cl}{\mathrm{cl}}
\newcommand{\Hom}{\mathrm{Hom}}
\newcommand{\End}{\mathrm{End}}

\usepackage{amsthm}
\usepackage[dvipsnames]{xcolor}
\usepackage{cancel}
\usepackage{tikz-cd}
\newtheorem{theorem}{Theorem}

\newcommand\wt{\widetilde}
\DeclareMathAlphabet\EuScript{U}{eus}{m}{n}
\SetMathAlphabet\EuScript{bold}{U}{eus}{b}{n}

\theoremstyle{definition}

\title{Categorical 't Hooft expansion and chiral algebras}
\abstract{Twisted holography captures protected aspects of well-known holographic dualities. We show how an holographic dual B-model background can be systematically derived from the 't Hooft expansion of the chiral algebras associated to four-dimensional ${\cal N}=2$ superconformal quiver gauge theories. A crucial tool is the match of planar BRST anomalies in the field theory and on the  worldsheet, especially in the presence of probe D-branes. Our construction is very general and can be applied to chiral algebras which do not have a four-dimensional origin. The resulting holographic dual backgrounds are typically non-geometric and appear to be novel. We expect our strategy to have a wide range of applications to other examples of twisted holography and, potentially, weak coupling holography.}
\author[1]{Davide Gaiotto,}
\author[1,2]{Adri\'an L\'opez-Raven,}
\author[1,2]{Hanne Silverans,}
\author[1,3]{Keyou Zeng}
\affiliation[1]{Perimeter Institute for Theoretical Physics, Waterloo, ON N2L 2Y5, Canada}
\affiliation[2]{Department of Physics \& Astronomy, University of Waterloo, Waterloo, ON N2L 3G1,
Canada}
\affiliation[3]{Center of Mathematical Sciences and Applications, Harvard University, Massachusetts 02138, USA}

\newcommand{\hs}[1]{\textcolor{Plum}{\textbf{Hanne}: #1}}

\begin{document}
\maketitle

\newpage
\section*{Notation}
\vspace{5pt}
\begin{center}
\begingroup
\renewcommand{\arraystretch}{1.4}
\begin{tabular}{ll}
$\lambda$ & 't Hooft coupling \\

$\text{CC}^{\bullet}(A) $ & Connes' complex for the cyclic cohomology of $A$ \\

$\text{HC}^{\bullet}(A)$ & cyclic cohomology of $A$ \\

$\text{CH}^{\bullet}(A,M)$ & Hochschild cochain complex of $A$ valued in $M$ \\

$\text{HH}^{\bullet}(A,M)$ & Hochschild cohomology of $A$ valued in $M$ \\

$\text{HH}^{\bullet}_{\lambda}(A),\text{HC}^{\bullet}_{\lambda}(A)$ & deformed Hochschild and cyclic cohomology \\

$Q$ or $Q_{\mathrm{BRST}}$ & full BRST differential\\

$Q_0$ & tree-level part of the BRST differential \\
 
$\hbar Q_1$ & 1-loop part of the BRST differential \\

$\lambda Q_1^l$ & linear (in $\lambda$) part of the 1-loop BRST differential \\

$\mathcal{A}_{a,b},\mathcal{B}_{a,b},\mathcal{C}_{a,b},\mathcal{D}_{a,b}$ & four towers of single trace operators\\

$\mathfrak{L}_{\lambda}$ & global symmetry algebra of single-trace operators \\

$\mathfrak{P}_{\lambda}$  & global symmetry algebra of mesonic operators \\

$\mathfrak{D}_{\lambda}$ & global symmetry algebra of determinant operators \\

$\mathfrak{M}_{\lambda}$ & bi-module associated to open determinant modifications \\

$\mathfrak{L}_{0},\mathfrak{P}_{0},\mathfrak{D}_{0},\mathfrak{M}_{0}$ & tree-level limit ($\lambda \to 0$) of $\mathfrak{L}_{\lambda},\mathfrak{P}_{\lambda},\mathfrak{D}_{\lambda},\mathfrak{M}_{\lambda}$ \\

$\text{BH}^{\bullet,\bullet}(V,Q)$ & planar global symmetry algebra \\

$\text{PMod}_V$ & category of planar $V$-modules\\

$\otimes_A$ & will denote the derived tensor product $\otimes^\mathbb{L}_A$
\end{tabular}
\endgroup
\end{center}

\newpage
\section{Introduction}
Certain families of gauge theories with classical gauge groups admit a 't Hooft expansion~\cite{tHOOFT1974}: a reorganization of the perturbative expansion where the rank $N$ of the gauge group is treated as being of order $\hbar^{-1}$. The resulting expansion is in many way analogous to the genus expansion of a String Theory. 
In particular, it naturally includes the analogue of D-branes and open string sectors associated to pairs of D-branes. 

The 't Hooft expansion is key to the holographic dictionary whenever the gravitational side of the duality involves String Theory \cite{maldacena1999large, Witten:1998qj}. Standard weakly-curved ten-dimensional String Theories emerge at large values of the 't Hooft coupling $\lambda \equiv \hbar N$. The duality is still expected to hold at small $\lambda$, but the dual String Theory background is strongly curved or perhaps non-geometric, i.e. described by a world-sheet theory which is not a sigma-model \cite{Haggi-Mani:2000dxu,Dhar:2003fi,Clark:2003wk,Karch:2002vn,Gopakumar:2003ns,Gopakumar:2003ns,Aharony:2006th,Yaakov:2006ce,Aharony:2007fs,Berkovits:2008qc,Berkovits:2019ulm,Gaberdiel:2021jrv}. 

This ``weak coupling'' regime in holography is of great interest but poorly understood. More generally, our inability to define non-geometric String Theory backgrounds hampers many potential applications of the 't Hooft expansion, such as the formulation of a String Theory dual to $SU(N)$ Yang-Mills theory.\footnote{Note that non-geometric backgrounds pose two challenges. The obvious one is to define a world-sheet theory for the String Theory. A more subtle one is to address IR divergences of String Theory without the guidance of a low energy effective QFT description in the target space.}

A priori, it is not known if a generic quantum field theory which admits a 't Hooft expansion should always admit a String Theory dual description, in the sense of a specific world-sheet theory whose genus expansion matches the t'Hooft expansion of the QFT, with boundary conditions matching all possible D-brane-like objects in the QFT. 

We would like to conjecture that this is indeed the case, and furthermore that there is a systematic way to translate the QFT data into the definition of a worldsheet theory. This conjecture is certainly implicit in much work on String Theory and holography, but we think it deserves an explicit formulation. We will refer to it as the conjecture that String Theory is ``'t Hooft complete''.\footnote{One could perhaps distinguish a strong and weak forms of the conjecture, requiring the $\lambda$ expansion to always converge or allowing perturbative-in-$\lambda$ constructions.}

The notion of 't Hooft completeness poses a conceptual challenge: it requires the existence of world-sheet theories which can systematically reproduce the infinite variety of Feynman diagram expansions which may occur in large $N$ QFTs. There have been several attempts to 
do so for specific theories \cite{Gopakumar:2003ns,Itzhaki:2004te,Razamat:2008zr,Gopakumar:2011ev,Gopakumar:2004qb,Gopakumar:2005fx}, sometimes with partial success, but no general prescription is known. 

Typically, the constructions rely on a conformal gauge perspective: the worldsheet theory is described as the BRST reduction of a 2d CFT coupled to a ghost system. We suspect that this assumption may be problematic. At the very least, we find it hard to imagine how the combinatorial data of the Feynman diagrams could be universally reorganized into the data of 2d CFTs.

Strictly speaking, the standard String Theory formalism only requires the world-sheet theory to be a ``dg-TQFT'', i.e. a quantum field theory whose stress-tensor is BRST-exact (sometimes denoted as CohFT) \cite{Witten:1990bs}. Via descent relations, such a dg-TQFT can be used to define integrands for a consistent collection of String Theory amplitudes.\footnote{Assuming that a prescription can be found to deal with ``IR'' divergences from the boundaries of the integration region.} 

Unitary, non-cohomological TQFTs such as 3d Chern-Simons theory often admit alternative algebraic/categorical definitions which dispense with local degrees of freedom. For conciseness, we will refer to such definitions as ``TFTs''. The algebraic and categorical structures which can occur in a dg-TQFTs are more intricate and mathematically very rich, but a ``dg-TFT'' description or definition is still possible using tools from Homological Algebra \cite{Kontsevich:1994qz,costello2007topological,2005math......9264C,2006math......1130C,2006math......5647C}.  

A general theme in TFT is that a theory which admits a topological boundary condition can be (re)constructed from data associated to the boundary condition alone. For example, a 3d Turaev-Viro TFT \cite{Turaev2016,TuraevViro1992} can be presented by giving a fusion category of boundary lines. 

Crucially, 2d dg-TFT can be associated to an ($A_\infty$) category of dg-topological boundary conditions, aka D-branes \cite{costello2007topological}. For example, the B-model \cite{witten1988topological,witten1991mirror} with a target space $X$ is associated to the derived category of coherent sheaves on $X$ \cite{kontsevich1995homological}. More general dg-categories have been proposed as descriptions of ``non-commutative'' target spaces or, more precisely, non-geometric 2d dg-TQFTs \cite{Ginzburg:2006fu, Kontsevich2008NotesOA,kontsevich2021pre}. 

This suggests a general strategy: use the 't Hooft expansion of the QFT to build the category of D-branes for the conjectural dual String Theory and use that category to define or at least constrain the corresponding worldsheet theory as a dg-TFT. Such a construction may potentially lead to a universal dictionary for weak-coupling holography. 

As a cautionary note, we should recall that IR divergences plague the integrals which define actual String Theory amplitudes.  In Topological String Theory, it is not uncommon for important contributions to be pushed to the boundary of the integration region \cite{Bershadsky:1993cx,bershadsky1993holomorphic}. A universal combinatorial description of the world-sheet correlation functions and integrands is thus not enough: we also need a universal combinatorial treatment of the IR divergences. In this paper we will focus on the planar ($\hbar \to 0$) limit of the 't Hooft expansion and on the construction of a classical String Theory dual theory, which allows us to mostly ignore this aspect. 

One of the main objectives of this paper is to test 't Hooft-completeness of String Theory in the context of two-dimensional chiral gauge theories, generalizing examples which arise as protected sub-sectors of four-dimensional ${\cal N}=2$ Superconformal Gauge Theories \cite{Beem:2013sza}. 

In particular, we will learn how to build categories of D-branes from the 't Hooft expansion data, which in known twisted holography examples \cite{Costello:2018zrm} match the coherent sheaves in the dual B-model geometries. For general 2d chiral gauge theories, the resulting categories {\it define} novel dual non-commutative three-dimensional Calabi-Yau geometries. 

Schematically, we show that any 2d chiral gauge theory which admits a 't Hooft expansion can always be associated to a ``two-dimensional non-commutative (nc) Calabi Yau cone'' $X_2$. The cone can be promoted to a 3d nc Calabi-Yau geometry in the form of a fibration 
\begin{equation}
    X_3(0) \equiv X_2(-1) \to \mathbb{C}P^1\,.
\end{equation}
A geometric transition gives a family of 3d nc Calabi-Yau geometries $X_3(\lambda)$ depending on the 't Hooft coupling $\lambda$. Intuitively, $X_2$ is a cone over a non-geometric ``space'' $X_2/\mathbb{R}$ akin to $S^3$, and $X_3(\lambda)$ is schematically $AdS_3 \times X_2/\mathbb{R}$. 

We claim that the 't Hooft expansion of the 2d chiral gauge theory matches Topological String Theory with target $X_3(\lambda)$. The nc-CY geometries are encoded in 3d Calabi-Yau dg-categories with a direct QFT definition. 

\begin{figure}[h]
    \centering
\includegraphics[width=0.7\linewidth]{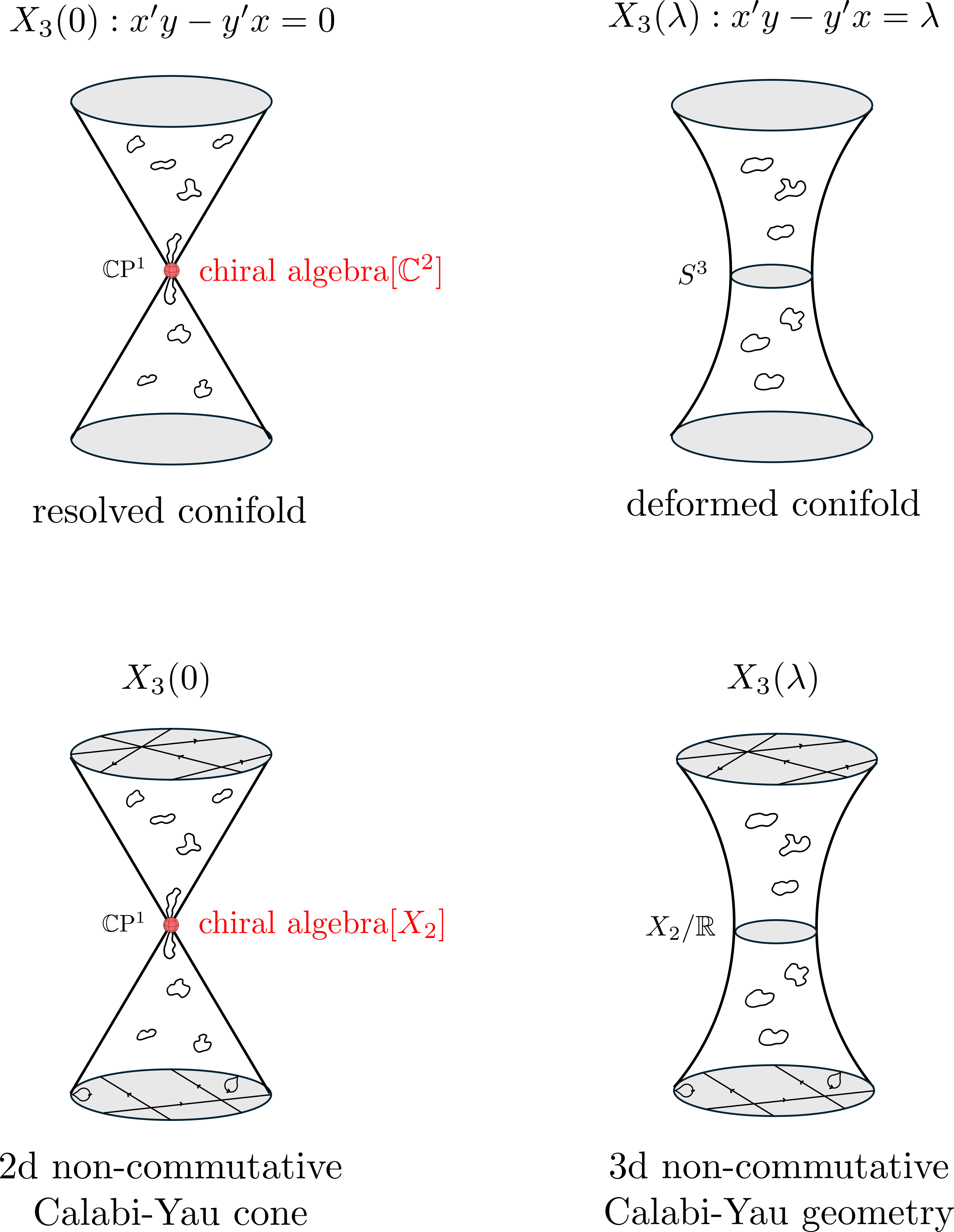}
    \caption{Top row: canonical twisted holography example, connecting a 2d gauged $\beta \gamma$ system supported on a stack of D-branes and B-model topological strings on $SL(2,\bC)$ \cite{Costello:2018zrm}. Bottom row: generalization studied in this paper, connecting general 2d chiral algebras supported on branes in a 2d nc-CY cone target space and topological strings on a 3d nc-CY geometry.}
    \label{fig:enter-label}
\end{figure}

\subsection{Homological Algebra and Beyond}

A general principle of String Theory is that the classical equations of motion of the theory controls the deformations of the world-sheet theory (for closed strings) and of its boundary conditions (for open strings). 
In particular, the closed string fields are identified as worldsheet couplings and open string fields as (possibly matrix-valued) boundary couplings. See \cite{Mazel:2024alu} for a recent related discussion.

In the BRST/BV formalism for the world-sheet theory, the string-theory equations of motion can be concisely formulated as the cancellation of the BRST anomalies created by formal deformations of the theory \cite{Zwiebach:1992ie} or boundary conditions \cite{WITTEN1986253,Witten:1992fb}. This formulation naturally involves the language of Homological Algebra. In particular, the BRST anomalies which arise from matrix-valued deformations of boundary conditions are encoded in an $A_\infty$ category of boundary conditions \cite{Gaiotto:2024gii}.

We will use BRST anomalies as a guiding principle to analyze the 't Hooft expansion of the QFT. We can consider a deformation of the QFT by single-trace operators and of D-brane-like objects by mesonic operators. At the planar level, the cancellation of BRST anomalies produced by the deformations can be interpreted as the equations of motions of the tentative dual String Theory. In particular, this gives us a candidate $A_\infty$ category of boundary conditions for the dual world-sheet theory.

We can describe briefly the ``planar tree-level'' $\lambda \to 0$ limit of the general analysis. A simple example is the computation of the space of gauge-invariant single trace local operators in the QFT. Concretely, a single-trace operator is some linear combination of terms of the form 
\begin{equation}
   \frac{1}{\hbar} \Tr \phi^{i_1} \cdots \phi^{i_n}\,,
\end{equation}
where $\phi^i$ denotes a collection of ``letters'' which can enter in the trace: fields in two-index representations of the large $N$ gauge group and their derivatives.  

The planar BRST differential acts at the leading order in $\lambda$ by replacing a letter by a sequence of letters:
\begin{equation}
    Q: \qquad \phi^i \to Q^i_j \phi^j + Q^i_{j_1 j_2} \phi^{j_1} \phi^{j_2} + \cdots
\end{equation}
Essentially by definition, a nilpotent transformation of this kind equips a  dual space $V$ of symbols $v_i$ with the structure of an $A_\infty$ algebra.

If the QFT arises as the world-volume theory of some D-branes $B$ in a formal $\lambda \to 0$ limit, $V$ controls boundary local operators on $B$ and the BRST anomalies of matrix-valued deformations \begin{equation}
    \int_\partial \phi^i O^\partial_i \, ,
\end{equation}
of the direct sum of multiple ($N$) copies of $B$ by boundary local operators $O^\partial_i$. 

In this language, the BRST cohomology of the space of single-trace operators is recognized as the {\it cyclic cohomology} $\HC^{\bullet}(V)$ of $V$. This mathematical structure occurs naturally in the study of dg-TFTs, precisely as a tool to probe the coupling of closed string states to a dg-topological boundary condition via a disk worldsheet \cite{Kapustin:2004df}. As a consequence, the $\lambda \to 0$ limit of the holographic dictionary between single-trace operators and closed string states in a dg-TFT description is universal and essentially tautological. 

Symmetries of the large $N$ QFT can also be described at the leading planar level as 
transformations 
\begin{equation}
    L: \qquad \phi^i \to L^i_j \phi^j + L^i_{j_1 j_2} \phi^{j_1} \phi^{j_2} + \cdots
\end{equation}
which commute with $Q$. Parsing through the definition, one can make contact with the {\it Hochschild cohomology} $\HH^{\bullet}(V)$, which also describes the symmetries of the world-sheet theory in a dg-TFT language. The $\lambda \to 0$ limit of the holographic dictionary for symmetries (and their action on operators) is thus also universal and tautological. 

A simple and powerful slogan is that any gauge theory with a single-trace action can be interpreted classically as the world-volume theory of $N$ D-branes $B$ in a formal String Theory background with a dg-TFT description for the world-sheet theory.

Beyond the leading order in the planar expansion, both the action of $Q$ and the action of symmetries on single-trace operators are deformed by terms which act on multiple consecutive letters at the time. These terms do not fit into a standard Homological Algebra dictionary.

Nevertheless, the planar corrections can be systematically organized in the form of a deformation $\HC_\lambda^{\bullet}(V)$ and $\HH_\lambda^{\bullet}(V)$ of the complexes defining $\HC^{\bullet}(V)$ and $\HH^{\bullet}(V)$, leading to a formal definition of the space of closed string states and of symmetries of a $\lambda$-dependent dg-TFT. 

In order to produce a standard dg-TFT description of the deformed world-sheet theory, we need to express $\HC_\lambda^{\bullet}(V)$ and $\HH_\lambda^{\bullet}(V)$ as standard cyclic and Hochschild cohomology for some other algebraic object. We will do so with the help of D-branes. 

\subsection{A fundamental enhancement}
Any large $N$ QFT can be modified in many different ways by adding some collection of fields transforming in vector representations of the large $N$ gauge group. We will consider many types of D-brane probes, some arising from fundamental fields defined on the whole QFT space-time and some supported at defects.  All of these modifications are expected to lead to the same closed String Theory dual, modified by some appropriate collection of probe D-branes.

In the $\lambda \to 0$ limit, certain calculations organize themselves in a natural way as dg-TFT data. The first example is the computation of gauge-invariant linear combinations of mesonic operators:
\begin{equation}\label{eq:mes_op}
   \alpha^a \phi^{i_1} \cdots \phi^{i_n} \beta^b
\end{equation}
where the $\alpha^a$ and $\beta^b$ letters denote fields with a single gauge index and their derivatives. 

The planar BRST differential acts at the leading order in $\lambda$ by replacing a letter by a sequence of letters. The action on $\phi^i$ is unchanged, but we have
\begin{align}\label{eq:diff_mes}
    &Q: \qquad \alpha^a \to Q^a_{a'} \alpha^{a'} + Q^a_{a' j} \alpha^{a'} \phi^{j} + \cdots \cr
    &Q: \qquad \beta^b \to Q^b_{b'} \beta^{b'}  + Q^b_{j b'} \phi^{j} \beta^{b'}  + \cdots
\end{align}
Essentially by definition, a nilpotent transformation of this kind equips dual spaces $U$ and $W$ with the structure of a right- or a left- $A_\infty$ module for $V$. 

This is the same mathematical structure as would arise in a dg-TFT to describe the interaction of $B$ with a probe brane $P_0$: $U = \mathrm{Hom}(B,P_0)$ describes junctions from $B$ to $P_0$ and $W = \mathrm{Hom}(P_0,B)$ describes junctions from $P_0$ to $B$. The space of mesonic operators turns out to be dual to the derived tensor product $U \otimes_V W$. We can check immediately that the space of operators \ref{eq:mes_op}, subject to the BRST differential \ref{eq:diff_mes}, reproduces the standard bar construction of the derived tensor product $U \otimes_V W$.

As a consequence, the $\lambda \to 0$ limit of the holographic dictionary between mesonic operators and open string states in a dg-TFT description is universal and tautological. Symmetries of open strings are also nicely reproduced as endomorphisms of $U$ and $W$ as $A_\infty$ $V$-modules.

As we vary and combine different choices of D-brane-like modifications, 
the resulting data can be assembled into a dg-category of $V$-modules which describes the $\lambda \to 0$ limit of the category of branes of the conjectural dual String Theory. In other words, any consistent recipe to add (anti)fundamental fields via mesonic terms in the gauge-theory action can be classically understood as adding formal probe branes in the formal dg-TFT description. 

Beyond the leading order, the planar corrections to these calculations can be systematically organized as a deformation $\mathrm{Mod}_\lambda(V)$ of the category of $V$-modules, leading to a formal definition of the category of branes for a $\lambda$-dependent, back-reacted dual world-sheet theory. In particular, one finds that $\HH_\lambda^{\bullet}(V)$ acts as a symmetry of $\mathrm{Mod}_\lambda(V)$ and $\HC_\lambda^{\bullet}(V)$ couples to it. 

This category gives the desired dg-TFT description of the holographic dual String Theory, at least classically. Once the planar expansion has been organized in this categorical language, we arrive at a weakly-coupled holographic statement which involves two mathematically well-defined entities: the planar expansion of the QFT and the classical String Theory described by the back-reacted category $\mathrm{Mod}_\lambda(V)$ of probe D-branes. 

We expect that the identification between the two sides 
may be formulated as a rigorous theorem of broad applicability. 
Once the classical holographic statement has been proven, 
one may tackle the greater challenge of matching the full $\hbar$ expansion. We will only briefly discuss the structure of that problem.

\subsection{Summary of main results}

For families of two dimensional chiral gauge theories in the planar limit, we construct the following data of the putative worldsheet dual:
\begin{itemize}
\item We express the data of a 2d chiral gauge theory in terms of a 2d-cyclic, finite-dimensional, graded associative super-algebra $A$. This formally defines a world-sheet theory with target $X_2$
such that the 2d theory is the world-volume theory of $N$ D-branes supported on $\bC \subset X_2 \times \bC$.
\item We compute the tree-level cohomology of single-trace operators, as well as the ``wedge'' subalgebra of global symmetries on the sphere, in terms of the Hochschild and cyclic cohomology of $A$. This establishes the classical couplings between the original D-branes and closed strings.
\item Using 2d chiral fundamental matter as a probe, we compute the tree-level cohomology of mesons and associated global symmetries in homological algebra terms and describe their planar deformation. We compute the deformation explicitly for a novel non-geometric example of Twisted Holography. These computations give a formal world-sheet description for space-filling probe D-branes in the back-reacted geometry $X_3(\lambda)$. We expect the category of D-branes in $X_3(\lambda)$ to admit a description as a category of modules for the global symmetry algebra of space-filling D-branes.
\item Using determinant operators as a probe, we compute the tree-level cohomology of ``determinant modifications'' and describe their planar deformation. The action of the above global symmetry algebras on determinant modifications provides a description of the dual D-branes as modules. We compute the action explicitly in the non-geometric example.
\end{itemize}


An important assumption we make about string theory is its 't Hooft completeness, introduced at the beginning of the introduction. In other words, we do not know a priori whether the 't Hooft expansion of a chiral gauge theory is fully captured by a dual string theory; rather, we take this as a working assumption and demonstrate how the algebraic/categorical data of a 2d dg-TFT can be reconstructed.

Finally, we emphasize that we do not provide the dual worldsheet theory data in the ``conventional" sense; we give no lagrangian for the dual worldsheet theory, much less a 2d CFT. Instead, we hope that by showcasing how dg TFT data captures universal features of the dual worldsheet in the form of algebraic dg-TFT data, we inspire a more general, albeit more abstract \cite{Segal1988,costello2007topological,lurie2008classification,Kontsevich2008NotesOA,kontsevich2021pre,kontsevich2023smooth}, means of defining the dual String Theory for theories with a 't Hooft expansion. 



\subsection{Structure of the paper}
Section \ref{sec:example} reviews the standard example of Twisted Holography for chiral algebras, a protected subsector of ${\cal N}=4$ SYM. Section \ref{sec:homB} reviews the Homological Algebra presentation of the B-model with flat target space. Section \ref{sec:hom} reviews and extends Homological Algebra calculations in the $\lambda \to 0$ limit. Section \ref{sec:HAcalculations} introduces a general class of chiral gauge theories admitting a 't Hooft expansion and associates them to 2d Calabi-Yau algebras and dg-TFTs. Section \ref{sec:closed} introduces algebraic structures associated to single-trace operators. Section \ref{sec:general_GCA} reviews and improves the notion of Global Symmetry Algebra of a chiral algebra. Section \ref{sec:generla_flavor} adds fundamental matter and algebraic structures associated to mesonic operators, including the analogue of the algebra of holomorphic functions on $X_3[0]$. Section \ref{sec:det} studies determinant-like operators and associated D-branes in an algebraic language. Section \ref{sec:nctree} introduces the simplest non-commutative example. Section \ref{sec:back} discusses in greater detail the structure of planar corrections. Section \ref{sec:planar_symmetry_algebra}
discusses explicit planar corrections various algebraic structures, including the algebra of holomorphic functions on $X_3[\lambda]$. Section \ref{sec:ncplanar} computes planar corrections and describes $X_3[\lambda]$ in the simplest non-commutative example. Section \ref{sec:conclude} reviews our conclusions and presents various open questions.

\section{A rich example} \label{sec:example}
In this section, we will focus on the holographic duality studied in \cite{Costello:2018zrm}, relating:
\begin{itemize}
    \item The protected chiral algebra subsector of four-dimensional ${\cal N}=4$ $SU(N)$ Supersymmetric Yang Mills theory \cite{Beem:2013sza}, aka supersymmetric chiral $SU(N)$ gauge theory. 
    \item The B-model topological string theory/BCOV theory \cite{Bershadsky:1993cx,Costello:2012cy,Costello:2015xsa} with target $SL(2,\bC)$. This theory is also conjectured to be a twist of the type IIB supergravity on $AdS_5\times S^5$ \cite{Costello:2016mgj}.
\end{itemize}
We will review and extend the known holographic dictionary, including
\begin{itemize}
    \item The match of single-trace operators and closed strings states, as well as the associated ``global symmetry algebra'' $\fL_\lambda$ \cite{Costello:2018zrm}.
    \item The match of mesonic operators and open strings states associated to space-filling branes, as well as the associated ``global symmetry algebra'' $\fP_\lambda$ \cite{Costello:2018zrm}. This algebra can also be identified with the algebra of functions on the backreacted geometry, thereby reproducing the category of D-branes as the category of $\fP_\lambda$-modules.
    \item The match of planar determinant correlation functions and ``giant graviton'' D-branes \cite{Budzik_2023}. 
\end{itemize}
In each case, we will illustrate how algebraic aspects of the string worldsheet theory emerge from planar calculations in the chiral algebra. A summary of the relationship will be given in Section \ref{sec:ex_conclusion}. The strategy will naturally generalize to known and novel examples of chiral algebras which admit a large $N$ expansion. We will introduce our general construction of chiral algebra based on Calabi-Yau algebra/category in Section \ref{sec:HAcalculations}. The construction of the global symmetry algebra $\fL_\lambda$ will be generalized in Section \ref{sec:general_GCA}. We will also discuss the global symmetry algebra $\fP_\lambda$ of mesonic operators in this generalized framework in Section \ref{sec:generla_flavor}, and the determinant modification in Section \ref{sec:det}.

This Section is designed to give a review of known results \cite{Costello:2018zrm,Budzik_2023} in the standard example of Twisted Holography. The presentation anticipates some of the ideas and naming conventions introduced in the rest of the paper and in particular the homological algebra framework reviewed at greater length in Section \ref{sec:closed}. At the very end of the Section, we re-derive the ``saddle equations'' controlling the large $N$ limit of determinant operators as a BRST anomaly cancellation and as the definition of a module for the global symmetry algebra of space-filling D-branes.  

\subsection{The chiral algebra}
The chiral algebra subsector of 4d SYM can be presented as a chiral 2d gauge theory, with action
\begin{equation}
    \frac{1}{\hbar} \int d^2 z \, \Tr X (\bar \partial Y+ [a_{\bar z},Y]) 
\end{equation}
involving two adjoint bosonic fields $X$, $Y$ of scaling dimension $\frac12$ and the $(0,1)$ component 
${a}_{\bar z}$  of a 2d gauge connection. The overall factor of $\hbar^{-1}$ will be helpful in setting up the 't Hooft expansion. 

The $a_{\bar z}$ connection can be gauge-fixed to $0$, at least locally. The gauge-fixing introduces a set of $bc$ ghosts which also transform in the adjoint representation of the gauge group. 

In the following we will work with $U(N)$-valued fields for notational simplicity. The difference between $SU(N)$ and $U(N)$ gauge theories is the presence of a decoupled free $U(1)$ factor in the latter. It will not affect planar calculations. Keeping track of gauge indices, the OPE of all the fields is
\begin{align}
    Y^i_j(z) X^k_t(w) \,&\sim\, \hbar \frac{\delta^i_t \delta^k_j}{z-w} \cr
    b^i_j(z) c^k_t(w) \,&\sim\, \hbar \frac{\delta^i_t \delta^k_j}{z-w} \, .
\end{align}
We will typically leave gauge indices implicit and use matrix notation for the fields. 

The BRST differential is the zero mode of the BRST current:
\begin{equation}\label{eq:cano_BRST}
   Q = \frac{1}{\hbar} \oint \frac{dz}{2 \pi i} \Tr \left(\frac12 b [c,c] + c[X,Y]\right)\,.
\end{equation}
Up to a small subtlety concerning the removal of ghost zero modes,\footnote{The subtlety is that we should take a {\it relative} BRST cohomology: $c$ should only appear through its derivatives and $U(N)$ invariance should be imposed by hand. Doing otherwise adds to the cohomology some spurious elements built from $c$ only. In correlation functions, one saturates $c$ zero modes by hand.} the local operators of the 2d gauge theory are defined as elements in the BRST cohomology $\EuScript{A}_N$ of the 2d free chiral algebra defined by the above OPE. 

The action of the BRST differential on local operators is the sum of two terms involving respectively one or two Wick contractions:
\begin{equation}
    Q = Q_0 + \hbar\, Q_1\,.
\end{equation}
We refer to these as the tree-level and 1-loop parts of the BRST differential. Notice that $Q_0^2=0$, $\{Q_0, Q_1\} =0$ and $Q_1^2=0$. 

The space $\EuScript{A}_N$ of local operators is also a chiral algebra. The OPE is computed by free OPE of cohomology representatives, up to shifts by BRST-exact operators. Alternatively, we can first compute unambiguous sphere correlation functions of cohomology representatives and then derive the OPE from these. 
The sphere correlation functions and OPEs will be the main observable of interest in this paper.\footnote{It is also possible to define and compute $\EuScript{A}_N$ conformal blocks on general Riemann surfaces. An important subtlety is that $a_{\bar z}$ can only be gauge-fixed to an holomorphic bundle and conformal blocks will involve integrals over the space of holomorphic bundles. The details of the large $N$ expansion and the holographic duality will be more complicated.} 

We should recall some special properties of this model which simplify our analysis but will not generalize to other examples. In particular, the conformal symmetry of the model is enhanced to a ``small'' ${\cal N}=4$ super-conformal algebra\footnote{The chiral algebra defined by the BRST cohomology of local operators 
also has an $SL(2)$ global symmetry which acts on cohomology representatives roughly as $b\leftrightarrow \partial c$. It will not play a role in this paper.} with BRST-closed generators which include
\begin{itemize}
    \item The stress tensor
\begin{equation}\label{eq:ad_stress_tensor_original_theory}
    T \equiv \frac{1}{2\hbar} \Tr\left(-2 b \partial c - X\partial Y + Y \partial X \right)
\end{equation}
with central charge $- 3 N^2$.
\item Level $-N^2$ Kac-Moody currents for an $SU(2)_R$ symmetry transforming $X$ and $Y$ as a doublet:
\begin{equation}
    J^{++} \equiv \frac{1}{2\hbar}\, \Tr X^2 \qquad \qquad J^{0} \equiv  \frac{1}{2\hbar}\, \Tr X Y \qquad \qquad J^{--} \equiv \frac{1}{2\hbar} \Tr \,Y^2\,.
\end{equation}
\item Four super-currents 
\begin{align}
    G^+ &\equiv \frac{1}{2\hbar}\, \Tr X b \qquad \qquad  \widetilde G^+ \equiv \frac{1}{2\hbar}\, \Tr X \partial c  \cr
    G^- &\equiv\frac{1}{2\hbar}\, \Tr Y b \qquad \qquad \widetilde G^- \equiv\frac{1}{2\hbar}\, \Tr Y \partial c  \,.
\end{align}
\end{itemize}

Twisted Holography concerns the 't Hooft expansion of the $SU(N)$ gauge theory correlation functions: we trade the rank $N$ for a 't Hooft coupling $\lambda \equiv \hbar N$ and expand in powers of $\hbar$ at fixed $\lambda$.\footnote{Therefore, we have a family of $\hbar$-adic vertex algebras \cite{li2004vertex}.} The expansion in powers of $\hbar$ is expected to match the genus expansion of the dual string theory, while $\lambda$ controls the period of the holomorphic three-form in the $SL(2,\bC)$ target space. We will mostly focus on the planar limit, dual to tree-level string theory calculations. 

The identification of the dual theory as a B-model with $SL(2,\bC)$ target space can be justified by realizing the chiral algebra as the world-volume theory of $N$ D-branes wrapping $\bC \subset \bC^3$ and computing the branes back-reaction in the BCOV description of the topological string theory. One of our objectives is to describe this back-reaction in an algebraic manner, which can be generalized to non-geometric situations.

\subsection{Single-trace Local Operators and the Planar Limit}
Operators built from less than $N$ fields can be organized into polynomials of single-trace operators. A sphere correlation function of such operators is computed as a sum over double-line free Feynman diagrams which may have multiple connected components, each with a specific genus. 

If we normalize traces by an $\hbar^{-1}$ prefactor, as we saw e.g. in the stress tensor \eqref{eq:ad_stress_tensor_original_theory}, and express everythign in terms of the 't Hooft coupling $\lambda$ and $\hbar$ by substuting $N = \frac \lambda \hbar$, we recover the standard String Theory-like form of the 't Hooft expansion: {\it connected} sphere correlation functions can be expanded as 
\begin{equation}
    f(\hbar,N) = \frac{1}{\hbar^2} f_0(\lambda)+ f_1(\lambda) + \hbar^2 f_2(\lambda) + \cdots
\end{equation}
The first term is the ``planar'' part of the connected correlation function.\footnote{Notice that correlation functions of specific single-trace operators are polynomials in $N$ and thus in $\lambda$. They are also Laurent polynomials in $\hbar$ and thus the genus expansion truncates.} 


The planar part of a correlation function is a sum of terms with different numbers of connected components, which appear with a different overall power of $\hbar$. This can be a source of confusion: the planar part of a correlation function is not simply defined as the leading contribution in the 't Hooft expansion! We will see in detail how this affects various 2d CFT structures. 

In a dual String Theory, single-trace operators should be associated to vertex operator insertions on the world-sheet, representing specific closed string states approaching specific locations in the holographic boundary. The planar data should match a tree-level String Theory calculation and non-planar effects should arise from loops.
In the case at hand, the dual calculations could in principle be done via Witten diagram in the BCOV theory on $SL(2,\bC)$ with specific bulk-to-boundary propagators. Such a computational strategy, though, does not obviously generalize to non-geometric settings.\footnote{The reference \cite{Zeng:2023qqp} approaches the problem via a KK reduction of BCOV to a three-dimensional Holomorphic-Topological theory. It may be possible to study the non-geometric backgrounds we are interested in by some sort of KK reduction from a non-commutative Calabi-Yau $Y_3[\lambda]$ to a geometric 3d HT theory. We leave this approach for future work.}

There are various ``natural'' normalizations one may choose for single- and multi-trace local operators:
\begin{itemize}
    \item In a CFT, it is natural to normalize operators so that two-point functions are of order $1$ (or better, functions of $\lambda$). Polynomials in single-trace operators with no extra power of $\hbar$ have such property:
    \begin{equation}
        O^{\mathrm{CFT}} \equiv \Tr (\cdots) \, ,
    \end{equation}
    where the ellipsis denotes a cyclic sequence of (derivatives of) adjoint fields. 
    With this normalization, sphere correlation functions are finite and dominated by disconnected products of two-point functions of single-traces in the $\hbar \to 0$ limit, as in a generalized free field theory. 
    \item The structure of the OPE is best described in a ``classical'' normalization 
    \begin{equation}
        O^\cl\equiv \hbar \Tr (\cdots) \, ,
    \end{equation}
    so that planar singular part of the OPE is of order $\hbar^2$ and the regular part of order $1$:
    \begin{align}
        O^\cl_i(z)  O^\cl_j(0) &\sim \left[O^\cl_i(z)O^\cl_j(w)\right] + \hbar^2 \left[ g_{ij}(\lambda)z^{\cdots}+ \sum c_{ij}^{k,n}(\lambda)z^{\cdots} \partial^n O^\cl_k(w)\right. \\ & \left.+ \sum c_{ij}^{k_1,k_2,n_1,n_2}(\lambda)z^{\cdots} \partial^{n_1} O^\cl_{k_1}(w)\partial^{n_2} O^\cl_{k_2}(w)\right] \nonumber +O(\hbar^4) +  \cdots
    \end{align}
    This makes manifest that the planar limit of the OPE is captured by a (non-linear) Poisson chiral algebra $\EuScript{A}_\infty[\lambda]$, i.e. a commutative algebra with a compatible derivative and associative\footnote{The multi-trace terms in the Poisson chiral algebra OPE are important for associativity: the Jacobi identity of three single-trace operators receives non-trivial contributions from the combination of a central term and a double-trace term in the OPE.} $\lambda$-bracket\footnote{We apologize for denoting the 't Hooft coupling as $\lambda$ in a situation where one may want to reserve the symbol for the notion of $\lambda$-bracket.}. This structure arises from the Fourier transform of the singular part of the OPE (see, e.g., \cite{kac2017introduction}). Obvious (and deceptively simple) examples are the OPE of the rescaled superconformal generators. E.g. 
    \begin{equation}
        T^\cl(z) \, T^\cl(w) \sim \hbar^2\left[ -\frac{3 \lambda^2}{2(z-w)^4} +  \frac{2 T^\cl(w)}{(z-w)^2} + \frac{\partial T^\cl(w)}{z-w} \right]+ (T^\cl T^\cl)(w) + \cdots \, ,
    \end{equation}
    is the Virasoro Poisson chiral algebra. This should be matched with an OPE calculation in the tree-level BCOV theory in $SL(2,\bC)$, with an appropriate holographic dictionary \cite{Zeng:2023qqp}.\footnote{From this perspective, the $\hbar$ expansion of $\EuScript{A}_N$ is thus a ``deformation quantization'' \cite{beilinson2004chiral} of the planar Poisson chiral algebra.  If we could prove the planar limit of the holographic correspondence, extending the proof to the full 't Hooft/genus expansion would entail comparing two deformation quantizations of the same Poisson VOA. The 3d HT theory mentioned in the previous footnote is also a natural tool to study deformation quantization of Poisson VOA \cite{tamarkin2003deformations,beilinson2004chiral}} 
    \item In this paper, we will stick to the normalization  
     \begin{equation}
        O \equiv \frac{1}{\hbar} \Tr (\cdots) \, ,
    \end{equation}
    which turns out to be most suitable for the formal deformation theory considerations we employ to study the world-sheet theory for the dual String Theory.
\end{itemize}

\subsection{Single-trace BRST cohomology}

The notion of single-trace operator has a subtle interplay with the BRST cohomology. The tree level part $Q_0$ of the BRST differential involves a single Wick contraction and reproduces classical BRST transformations:
\begin{align}
    Q_0\,c &= \frac12 [c,c] \cr
    Q_0\,X &= [c,X] \cr
    Q_0\,Y &= [c,Y] \cr
    Q_0\,b &= [c,b] + [X,Y]\,.
\end{align}
In particular, it maps a single-trace operator to a single-trace operator. It acts as a derivation on products of single-trace operators. 

The 1-loop part $\hbar Q_1$ of the BRST differential involves two Wick contractions and is a bit more complicated. When acting on a single-trace operator, the result includes:
\begin{itemize}
    \item Single-trace which arise from the contraction of consecutive symbols. They scale as $\lambda$. We refer to this as the linear part $\lambda Q^l_1$ of the planar answer.
    \item Double-trace terms which arise from the contraction of non-consecutive symbols. These terms are still planar and contribute, say, to the definition of $\EuScript{A}_\infty[\lambda]$.  
\end{itemize}
The linearized operator $Q_0 + \lambda Q_1^l$ is nilpotent and defines a complex $\mathrm {Ops}_\lambda$ of single-trace operators which should match (in the sense of complexes) the space of vertex operators in the dual world-sheet theory. 

In this particular chiral gauge theory, the analysis is simplified by the existence of a large collection of single-trace operators which are exactly BRST-closed. It is easy to see that $\Tr X^n$ is BRST-closed: it is classically invariant and it does not admit two Wick contractions with the BRST current. The $SU(2)_R$ symmetry then implies that the whole ``A'' tower of symmetrized traces
\begin{equation}
    {\cal A}_{a,b} \equiv \frac{1}{(a+b)\hbar} \,\STr X^a Y^{b} 
\end{equation}
is BRST-closed. Acting with super-conformal generators produces three more towers:
\begin{align}\label{eq:BCD_twoer}
    {\cal B}_{a,b} &\equiv \frac{1}{\hbar} \,\STr X^a Y^{b} b\cr
    {\cal C}_{a,b} &\equiv \frac{1}{\hbar} \,\STr X^a Y^{b} \partial c\cr
    {\cal D}_{a,b} &\equiv \frac{1}{\hbar} \,\STr X^a Y^{b} b \partial c + \cdots
\end{align}
The operator ${\cal D}_{0,0}$ is proportional to the stress tensor.

The full expression of the last tower is complicated, as the OPE with super-conformal generators which produces it includes terms with two Wick contractions, producing both single-trace corrections proportional to $\lambda$ and double-trace corrections. The planar single-trace parts define the corresponding element in the cohomology of $\mathrm {Ops}_\lambda$.

So far we have only considered the single-trace part of the BRST cohomology. As the BRST charge is a derivation of the OPE, a simple strategy to define BRST-closed multi-trace operators is to look at the regular part of an OPE of BRST-closed single-trace operators. The regular part of OPE is usually called the regularized product. The price to pay is that such regularized product is not associative in the standard sense. Hence the whole BRST cohomology can not be simply identified as polynomial functions of the single-trace cohomology. However, we still expect the BRST cohomology of local operators built from a number of fields which remains finite as $N$ is increased to consist of regularized products of operators from the single-trace cohomology, namely the above four towers. This expectation has not been systematically tested and will not be needed in the following. 

As $Q_0$ acts on individual traces in a product as a derivation, we can compute the $Q_0$ cohomology on the space of single-trace operators. This calculation is best done with the Homological Algebra tools reviewed in Section \ref{sec:HAcalculations}. The result is precisely given by the  four towers above. Though the results given above are specific to the chiral gauge theory we considered, we will see in Section \ref{sec:closed} that a similar pattern holds for the single-trace cohomology in more general chiral gauge theories.

\subsection{The Global Symmetry Algebra of sphere correlation functions}
\label{sec:GCA_standard}
Recall that any chiral algebra is associated to a Vertex Operator Algebra of 
Fourier modes. We will use a ``math'' labelling of the modes, 
\begin{equation}
	O_{n} \equiv \oint_{|z|=1} \frac{dz}{2 \pi i} z^n \, \cO(z) \, .
\end{equation}
As an exception to this notation, we still denote the global conformal generators as $L_{-1},L_0,L_1$. The action of non-negative modes on local operators and the commutator of two modes are controlled by the singular part of the OPE. 

Recall that a quasi-primary operator is an operator annihilated by the $L_1$ mode of the stress tensor. We are interested in the modes $O_n$ of quasi-primary single-trace operators $\cO$ which annihilate the vacuum at $0$ and $\infty$, i.e. such that $0\leq n \leq 2 \Delta-2$ if the scaling dimension of $\cO$ is $\Delta$. These modes form an irreducible representation of dimension $2 \Delta -1$ under the global conformal group: 
\begin{align}
	[L_{-1}, O_{n}] &= - n O_{n-1} \cr
	[L_{0}, O_{n}] &= (\Delta -1- n) O_{n} \cr
	[L_{1}, O_{n}] &= (n+2 - 2 \Delta) O_{n+1} \,.
\end{align}
An important special property of these modes is that the action on local operators and their commutation relations do not receive contribution from the central terms in the OPE. For example, the Virasoro central charge does not contribute to the commutator of the global conformal generators. As a consequence, the Jacobi identities for three such modes do not receive contributions from non-linear terms in the planar OPE.

The linear terms in the planar OPE and the linear terms in the action of the BRST differential thus equip the modes with the structure of a dg-Lie algebra, which we denote as the {\it Planar Global Symmetry Algebra}\footnote{The name is a bit of a misnomer, as the non-linearities would still contribute to Ward identities for planar correlation functions. In particular, knowledge of $\fL_\lambda$ alone is not obviously sufficient to determine planar correlation functions or OPE.} $\fL_\lambda$. 

The restriction on the mode number is crucial. The linear part of the commutators of general single-trace modes is not associative and does not define a Lie algebra. Because of the definition, $\fL_\lambda$ is the mode algebra of $\mathrm{Ops}_\lambda$, and acts on it. 

Another useful perspective is that modes $\cI$ in $\fL_\lambda$ can be added formally to the BRST differential to deform the chiral algebra. The deformation creates a BRST anomaly $(Q+\cI)^2$ and we can select the planar single-trace part
\begin{equation}
    \{Q,\cI\} + \frac12 \{\cI,\cI\}\,,
\end{equation}
where the bracket denotes the single-trace planar part of the commutator, e.g. the Lie algebra bracket of $\fL_\lambda$.

A basic ingredient of the Twisted Holography conjecture is that the dg Lie algebra $\fL_\lambda$ is quasi-isomorphic to the Lie algebra of holomorphic, divergence free, polynomial polyvector fields on $SL(2,\bC)$. These are the global symmetries of B-model worldsheet theory and act on the vertex operators which represent $\mathrm{Ops}_\lambda$ insertions. We will generalize this identification to other examples. 

We should be careful to distinguish symmetries of the world-sheet theory and symmetries of the corresponding String Theory. Although the two naively coincide, in practice the symmetries can be broken or deformed in the String Theory due to IR divergences. In the case at hand, the subtleties concern the holographic boundary conditions: a polynomial vector field acting on a field configuration which decays at the boundary may produce components which do not decay at the boundary.  

We expect these subtleties to be analogous to the non-linear corrections to the action of $\fL_\lambda$ on single-trace local operators in the chiral algebra. It would be interesting to explore this point and the analogous statement about $\mathrm{Ops}_\lambda$ further. 

A direct comparison between modes in $\fL_\lambda$ and polyvector fields is a bit laborious, especially for the action of modes from the $\cD$ tower. The comparison is facilitated by the $PSU(2|2)$ global super-conformal symmetry group. The conformal generators and the zero modes of the $SU(2)$ currents map to the vector fields for the left and right action of $SL(2)$ on itself. The modes $G^\pm_{\pm \frac12}$ and $\wt G^\pm_{\pm \frac12}$ of the super-currents map are then identified respectively with the coordinate functions on $SL(2,\bC)$ and with four bi-vectors. 

The modes of the $\cB$ tower transform in representations with the same spin for the two bosonic $SL(2)$'s and are identified with polynomials in $SL(2,\bC)$ up to an overall scale for each spin. The $PSU(2|2)$ action completes the matching to the remaining three towers. We expect the Jacobi identities of the Lie algebra to essentially fix the unknown proportionality constants. We will propose an alternative computational strategy momentarily.

The $\lambda \to 0$ limit $\fL_0$ of the global symmetry algebra is well-defined and gives and polynomial, divergence-free holomorphic poly-vectorfields on the resolved singular conifold geometry \begin{equation}
    \cO(-1) \times \cO(-1) \to \mathbb{C}P^1\,.
\end{equation}
This is the natural ``ambient'' geometry for the original stack of $N$ D-branes, which is deformed to $SL(2,\bC)$ by the back-reaction \cite{Gopakumar:1998ki}.

In a standard patch of $\mathbb{C}P^1$, $\fL_0$ maps to a sub-algebra of the Lie algebra of polynomial, divergence-free holomorphic poly-vectorfields on $\bC \times \bC^2$. The latter can be identified with the linearized algebra $\fL_0^\bC$ of all non-negative single-trace modes\footnote{We emphasis that $\fL_0^\bC$ is not the complexification of $\fL_0$, but rather an extension of $\fL_0$ that includes all modes of the single trace operators. Their distinction is also explained in Appendix \ref{sec:global_sym}}, which is well-defined in the $\lambda \to 0$ as the central terms in the OPE vanish in the limit.\footnote{Instead, the analogous $\fL_\lambda^\bC$ is not associative.}

We will return to this comparison in Section \ref{sec:HAcalculations} with powerful homological algebra methods. 

There is an important subtlety which we should mention here. In the discussion above, we have compared modes in the cohomology of $\fL_\lambda$ to holomorphic poly-vectorfields. We should remember that the cohomology of a dg-Lie algebra is not just a Lie algebra: it also naturally gains higher operations making it into an $L_\infty$ algebra built via {\it Homotopy Transfer} \cite{loday2012algebraic}. A good way to understand this fact is to consider again the quadratic BRST anomaly. 
When we add an element $\cI$ in $\fL_\lambda$ to the BRST differential, the corresponding BRST anomaly is given by 
\begin{equation}
    \{Q,\cI\} + \frac12 \{\cI,\cI\}\,.
\end{equation}
However, when the cohomology is not just a Lie algebra but have non-vanishing higher operations, the BRST anomaly becomes a more complicated expression when we try to describe the deformation via canonical representatives of a cohomology class $[\cI]$:
\begin{equation}
    \frac12 \{[\cI],[\cI]\} + \frac16 \{[\cI],[\cI],[\cI]\}_3 + \cdots
\end{equation}
 Essentially, this happens because $\{[\cI],[\cI]\}=0$ does not imply that the BRST anomaly vanishes. The higher operations would vanish automatically if the cohomology was supported in ghost number $0$, but that is not the case here. 

The full physical identification between $\fL_\lambda$ and an algebra of poly-vectorfields thus requires that the higher operations vanish with appropriate choices of cohomology representatives (a ``formality theorem''). This will be nicely accounted for in the D-brane based proof we discuss now. 

\subsection{Fundamental matter and space-filling branes}
Next, we consider our first D-brane-like modification of the system. We add $k$ bosonic and $k$ fermionic (anti)-fundamental matter fields to the gauge theory and study the 't Hooft expansion of the resulting chiral algebra. Holographically, this modification is expected to add $k$ space-filling ``probe'' D-branes and $k$ space-filling ghost probe D-branes, i.e. a probe D-brane $P$ dressed by an $\bC^{k|k}$ Chan-Paton bundle.

In a BCOV description, these D-branes support a $U(k|k)$ holomorphic Chern-Simons theory coupled to the BCOV fields. 

We can denote the extra chiral algebra fields collectively as $I^A$ and $J_A$, with $A$ running over $k$ bosonic and $k$ fermionic values. We normalize the OPE as 
\begin{equation}
    I^A_i(z) J_B^j(w) \sim \hbar \frac{\delta^j_i \delta^A_B}{z-w}\,.
\end{equation}
Somewhat tediously, one sometimes needs to keep track of the Grassmann parity $(-1)^{|A|}$ of individual components of the fundamental fields,~e.g. 
\begin{equation}
    J_A^i(z) I^B_j(w)  \sim  (-1)^{|A|+1} \hbar\frac{\delta^i_j \delta^B_A}{z-w}\,.
\end{equation}

Operators of ``small'' size compared to $N$ can be written as polynomials in single-trace operators and mesonic operators, i.e. open strings of adjoint symbols sandwiched between an anti-fundamental and a fundamental symbol. We will include a factor of $\hbar^{-1}$ in front of mesons. For example, an important class of mesonic operators takes the form 
\begin{equation}
    \cM^A_{a,0;B} \equiv \frac{1}{\hbar} I^A X^a J_B
\end{equation}
together with their $SU(2)_R$ partners $\cM^A_{a,b;B}$ \footnote{Again, one can consider a CFT normalization, which involves a factor of $\hbar^{-\frac12}$, or a classical normalization with no factor of $\hbar$ which leads to a Poisson chiral algebra of mesons with OPE singularities appearing at order $\hbar$. At this order, single-trace operators enter OPE but are central.} Recall that $X$ and $Y$ transform as a doublet under the $SU(2)_R$ action. An explicit form of $\cM^A_{a,b;B}$ can be obtained by expanding the generating function $\frac{1}{\hbar} I^A (X + u Y)^{a + b} J_B$
and extracting the coefficient of $u^b$.

If we keep $k$ arbitrary, correlation functions containing mesonic operators can be decomposed into ``flavour-ordered'' pieces multiplying products of cyclic combinations of  flavour Kronecker $\delta^A_B$ symbols. 
Flavour-ordered, connected correlation functions are the natural quantities appearing in the 't Hooft expansion, with each cycle of flavour indices corresponding to a boundary of the dual string world-sheet. The mesonic operators correspond to open strings. 

The leading order ``planar'' contribution to a connected correlation function of this form has a single boundary and the topology of a disk. It appears at order $\hbar^{-1}$. 

The BRST charge of the flavoured theory includes an extra mesonic term 
\begin{equation}
    \frac{1}{\hbar} \oint \frac{dz}{2 \pi i} I^A c J_A \,. 
\end{equation}
At tree-level, the action of $Q_0$ on $b$ is thus modified: 
\begin{align}
    Q_0\,b &= [c,b] + [X,Y] + J_A I^A \cr
    Q_0\,I^A &= (-1)^{|A|+1} I^A c \cr
    Q_0\,J_A &= c J^A\,.
\end{align}
In particular, $\cM^A_{a,b;B}$ are $Q_0$-closed, but the trace $\cM^A_{a,b;A}$ is the $Q_0$ image of 
$\cB_{a,b}$. We will see later on that the traceless part of $\cM^A_{a,b;B}$ exhausts the mesonic part of the $Q_0$ cohomology. 

At one loop, $\hbar Q_1$ maps $\cM^A_{a,b;B}$ to $\hbar \delta^A_B \cC_{a,b}$. The introduction of flavour thus lifts both the $\cB$ and $\cC$ towers and the surviving mesons can be thought of as valued in $\mathfrak{psu}_{k|k}$. 

In the following, we will focus on meson operators with $A \neq B$, so that we can ignore the mixing with single-trace operators. The action of $\hbar Q_1$ can be restricted to a linear planar part $\lambda Q_1^l$ which maps mesons to mesons by contracting the BRST current with consecutive symbols. The operator $Q_0 + \lambda Q_1^l$ acting on mesons defines a complex $\mathrm{Ops}^P_\lambda$ of mesonic operators which should match a space of boundary vertex operators in the dual world-sheet theory. 

\subsection{Mesonic GSA}
Following the same strategy as for single-trace operators, we can define the linearized global mode algebra of the mesons. In particular, the zero modes of the $\cM^A_{0,0;B}$ currents generate the global $\mathfrak{u}_{k|k}$ flavour symmetry of the problem. The flavour structure in the meson-meson OPE is such that the algebra takes the form $\mathfrak{u}_{k|k}[\fP_\lambda]$ for a unital algebra $\fP_\lambda$ (neglecting the subtleties about the diagonal mixing with $\fL_\lambda$), 
with unit given by the  $\cM^A_{0,0;B}$ zero modes.

More precisely, we should define $\fP_\lambda$ as a dg-algebra, using the modes of all mesonic operators, and then pass to cohomology as an $A_\infty$ algebra via Homotopy Transfer. Assuming that the cohomology consists of the modes of $\cM^A_{a,b;B}$, though, the higher operations vanish automatically because the cohomology is supported in ghost number $0$. This allows us to treat $\fP_\lambda$ as an algebra. 

Another entry of the Twisted Holography dictionary identifies $\fP_\lambda$ as the algebra of holomorphic functions on $SL(2,\bC)$, i.e. with the classical algebra of global gauge transformations in the hCS theory which fixes the background connection. In the $\bar A=0$ background, these are $\delta_\alpha \bar A = \bar \partial \alpha + [\bar A, \alpha] = \bar \partial \alpha = 0$. That is, $\alpha$ is holomorphic. The Lie algebra of symmetries of hCS is then $\mathfrak{u}_{k|k}[\mathcal{O}(SL_2(\bC)]$, with $\mathcal{O}(SL_2(\bC))$ the unital algebra of holomorphic functions on $SL_2(\bC)$. 

As for the single-trace symmetries, we identify $\mathfrak{u}_{k|k}[\fP_\lambda]$ as the algebra of boundary symmetries of the world-sheet theory. In the hCS theory, subtleties in the holographic dictionary will introduce non-linearities which should match the non-linear planar terms in the mode algebra. 

The explicit match is a bit easier to test than the analogous one for $\fL_\lambda$. The mesons $\cM^A_{1,0;B}$ and $\cM^A_{0,1;B}$ each contribute two modes. We get a total of four elements in $\fP_\lambda$, to be identified with the coordinate functions on $SL(2,\bC)$. We can denote them as $(x,x')$ and $(y,y')$ respectively. 

A straightforward tree-level calculation shows that these generators commute and $x' y - y' x=0$ in $\fP_0$. A 1-loop correction deforms that relation to the generating relation for $SL(2,\bC)$: $x' y - y' x=\lambda$ \cite{Costello:2018zrm}. Composing these generators, specific polynomial holomorphic functions on $SL(2,\bC)$ can be associated to modes in $\fP_\lambda$.

Of course, the $SU(2)_R$ symmetry and $SL(2)$ conformal symmetry fix most of the identification up to overall coefficients: if we denote as $z_\alpha = (x,y)$ and $z'_\alpha = (x',y')$ the two $SU(2)$ doublets of generators, the symmetrization of $SU(2)_R$ indices in 
\begin{equation}
    z_{\alpha_1}\cdots z_{\alpha_{n-m}} z'_{\alpha_{n-m+1}}\cdots z'_{\alpha_{n}}
\end{equation}
gives the $m$-th mode of the meson of $SU(2)_R$ spin $n/2$.

One should note that the back-reacted geometry $SL(2,\bC)$ is an affine variety. We can therefore identify the category of coherent sheaves on $SL(2,\bC)$ with the category of modules over the algebra of polynomial holomorphic functions on $SL(2,\bC)$. Thus in this example, we are able to reconstruct the D-branes category by studying the mesonic global symmetry algebra $\fP_{\lambda}$, realized simply as the category of modules over $\fP_{\lambda}$. 

\subsection{Open vs Closed}
We can probe the relation between $\fL_\lambda$ and $\fP_\lambda$ by looking at the interplay between single-trace operators and mesons. Operators with different ghost number play a slightly different role here, though ultimately all statements here can be formulated uniformly by looking at the BRST anomalies which occur when adding to the BRST charge modes in both algebras.

In the string theory dual, turning on a closed string mode which is a function creates a BRST anomaly on space-filling branes. The anomaly is the restriction of the function to the branes. In the QFT, this is dual to the observation that $\cB_{a,b}$ is not closed anymore and rather maps to 
a diagonal meson. Accordingly, the action of $Q$ maps modes of $\fL_\lambda$ 
associated to functions to modes of $\fP_\lambda$ associated to the same functions. 

Closed string modes which are vectorfields act as symmetries on the 
open string modes by the standard action on functions. Accordingly, the linearized action of modes of the $\cA$ and $\cD$ towers map mesons to mesons and give an action of the ghost number $0$ part of $\fL_\lambda$ as a derivation of $\fP_\lambda$. 

Closed string modes which are bivectors can be turned on infinitesimally to give a non-commutative deformation of the gauge algebra on space-filling branes. In the QFT, we can imagine adding a mode of the $\cC$ tower to the BRST differential. This means that $[X,Y]$ is not exact but rather equals some polynomial in $X$ and $Y$. When we compute the planar OPE of two mesons, we generically produce an expression which is not symmetrized in $X$ and $Y$. Normally, it would be symmetrized by adding BRST-exact terms. If the BRST charge is deformed, the product in $\fP_\lambda$ will be deformed accordingly. 

These three cases can be unified by looking at the ``Hochschild cohomology'' $\HH^{\bullet}(\fP_\lambda,\fP_\lambda)$ of $\fP_\lambda$ , together with a map $\fL_\lambda \to \HH^{\bullet}(\fP_\lambda,\fP_\lambda)$. For exmaple, $H^{0}(\fP_\lambda,\fP_\lambda)$ is a copy of $\fP_\lambda$ corresponding to BRST anomalies, $H^{1}(\fP_\lambda,\fP_\lambda)$ corresponds to derivations of $\fP_\lambda$ and $H^{2}(\fP_\lambda,\fP_\lambda)$ corresponds to non-commutative deformations of $\fP_\lambda$. More generally, the entire $\HH^{\bullet}(\fP_\lambda,\fP_\lambda)$ can be interpreted as controlling the $A_\infty$ deformation of $\fP_\lambda$. General considerations about the structure of BRST anomalies and open closed coupling tell us that there is an $L_\infty$ morphism from 
$\fL_\lambda$ to the dg-Lie algebra $\HH^{\bullet}(\fP_\lambda,\fP_\lambda)$. 

Once $\fP_\lambda$ is identified with the algebra of functions on $SL(2,\bC)$, $\HH^{\bullet}(\fP_\lambda,\fP_\lambda)$ is known to be $L_\infty$ quasi-isomorphic to the Lie-algebra of polynomial poly-vectorfields, without higher operation. As a consequence, we have an $L_\infty$ morphism from 
$\fL_\lambda$ to polynomial poly-vectorfields.

However, the above morphism $\fL_\lambda\to \HH^{\bullet}(\fP_\lambda,\fP_\lambda)$ is not a quasi-isomorphism of commplexes. In order to complete the holographic dictionary, we should also characterize the image of $\fL_\lambda\to \HH^{\bullet}(\fP_\lambda,\fP_\lambda)$ as consisting of divergence-free polynomial poly-vectorfields, as in the B-model. It is not difficult to do so ``by hand'', by computing the action of a generating set of modes of $\fL_\lambda$. 

It would be better to do so in an Homological Algebra language which can be generalized to the general examples discussed in Section \ref{sec:HAcalculations}. We leave this question as an open problem.

\subsection{Determinant operators}
The prototypical determinant operator is $\det X$. The same argument used for $\Tr X^n$ shows that it is a BRST-closed local operator. It is a quasi-primary of dimension $\frac{N}{2}$. Using $SU(2)_R$ rotations, one gets a whole family of BRST-closed determinant operators $\det (X + u Y)$ which are also quasi-primary operators of dimension $\frac{N}{2}$.

The insertion of a $\det X$ operator in a correlation function 
requires order of $N$ Wick contractions with $Y$ fields in other operators. 
For example, a two-point function $\langle \det X(z) \det Y(w)\rangle$
involves an overall factor of $(z-w)^{-N}$. Accordingly, correlation functions of determinant operators tend to scale as $e^{\frac{1}{\hbar} S_{o}}$ in a 't Hooft expansion, with $S_{o}$ behaving as the action of a dual D-brane.

If we insert multiple determinants in a correlation function, we will add up contributions which have $n_{ij}$ Wick contractions between the $i$-th and $j$-th determinants. In a 't Hooft expansion, we may encounter a variety of non-trivial saddles as we vary the order $1$ parameters $\hbar n_{ij}$. Holographically, each determinant insertion imposes the presence of a ``giant graviton'' D-brane with a specific asymptotic behaviour at the holographic boundary. These boundary conditions may be satisfied by D-branes with a variety of non-trivial shapes in any given correlation function. In this section we will review and improve the known correspondence between saddles in the 't Hooft expansion and D-branes in $SL(2,\bC)$ with specified asymptotic behaviour.

It is useful to consider expressions such as $\det(m + X + u Y)$. These can be interpreted as generating functions for ``shortened'' determinant operators. i.e. traces of matrices of minors. Individual shortened determinants can be recovered by a contour integral 
\begin{equation}
    \oint_{|m|=1} \frac{dm}{m^{k+1}} \det(m + X + u Y) \,.
\end{equation}
All of these options are BRST-closed. These generating functions have a slightly counter-intuitive but useful behaviour in the large $N$ expansion: the associated correlation functions admit saddles where the number of Wick contractions remains finite and the answer is dominated by an overall factor 
of 
\begin{equation}
    m^N= e^{\frac{1}{\hbar} \lambda \log m} \, .
\end{equation}
Even though the saddle focuses on parts of the generating series with a finite number of $X$ and $Y$ fields, the 't Hooft expansion still includes surfaces with boundaries and the dual geometry includes a giant-graviton brane with a simple shape and action $\log m$. 

The D-brane-like properties of the determinant  becomes more manifest if we express the determinant as an integral over auxiliary (anti)fundamental fermions:\footnote{Replacing fermions with bosons gives inverse determinant operators. These are also interesting, but the analysis for the two cases is very similar.}
\begin{equation}
    \det(m + X)  = \int \left(\frac{d \psi d\bar \psi}{\hbar}\right)^N e^{\frac{1}{\hbar} \bar \psi (m + X)  \psi} \,.
\end{equation}
A straightforward expansion in Feynman diagrams with an $m^{-1}$ propagator for the auxiliary fermions leads to a 't Hooft expansion with an open string sector. 

In \cite{Budzik_2023}, the saddles for correlation functions of multiple determinants were recovered by merging all determinants 
\begin{equation}
    \prod_{i=1}^n \det\left(m_i + X(z_i) + u_i Y(z_i) \right)  = \int \prod_i \left(\frac{d \psi_i d\bar \psi^i}{\hbar}\right)^N e^{\sum_i \frac{1}{\hbar} \bar \psi^i \left(m_i + X(z_i) + u_i Y(z_i) \right)  \psi_i} \,.
\end{equation}
into a single normal-ordered expression within the auxiliary integrals: 
\begin{equation}
    e^{\sum_i \frac{1}{\hbar} \bar \psi^i \left(m_i + X(z_i) + u_i Y(z_i) \right)  \psi_i} = e^{\frac{1}{\hbar}\sum_{i<j} \frac{u_i-u_j}{z_i-z_j}\bar \psi_i \psi_j \bar \psi_j \psi_i }:e^{\sum_i \frac{1}{\hbar} \bar \psi^i \left(m_i + X(z_i) + u_i Y(z_i) \right)  \psi_i}:\,.
\end{equation}
A Hubbard-Stratonovich transformation eliminates the quartic terms at the price of introducing auxiliary variables $\rho^i_j$ for $i \neq j$. Setting $\rho_i^i=m_i$, the integrand becomes
\begin{equation}
    e^{-\frac{1}{\hbar}\sum_{i<j} \frac{z_i-z_j}{u_i-u_j}\rho^i_j \rho^j_i }:e^{\frac{1}{\hbar} \sum_{i,j} \rho^i_j \bar \psi_i \psi_j  + \sum_i \frac{1}{\hbar} \bar \psi^i \left(X(z_i) + u_i Y(z_i) \right)  \psi_i}:\,.
\end{equation}
As no other determinants are present and only a finite number of Wick contractions remain to be done, the $\psi$ integral behaves as $\det \rho^N$
and one arrives at saddle equations 
\begin{equation}
    \frac{z_i-z_j}{u_i-u_j}\rho^i_j = \lambda (\rho^{-1})^i_j\,, \qquad \qquad i \neq j
\end{equation}
for the $\rho^i_j$ integral. These equations were shown to determine the shape of a dual giant graviton D-brane in $SL(2,\bC)$. We will momentarily explain and extend that result.

There is a neat alternative way to arrive at these equations. We can just start from the normal-ordered exponent with a generic source $\rho$ 
\begin{equation}
    :e^{\frac{1}{\hbar} \sum_{i,j} \rho^i_j \bar \psi_i \psi_j  + \sum_i \frac{1}{\hbar} \bar \psi^i \left(X(z_i) + u_i Y(z_i) \right)  \psi_i}: 
\end{equation}
and compute the BRST variation. The BRST operator here consists of the part that acts on the chiral algebra fields $X,Y$ \eqref{eq:cano_BRST} and the part that acts on the auxiliary fermions $\bar{\psi}^i,\psi_j$, which is given by
\begin{align}
    \psi_i &\to c(z_i) \psi_i \cr
    \bar \psi_i &\to \bar \psi_i c(z_i)
\end{align}
We find that the BRST variation is proportional to
\begin{equation}
     \sum_{i,j}\left[\frac{u_i-u_j}{z_i-z_j} \bar \psi^i  \psi_j -\rho^i_j \right]\bar \psi^j (c(z_i) - c(z_j)) \psi_i \, .
\end{equation}
The first term comes from the 1-loop part $Q_1$ of \eqref{eq:cano_BRST} and the second the variation of the auxiliary fermions. At planar order we can replace $\bar \psi^i  \psi_j$ in parentheses with the Wick contraction $\lambda (\rho^{-1})^i_j$ and we recover the saddle equations 
\begin{equation}
     \sum_{i,j}\left[\lambda (u_i-u_j)  (\rho^{-1})^i_j -(z_i-z_j) \rho^i_j \right]\bar \psi^j \frac{c(z_i) - c(z_j)}{z_i-z_j} \psi_i \, .
\end{equation}
as conditions for planar BRST invariance of the combined operator. 

\subsection{Determinant modifications}
The integral expression for determinant operators allows us to define a large class of ``determinant modifications'', in the form of insertions of mesonic operators $\hbar^{-1} \bar \psi \cdots \psi$ in the auxiliary integral, which replace one symbol in the determinant with some string of symbols. These modifications behave as open string states attached to the dual D-brane \cite{Balasubramanian:2002sa}.\footnote{Several aspects of the construction and computation of open modifications described in this Section originally emerged in unpublished work with Kasia Budzik.}

It is also useful to formally add determinant modifications to the auxiliary action to define finite deformations. E.g. 
\begin{equation}
    \det(m + X+ \epsilon Y^2)  = \int \left(\frac{d \psi d\bar \psi}{\hbar}\right)^N e^{\frac{1}{\hbar} \bar \psi (m + X+\epsilon Y^2)  \psi}\,. 
\end{equation}
The algebraic structures we define below can all be understood in terms of the BRST anomaly (e.g. variation) of such formal expressions, possibly after deforming the bulk BRST charge as well.  

Even at the planar level, determining which modifications preserve BRST invariance takes some work. An important subtlety is that an expression involving the auxiliary field may vanish upon integration. For example, the BRST variation of the integral expression for $\det m+X(z)$ is
\begin{equation}
    \int \left(\frac{d \psi d\bar \psi}{\hbar}\right)^N e^{\frac{1}{\hbar} \bar \psi (m + X)  \psi} \, \frac{1}{\hbar} \bar \psi [c,m + X]  \psi
\end{equation}
and only vanishes after integration by parts. We will discuss momentarily how to use the BV formalism for the auxiliary integral to systematically deal with integration by parts. 

An important observation is that an insertion of $\bar \psi \psi$ 
can be traded of for $\hbar \partial_m$. Because of the overall prefactor $m^N$ to the 't Hooft expansion, this takes the form of $\lambda m^{-1}$ up to corrections subleading in the planar approximation. This effect is important in computing the planar OPE between single-trace operators and (modified) determinants, as well as the planar BRST variations.  
The analysis for non-trivial saddles is completely parallel with $m \to \rho$. 

In order to set up a BV formalism for the auxiliary integral, we introduce anti-fields $u$ and $\bar u$. To insure BRST invariance of the determinant, 
we can extend the auxiliary action to a BV action
\begin{equation} \label{eq:detBV}
    S_{\BV} = m \bar \psi \psi - \bar u c \psi + \bar \psi X \psi - \bar \psi c u \,.
\end{equation}
The two extra terms insure that the tree-level BV differential $\{S_{\BV}, \bullet\}$ maps $\psi \to c \psi$ and $\bar \psi \to \bar \psi c$ and 
cancels the chiral algebra BRST variation. 

More precisely, the consistency of the setup is constrained by a master equation:
\begin{equation}\label{eq:ad_BRST_plus_BV_laplacian}
    (Q_{\mathrm{BRST}}+ \hbar \Delta_\BV)e^{\frac{1}{\hbar} S_\BV} =0\,,
\end{equation}
where $Q_{\mathrm{BRST}}$ is the chiral algebra BRST variation and $\Delta_\BV$ the BV Laplacian 
\begin{equation}\label{eq:BV_laplacian}
    \Delta_\BV = \frac{\partial}{\partial u}\frac{\partial}{\partial \bar \psi}+\frac{\partial}{\partial \bar u}\frac{\partial}{\partial \psi}\,,
\end{equation}
which guarantees that the BRST variation integrated by parts to zero inside the auxiliary integral. 

Additionally, we may study space of deformations of such a BV action and the BRST anomalies they induce. These anomalies endow our space of deformations with a dg-Lie algebra structure:
\begin{align}\label{eq:differential_and_bracket_for_det_modificaitons_from_anomalies}
(Q_{\mathrm{BRST}}+ \hbar \Delta_\BV)&e^{\frac{1}{\hbar} (S_\BV + \sum_i\epsilon_i \mathcal{O}^i)} = \nonumber\\ &=\frac 1 \hbar \left(\sum_i\epsilon_i Q_{\text{det}}\mathcal{O}^i + \sum_{i,j}\epsilon_i\epsilon_j \{\mathcal{O}^i,\mathcal{O}^j\}_{\text{det}} + O(\hbar)\right)e^{\frac{1}{\hbar} (S_\BV + \sum_i\epsilon_i \mathcal{O}^i)}     
\end{align}
The anomaly linear in the $\mathcal{O}$'s endows them with a differential which at tree level is given by
\begin{align}
Q_{\text{det}}\mathcal{O} \coloneqq Q_{0}\mathcal{O} + \{S_{\BV}, \mathcal{O}\}_{\text{det}}  
\end{align}
with $Q_0$ the tree level BRST differential.
The bracket can be read off from the anomaly that is quadratic in the $\mathcal{O}$'s which at tree level is given by the usual BV bracket associated to the BV laplacian \eqref{eq:BV_laplacian}. 

Determinant modifications should be thought of as being in correspondence with these formal modifications of the action and are thus also equipped by the modified differential and bracket 
with the structure of a dg-Lie algebra. As usual, we can include higher $\hbar$ corrections restricted to 
a planar linear part which maps mesonic modifications to mesonic modifications. 

It is natural to replace the determinant with a $k$-th power of the determinant, by taking $k$ copies of all the auxiliary fields. This makes determinant modifications into $k \times k$ matrices $\mathfrak{gl}_k[\fD_\lambda]$
for some unital dg-algebra $\fD_\lambda$. We will denote this as the global symmetry algebra of the determinant operator(s). 

The algebra $\fD_\lambda$ should be identified with the algebra of global gauge transformations for the dual D-brane or, better, boundary world-sheet symmetries for the giant graviton brane $D$ dual to the determinant insertions. Although our notation does not keep track of that, the algebra depends on the choice of saddle via $\rho$.

In Section \ref{sec:HAcalculations} we will compute $\fD_0$. 

Much as we saw for $\fP_\lambda$, there is an interplay between $\fD_\lambda$ and $\fL_\lambda$:
\begin{itemize}
    \item A mode in $\fL_\lambda$ may create a modification out of an un-modified determinant. This gives a linear map $\fL_\lambda \to \fD_\lambda$.
    \item A mode in $\fL_\lambda$ may map a modification of a determinant to a different modification. This gives a bi-linear map $\fL_\lambda \times \fD_\lambda \to \fD_\lambda$.
    \item A mode in $\fL_\lambda$ acting on a determinant with two modification may merge them into a single modification. This gives a tri-linear map $\fL_\lambda \times \fD_\lambda \times \fD_\lambda \to \fD_\lambda$, etcetera. 
\end{itemize}
As we did for $\fP_\lambda$, we can do the above calculations at generic $k$, acting on a determinant modified by a sequence of mesons with matching consecutive flavour indices. 

The structure can be arranged into an $L_\infty$ morphism from $\fL_\lambda$ and the Hochschild cohomology of $\HH^{\bullet}(\fD_\lambda,\fD_\lambda)$. It encodes the BRST anomalies which appear if we add a mode of $\fL_\lambda$ to $Q_{\mathrm{BRST}}$ and determinant modification to $S_\BV$.

Both the space-filling D-brane $P$ and the giant graviton D-brane $D$ thus capture the bulk symmetries of the back-reacted world-sheet theory. We will now see that the giant graviton branes do better than $P$ in a different respect: they capture categorically the closed string states associated to single-trace local operators $\mathrm{Ops}_\lambda$.

There are two dual recipes: we either consider connected two-point functions of a modified determinant operator and a single-trace operator or expand a modified determinant operator at large $m$ 
and pick the single-trace part. As usual, it is convenient to use multiple determinants in order to pick out disk diagrams more easily. For example, 
\begin{equation}
    (\det\nolimits m+X)^k = \exp k \,\Tr \log (m+X) = m^N \exp k \sum_i (-1)^{i+1} \, m^{-i} \, \Tr X^i
\end{equation}
and we can take the part linear in $k$. When adding modifications, we can focus on a term where the flavour indices are contracted in a specific cyclic pattern. 

The overall result is a collection of multi-linear cyclic maps 
from $\fD_\lambda$ to $\mathrm{Ops}_\lambda$. With a bit of work, one can interpret this as a pairing between the cyclic cohomology $\HC^{\bullet}(\fD_\lambda)$ and the space of single-trace local operators. Giant graviton branes are thus a natural ingredient in a dg-TFT description of the deformed world-sheet theory. 

\subsection{Open Modifications}
We now focus on the interplay between the giant graviton brane(s) $D$ and the space-filling brane $P$. Accordingly, we add the $I$, $J$ fields to the chiral algebra and consider ``open'' determinant modifications such as $\frac{1}{\hbar} I^A \psi$ or $\frac{1}{\hbar}\bar \psi J_B$. Of course, such modifications will need to appear in pairs. The planar BRST complex of open modifications give 
two spaces $\fM_\lambda$ and $\wt \fM_\lambda$ respectively. We expect BRST-closed open modifications of the schematic form $\frac{1}{\hbar} I^A Y^n \psi$ and 
$\frac{1}{\hbar} \bar \psi Y^n J^A$.

Open modifications correspond holographically to open strings stretched between the space-filling branes and the giant graviton. On the world-sheet, they correspond to specific boundary-changing vertex operators (see Figure~\ref{fig:open modfs}). 
We can tentatively probe the action of boundary local operators on $P$ onto this junctions by looking at the action of $\fP_\lambda$ onto the open modifications. It is also possible to define an action of $\fD_\lambda$ via the modified BV brackets, which should match the action of local operators on $D$. 

\begin{figure}[h]
    \centering
    \includegraphics[width=0.65\linewidth]{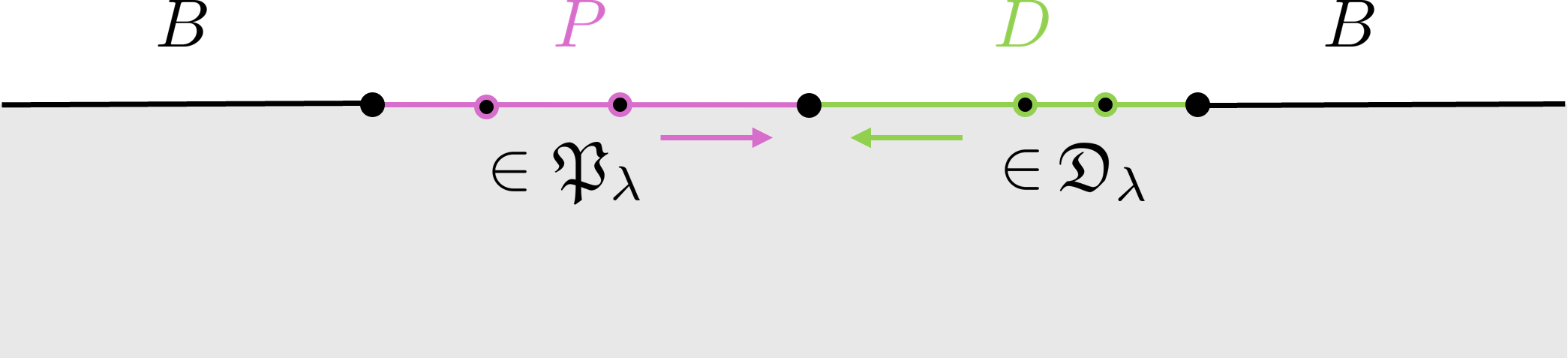}
    \caption{Illustration of the algebra of operators associated to the junction between the $P$ and $D$ branes acting as a left $\fP_{\lambda}$-module and a right $\fD_{\lambda}$-module.
    Purple dots, aka boundary local operators on $P$, can combine with elements from the junction algebra from the left. Green dots, aka boundary local operators on $D$, can combine with elements from the junction algebra from the right. }
    \label{fig:open modfs}
\end{figure}

We act with a meson on a determinant with open modifications. If our starting point is a determinant modified by $\frac{1}{\hbar} I^A \psi$ and $\frac{1}{\hbar}\bar \psi J_B$ and we act with a mode of $\cM^C_{a,b;D}$, the planar part of the answer is the sum of an action on $\frac{1}{\hbar} I^A \psi$ proportional to $\delta^A_D$, an action on $\frac{1}{\hbar}\bar \psi J_B$ proportional to $\delta^C_B$ and a 1-loop part 
which merges the two open modifications into a $\bar \psi \psi$ modification.

The first two parts are our main focus. They equip $\fM_\lambda$ and $\wt \fM_\lambda$ respectively with the structure of a left- and a right- module for  $\fP_\lambda$.

Consider in particular the action of the four modes $x$, $y$, $x'$ and $y'$ in $\fP_\lambda$. The meson $\frac{1}{\hbar} I^C X J_D$ can only have one Wick contraction with $\frac{1}{\hbar} I^A \psi$. Accordingly, the action of $x'$ vanishes and the action of $x$ produces a $\frac{1}{\hbar} I^C X \psi$ modification. Up to a total derivative, this gives back $-m \frac{1}{\hbar} I^C \psi$, i.e. $x$ acts as multiplication by $m$. 

On the other hand, $\frac{1}{\hbar} I^C Y J_D$ can also have a simultaneous Wick contraction with the action in the auxiliary integral, producing the combination of fields $\frac{1}{\hbar} I^C \psi \bar \psi \psi$ which at the planar level is equivalent to $\lambda m^{-1}\frac{1}{\hbar} I^C \psi$. We thus learn that $y'$ acts as $\frac{\lambda}{m}$, as required by the algebra relation $x y' - x' y = x y' =  \lambda$, and $y$ produces a new modification 
\begin{equation}
    \frac{1}{\hbar} I^C Y \psi + (\bar \psi \psi) \frac{1}{\hbar} \partial I^C \psi \,.
\end{equation} 

The action of the generators of $\fP_\lambda$ determines the action of the full algebra of holomorphic functions on $SL(2,\bC)$, which reproduces the expected answer for a D-brane $D$ supported at $x'=0$, $x+m=0$. This is indeed the expected shape of the ``basic'' saddle for $\det(m + X)$. Similarly, the basic saddle of $\det (m + X + u Y)$ gives a brane supported on $x'+ u y'=0$ and $m+ x + u y = 0$.

An analogous calculation can be done for non-trivial saddle of multiple determinants. We now have a vector $\frac{1}{\hbar} I^A(z_i) \psi_i$ of modifications and the four generators act as matrices 
on this vector, built from $\rho$, $\rho^{-1}$ and diagonal matrices $\mu$ and $\zeta$ containing the parameters $u_i$ and positions $z_i$. We have relations $\rho + x + y \mu =0$ and $x' + y' \mu = \zeta \rho$.

Remarkably, these two equations and the third implied relation below
\begin{align}
    x &= - \mu y - \rho \cr
    x' &= \zeta \rho - y' \mu \cr
  y' &= y \zeta + \lambda \rho^{-1}
\end{align}
define an action of the commutative algebra $\fP_\lambda$ iff the saddle equations for $\rho$ are satisfied. These expressions match the commuting matrices which appear in the spectral curve presentation of the dual giant graviton saddle \cite{Budzik_2023}.

We have thus connected directly the giant graviton probe of the dual geometry and the $\fP_\lambda$ probe and presented the dual brane in terms of a left- and right- module for the algebra of functions on $SL(2,\bC)$, corresponding to the boundary-changing local operators between $P$ and $D$. 

One could consider further deformations of the determinant/space filling brane system that are baryonic
\begin{align}
\epsilon_{a_1\cdots a_n} J^{a_1}\cdots J^{a_n} = \int \left(\frac{d \bar\psi}{\hbar}\right)^N e^{\bar\psi J} 
\end{align}
or antibaryonic 
\begin{align}
\epsilon^{a_1\cdots a_n} I_{a_1}\cdots I_{a_n} = \int \left(\frac{d\psi}{\hbar}\right)^N e^{I\psi } 
\end{align}
and combinations of these with the determinant insertion. In the large $N$ limit these correspond to the presence of an instanton background in the bulk arising from the interaction between the space-filling and determinant branes. Such instantons are studied in \cite{Lopez-Raven:2024vop}, and we leave it for future work to incorporate them within this dg-TFT framework.
\subsection{Conclusions}
\label{sec:ex_conclusion}
The main outcome of this section was an algorithmic proposal to extract certain algebras of world-sheet local operators from the 't Hooft expansion of OPE's between single-trace operators, mesons and determinant modifications. 

In particular, \begin{enumerate}
    \item The linearized Global Symmetry Algebra $\fL_\lambda$ of single-trace operators reproduced global gauge symmetries in the dual closed string theory, aka the worldsheet bulk local operators.
    \item The linearized Global Symmetry Algebra $\fP_\lambda$ of mesonic operators reproduced global gauge symmetries in the dual open string theory supported on space-filling branes, i.e. the corresponding boundary local operators.
    \item Determinant modifications reproduced the algebra $\fD_\lambda$ of open string states on giant-graviton branes, i.e. the corresponding boundary local operators. 
    \item Open determinant modifications reproduced open string states stretched between giant-graviton branes and space-filling branes as a bi-module $\fM_\lambda$ for the associated algebras, i.e. the corresponding boundary-changing local operators. 
\end{enumerate}

There are also rich connections between the algebras $\fL_\lambda$, $\fP_\lambda$ and $\fD_\lambda$, as closed string modes can always be used to deform open string theories on branes. We have
\begin{enumerate}
    \item As a result of open closed coupling, we have a $L_\infty$ morphism from $\fL_\lambda$ to the Hochschild cohomology $\HH^{\bullet}(\fP_\lambda,\fP_\lambda)$, whose images lie in the kernel of divergence operator. 
    \item A mode in $\fL_\lambda$ can create a determinant modification, change a determinant modification, merge two modifications into one, and so on. Hence, we also have a $L_\infty$ morphism from $\fL_\lambda$ to $\HH^{\bullet}(\fD_\lambda,\fD_\lambda)$.
\end{enumerate}

We will next do some more explicit calculations in the $\lambda \to 0$ limit
and then extend the analysis to more-general chiral algebras.

\section{Homological algebra and the B-model} \label{sec:homB}
It turns out that many of the $Q_0$-cohomology calculations we encountered in the previous section have a neat Homological Algebra interpretation. This should not be surprising: many constructions in Homological Algebra effectively formalize aspects of 2d dg-TFT, possibly with boundaries. In the $\lambda \to 0$ limit, the $Q_0$-cohomology calculations reproduce dg-TFT aspects of the B-model with target $\bC^3$, the original D-branes $B$ wrapping $\bC$, space-filling branes $P$ and branes $D$ associated to determinants. Accordingly, we will review in this section Homological Algebra aspects of the B-model and in the next how Homological Algebra appears in $\lambda \to 0$ calculations. 

While this section mainly illustrates the cohomology calculation and how they reproduce the space of bulk local operators of the B-model, the algebraic structures on the $Q_0$-complex are also of crucial importance. As will be discussed in Sections \ref{sec:B_model_C},\ref{sec:homological_dgTFT} and \ref{sec:Bulk_sym_def}, the Lie algebra structure on the bulk local operators encodes information about symmetries and infinitesimal deformations of both the bulk and boundary theories on branes, and will be further elaborated in later sections. In general, one should expect an $L_\infty$ structure, but in our main examples the higher operations vanish by virtue of the formality theorem \cite{Kontsevich:1997vb}. One can also consider the whole $E_2$ algebra structure on the bulk local operators. Reproduction of this structure from the homological computations of the boundary algebra is a well established mathematical result know as the "Deligne conjecture" \cite{Voronov1995,tamarkin1998another,mcclure1999solution,brav2023cyclic}. However, this aspect $2d$ dg-TFT will not be pursued here, as our main approach and focus remain on the symmetry algebras derived from various boundary theories on branes. 

\subsection{The B-model with target $\bC$}
\label{sec:B_model_C}
The B-model with target $\bC$ can be presented in a maximally simplified form \cite{Cattaneo:1999fm} as a 2d TFT with a first order action:
\begin{equation}
    \int \theta \, \text{d} \zeta
\end{equation}
where both $\theta$ and $\zeta$ are super-fields consisting of formal sums of forms of all degrees on the worldsheet, with a BRST differential acting as the de Rham differential. Classically, the B-model is a theory of constant maps from the world-sheet to a target space. 

The cohomology of the free BRST differential is built from the zero-form components of $\zeta$ and $\theta$, with worldsheet derivatives being exact. In order to lighten the notation, we will use the same symbol below for the 
super-fields and for their zero-form component fields. Higher form components occur via descent relations
when superfields are integrated over cycles on the world-sheet to define higher operations and integrated correlation functions. We refer to \cite{witten1988topological} for a general review.

The field $\zeta$ is bosonic and has ghost number $0$. It represents a holomorphic coordinate on $\bC$. The field $\theta$ is fermionic and has ghost number $1$. Local operators built as polynomials in $\zeta$ and $\theta$, i.e. elements of $\bC[\zeta, \theta]$, are identified as polyvector fields on $\bC$, i.e. elements of $\bC[\zeta, \partial_\zeta]$. 

The identification is strengthened by considering natural algebraic structures on the space of local operators. 
First of all, polynomial local operators can be safely multiplied: the propagator of the theory pairs 0-form and 1-form components of the superfields. 

There is also a ``bracket'' operation $\{\bullet, \bullet\}$ where the first descendant of a local operator is integrated on a circle around the other \cite{beem2020secondary}. In this free theory the bracket involves a single Wick contraction: it satisfies Leibniz rules and acts as 
\begin{equation}
    \{\theta,\zeta\} = 1
\end{equation}
It reproduces the Schouten bracket on poly-vector fields. 

The field theory meaning of the bracket operation is that it controls the BRST anomalies which may arise when the theory is deformed by some interaction term $\cI(\zeta,\theta)$. 
The anomaly is computed at tree level as 
\begin{equation}
    \{\cI, \cI\} \, ,
\end{equation} 
and in particular a deformation is non-anomalous if this vanishes. The product and the bracket define a Gerstenhaber algebra structure on polynomial local operators.

At loop order, one may have higher order perturbative corrections to this result. A ``formality theorem''  \cite{Kontsevich:1997vb} provides a scheme where the loop correction vanish and the product and bracket on local operators exactly matches the product and Schouten bracket on polyvector fields.

It is useful to elaborate on the idea that local operators which can be added to the action define the ``global symmetry algebra'' of the theory. Irrespectively of the BRST anomaly being cancelled or not, a deformation $\cI$ of the action modifies the BRST differential on local operators to
\begin{equation}
    \cO \to \{\cI,\cO\}
\end{equation}
and deforms similarly the action of the BRST charge on any other object in the theory: boundary conditions, defects, etc. 

The deformation of the BRST charge is, essentially by definition, a symmetry of the system. If $\cI$ is a vector field, this is precisely the standard action of the vector field on $\bC$.  

\subsection{Distributional local operators}
This simplified presentation of the B-model, as opposed to the full topological twist of a $(2,2)$ sigma model with $\bC$ target, makes it tricky to talk about objects in $\bC$ which are not holomorphic. 
Some can be described as disorder vertex operators. In particular, we can introduce a family of vertex operators 
\begin{equation}
    \delta(\zeta-z)
\end{equation}
playing the same role as a distributional vertex operator $\delta^{(2)}(\zeta-z, \bar \zeta- \bar z)d \bar \zeta$ in the twisted sigma model.\footnote{Here $z$ labels a point in the target space, not the worldsheet!} We assign ghost number $1$ to this operator. 
It is the first member of a tower consisting of derivatives $\partial_z^n \delta(\zeta-z)$
and of $\theta \partial_z^n\delta(\zeta-z)$. 

Such distributional vertex operators are typically needed in order to get some sensible correlation functions on a compact Riemann surface, because $\zeta$ has a zero mode which would make the path integral diverge unless it is soaked by a delta function. It seems reasonable to define the product and bracket of a polynomial local operator and a distributional one, but it does not seem natural to define the product or bracket between two distributional vertex operators. 

Distributional local operators should be crucial in defining a holographic dictionary. For example, in a B-model with target space $SL(2,\bC)$ we should consider local operators which correspond to the boundary-to-bulk propagators representing boundary insertions in holography. In a $\lambda \to 0$ limit, we will recognize them as local operators which are distributional in the directions parallel to the original D-branes.


\subsection{Boundary conditions in the B-model with target $\bC$}

We will often employ two types of boundary conditions: Neumann b.c. ${\theta\bigr|}_\partial =0$ and Dirichlet ${\zeta\bigr|}_\partial = 0$, possibly deformed to ${\zeta\bigr|}_\partial = z$. Polynomial boundary local operators are respectively identified with polynomials $\bC[\zeta]$ or $\bC[\theta]$, with the natural multiplication corresponding to the composition of the corresponding operators. For Neumann b.c. we may also want the boundary version of $\partial_z^n\delta(\zeta-z)$.

We should quickly elaborate on the relation between (polynomial) boundary local operators and boundary symmetries. Consider e.g. Neumann boundary conditions. We can modify such boundary conditions by adding a Chan-Paton factor $M$ and a boundary interaction
\begin{equation}
    \cI^\partial(\zeta) \in \End(M)[\zeta]\,.
\end{equation}
Such an interaction may give rise to a boundary BRST anomaly, which is precisely 
\begin{equation}
   \cI^\partial(\zeta)^2
\end{equation}
and thus ``knows'' about the algebra structure on the space $\bC[\zeta]$ of polynomial boundary local operators. Furthermore, the interaction induces a BRST differential, by using the commutator of $\bC[\zeta]$ valued matrices $[-,-]$:
\begin{equation}
    \cO^\partial \to [ \cI^\partial, \cO^\partial]
\end{equation}
on bundary operators
\begin{equation}
    \cO^\partial(\zeta) \in \End(M)[\zeta] \, ,
\end{equation}
which is a symmetry of the boundary condition (holomorphic 
rotation of $M$).

\subsection{Disc correlation functions}

Dirichlet b.c. give a fermionic zero mode on the disk $D^2$, allowing for the definition and calculations of disk correlation functions involving bulk and boundary polynomial vertex operators. In particular, we normalize the disk 1pt function of $\theta$ to be $1$. Similarly, we normalize the Neumann disk 1pt function of $\delta(\zeta)$ to be also $1$. 

There are unique boundary-changing local operators $\cO_{ND}$ and $\cO_{DN}$ intertwining Neumann and Dirichlet boundary conditions. They are annihilated by the action from the respective sides of $\theta$ or $\zeta$ boundary local operators. It is natural to normalize these junctions so that the disk correlation function 
\begin{equation}
    \langle \cO_{ND} \, \cO_{DN}\rangle^{D^2}_{N,D}=1\,,
\end{equation} with half Neumann and half Dirichlet b.c. is also normalized to $1$. Then the fusion of the two junctions to a Dirichlet local operator must produce $\theta$, while the fusion in the opposite direction must produce $\delta(\zeta)$:
\begin{equation}
    \cO_{DN}\, \cO_{ND} = \theta_{DD} \qquad \qquad \cO_{ND} \, \cO_{DN}=\delta(\zeta)_{NN}\,.
\end{equation}

There is no BRST-invariant junction between Dirichlet boundary conditions placed at different positions in the target space. It is convenient, though, to represent the (absence of) junctions via a two-dimensional space with a constant $z-z'$ differential mapping one summand onto the other. Then at $z=z'$ we recover $\bC[\theta]$. 

A good way to understand this setup is to describe the Dirichlet b.c. $\zeta = z$ as the deformation of the b.c. $\zeta = 0$ by a boundary coupling $z \theta$. 
Then the boundary-changing local operators between $\zeta = z$ and $\zeta =z'$ are presented as $\bC[\theta]$ equipped by a differential given by multiplication by $z \theta$ from the left and $z' \theta$ from the right, 
i.e. $(z-z')\theta$. The cohomology vanishes unless $z=z'$.

\subsection{Homological algebra and dg-TFT}
\label{sec:homological_dgTFT}
Homological algebra emerges first in a dg-TFT as one probes the potential BRST anomalies which arise from a deformation of the theory or of a boundary condition. Bulk local operators which can be used to deform the theory are equipped with an $L_\infty$ structure and a compatible product, such as the Gerstenhaber algebra of holomorphic polyvector fields for the B-model with some target space $X$. Boundary operators which can be used to deform the boundary condition are equipped with an $A_\infty$ structure, such as the algebra of holomorphic functions 
for Neumann boundary conditions in the B-model.  

A typical situation is that we are given a brane $B$ and the corresponding $A_\infty$ algebra $A$, sometimes equipped with the data of some disk correlation functions, and we want to learn information about the 
bulk local operators of the world-sheet theory it belongs to, or about other D-branes. 

The simplest example could be the Dirichlet brane for the B-model with target $\bC$: $A=\bC[\theta]$ and the only non-zero disk correlation function is $\langle \theta \rangle_{D^2} =1$. The disc correlation function can be thought of as a trace on $\bC[\theta]$. This is called an ``one-dimensional Calabi-Yau algebra'', because the trace lowers ghost number by $1$. It is a good blueprint for situations where we probe a target space by looking at point-like D-branes, and for non-geometric generalizations of that notion. In general, $n$-dimensional Calabi-Yau algebras and their generalizations give a standard description of dg-TFTs \cite{costello2007topological,lurie2008classification}. We also give a briefe review of various definition of Calabi-Yau algebras and their physical interpretation in Appendix \ref{appendix:CY}.

Neumann b.c. provide a more intricate example: $A$ is the algebra $\bC[z]$, but non-zero disk correlation functions require the insertion of a single distributional boundary local operator. Distributions supported at $\zeta=0$  can be thought of as the dual space $A^\vee$, with disk two-point functions giving the duality pairing. We will see that recovering a dg-TFT description of the world-sheet theory from this data is a bit more laborious. 

First of all, we can recall two distinct ways to recover bulk local operators from the data of a boundary condition. 

\subsection{Bulk local operators as symmetries/deformations of the boundary data.}
\label{sec:Bulk_sym_def}

Bulk local operators can be employed to infinitesimally deform the worldsheet theory, which in turn will give an infinitesimal deformation of the 
boundary $A_\infty$ structure.

There is a hierarchy of possible deformations:
\begin{enumerate}
    \setcounter{enumi}{-1}
    \item The simplest possible effect is that the deformation causes a BRST anomaly on the brane, i.e. a {\it curving} of the $A_\infty$ algebra. This is computed simply by bringing the bulk local operator to the boundary. A classic example is a superpotential deformation $\cI(\zeta)$, which generates an anomaly on Neumann b.c. and leads to the theory of matrix factorizations \cite{Kapustin:2002bi}, i.e. solutions of the boundary BRST anomaly cancellation condition:
    \begin{equation}
        \cI(\zeta) + \cI^\partial(\zeta)^2 =0\, .
    \end{equation}
    This effect is simply described as a map from bulk local operators to $A$ which preserves ghost number. 
    \item The next simplest effect is that the deformation modifies the BRST differential on boundary local operators. For example, a bulk deformation by $\zeta$ will induce $Q \theta=1$ on Dirichlet b.c, while a bulk $\theta$ will induce $Q\zeta=1$ on Neumann b.c. This effect is described as a map from bulk local operators to the space of maps $A \to A$, which lowers the ghost number by $1$. 
    \item Bulk deformations can also deform the multiplication of boundary local operators. For example, we will see that in the B-model with target $\bC^n$, Poisson bivectors give a deformation quantization of the algebra of functions on Neumann bounary conditions. This effect is described at the leading order as a map from bulk local operators to the space of maps $A \otimes A \to A$, which lowers the ghost number by $2$. 
    \item Deformations of boundary $n$-ary operations are described by a map from bulk local operators to the space of maps $A^{\otimes n} \to A$, which lowers the ghost number by $n$.
\end{enumerate}
A collection $\fh$ of maps $A^{\otimes n} \to A$ for all $n$, lowering ghost number by $n$, is by definition an element of the Hochschild complex $\CH^\bullet(A,A)$. The complex is equipped with a differential $Q_{\CH}$, 
such that the map from bulk local operators to $\CH^\bullet(A,A)$ intertwines the bulk BRST charge and $Q_{\CH}$. \footnote{A small subtlety which we will typically neglect is that $A$ is unital and $n$-ary operations by definition act trivially on the identity. Physically, the identity operator has trivial descendants. Accordingly, one should employ the {\it relative} Hochschild complex $\CH_{\mathrm{rel}}^\bullet(A,A)$ of maps $(A/\bC 1)^{\otimes n} \to A$ here and elsewhere in the paper. Furthermore, if $B$ is a direct sum of elementary D-branes and $A$ is secretly a category, the definition of relative Hochschild complex should be further amended to remove identity operators for each individual D-brane. See a discussion in Section \ref{sec:closed}.}

In other words, there is a chain map from the space of bulk local operators which can deform the theory (which we expect to include polynomial ones, but not distributional ones) to the Hochschild cohomology $\HH^\bullet(A,A)$ for the algebra of boundary operators. This chain map is often a quasi-isomorphism of complexes, and should also extend to a quasi-isomorphism of $L_\infty$ algebras or $E_2$ algebras. Intuitively, we expect it to capture all bulk local operators which can act on the brane.

We will denote the maps associated to a local operator $\cI$ as
\begin{equation}
    \{\cI|\bullet, \cdots, \bullet\}_n\,.
\end{equation}

We can also push the deformation theory beyond the leading order. The Hochschild complex $\CH^\bullet(A,A)$ is equipped with a bracket, which is roughly defined by summing over all possible ways one map can be composed with the other and makes it into a dg-Lie algebra. Deformations satisfy a Maurer-Cartan equation $Q_{\CH} \cI + \{\cI, \cI\}=0$ which must be the image of the MC equation satisfied by the bulk local operators defining the deformations. 

In other words, the map from bulk local operators to the Hochschild complex $\CH^\bullet(A,A)$ is a morphism of $L_\infty$ algebras. This map can be rather non-trivial. For example, in the case of the B-model with target $\bC^n$ and
Neumann b.c., a bulk deformation given by a Poisson bivector $\cI$ modifies the boundary algebra in a way which is highy non-linear in $\cI$, giving the corresponding deformation-quantization star-product. The coefficients are the higher maps
\begin{equation}
    \bC[\zeta_i, \theta^i]^{\otimes k} \to \CH^2(\bC[\zeta_i],\bC[\zeta_i])
\end{equation}
in the morphism of $L_\infty$ algebras are computed by certain loop Feynman diagrams \cite{Kontsevich:1997vb}. 

In our analysis, we will often encounter situations where a theory and boundary conditions are assembled by stacking simpler theories and boundary conditions. For example, a B-model with target space $\bC^3$ can be obtained by stacking three copies of the B-model with target space $\bC$. A brane wrapping $\bC \subset \bC^3$
can be obtained by stacking a Neumann b.c. in one direction and Dirichlet b.c. in other directions. 

If we are given two systems with boundary $A_\infty$ algebras $A_1$ and $A_2$, the algebra for the combined system will be a tensor product $A_1 \otimes A_2$.\footnote{The definition of tensor product for $A_\infty$ algebras is typically scheme-dependent, but when at least one of two factors is a dg-algebra, we can take the standard tensor product.} Bulk local operators should also be the tensor product of the operators in the two factors. It turns out that there is an isomorphism of Gerstenhaber algebra
\cite{LE2014Hochtensor}:
\begin{equation}
    \HH^\bullet(A_1 \otimes A_2,A_1 \otimes A_2)  \simeq \HH^\bullet(A_1,A_1) \otimes \HH^\bullet(A_2,A_2)\,
\end{equation}
which should be compatible with the analogous statement for bulk local operators. 

Notice that we only discussed brackets of bulk local operators for now, and not the product. This is because the product does not have an immediate interpretation as a BRST anomaly. Stacking offers a somewhat indirect interpretation: the bracket in a product system is the combination of the bracket in one of the two factors and the product in the other factor. Correspondingly, the Hochschild complex can be equipped with a product as well, making it into a Gerstenhaber algebra \cite{Gerstenhaber} in a manner compatible with the bulk product. 

\subsection{Hochschild cohomology and Dirichlet boundary conditions}
We can illustrate these ideas for the case of the B-model with $\bC$ target space and Dirichlet b.c., so that $A=\bC[\theta]$ and bulk local operators are $\bC[\zeta,\theta]$: 
\begin{enumerate}
\setcounter{enumi}{-1}
    \item The ``bulk to boundary'' map is the obvious $\bC[\zeta,\theta] \to \bC[\theta]$ given by setting $\zeta=0$:
    \begin{equation}
        \{1|\}_0 =1 \qquad \qquad \{\theta|\}_0 =\theta
    \end{equation}
    and everything else vanishes.
    \item The deformation of the differential arises from a Feynman diagram with a single Wick contraction, 
    so it requires the presence of a single factor of $\zeta$ and removes a $\theta$ from the boundary operator:
    \begin{equation}
        \{\zeta|\theta\}_1 =1 \qquad \qquad \{\theta\zeta|\theta\}_1 =\theta\,.
    \end{equation}
    \item More generally (we are not careful here with overall non-zero coefficients), 
     \begin{equation}
        \{\zeta^n|\theta, \cdots, \theta\}_n =1 \qquad \qquad \{\theta\zeta^n|\theta, \cdots, \theta\}_n =\theta
    \end{equation}
    and everything else vanishes.
\end{enumerate}
The above physical considerations can be summarized as yielding a tower of maps $\bC[\zeta,\theta] \to \mathrm{Hom}(\bC[\theta]^{\otimes n},\bC[\theta])$. Simple computation shows that these maps preserve the Hochschild differential, and hence define a chain map $\bC[\zeta,\theta] \to \CH^{\bullet}(\bC[\theta],\bC[\theta])$. With a little bit more work one can show that it is a quasi-isomorphism \cite{loday2013cyclic}. Therefore, polynomial bulk local operators are thus fully accounted by the Hochschild cohomology, i.e. 
\begin{equation}
    \HH^{\bullet}(\bC[\theta],\bC[\theta]) \simeq \bC[\zeta,\theta]\,.
\end{equation}
\subsection{Hochschild cohomology and Neumann boundary conditions}
In the case of the B-model with $\bC$ target space and Neumann b.c., we have $A=\bC[\zeta]$ and bulk local operators are $\bC[\zeta,\theta]$:
\begin{enumerate}
\setcounter{enumi}{-1}
    \item The ``bulk to boundary'' map is the obvious $\bC[\zeta,\theta] \to \bC[\zeta]$ given by setting $\theta=0$:
    \begin{equation}
        \{\zeta^n|\}_0 =\zeta^n 
    \end{equation}
    and everything else vanishes.
    \item The deformation of the differential arises from a Feynman diagram with a single Wick contraction:
    \begin{equation}
        \{\zeta^n \theta|\zeta^m\}_1 =m \zeta^{n+m-1} \,.
    \end{equation}
    \item Bulk interactions cannot give any higher deformations. 
\end{enumerate}
As in the discussion of the previous section, the above physical analysis gives rise to a chain map $\bC[\zeta,\theta] \to \CH^{\bullet}(\bC[\zeta],\bC[\zeta])$, which is in fact a quasi-isomorphism. Therefore, polynomial bulk local operators are thus fully accounted by the Hochschild cohomology, i.e. 
\begin{equation}
    \HH^{\bullet}(\bC[\zeta],\bC[\zeta]) \simeq \bC[\zeta,\theta]\,.
\end{equation}

\subsection{Disc correlation functions and $\HH^\bullet(A,A^\vee)$}
Another probe of bulk local operators are correlation functions on the disk, with one bulk insertion, one boundary insertion and possibly $n$ integrated boundary insertions. 
Each bulk operator $o$ thus gives maps  
\begin{equation}
        (o|\bullet;\bullet,\cdots, \bullet)_n 
\end{equation}
from $A^{\otimes (n+1)}$ to the complex numbers, or equivalently from $A^{\otimes n}$ to $A^\vee$.

A collection of such maps defines an element of the Hochschild complex $\CH^\bullet(A,A^\vee)$. The complex is again equipped with a differential $Q_{\CH}$, such that the map from bulk local operators to $\CH^\bullet(A,A)$ intertwines the bulk BRST charge and $Q_{\CH}$. 

Depending on the choice of D-brane, the bulk operators detected by $\CH^\bullet(A,A^\vee)$
may be normalizable or distributional. For example, in the case of the B-model with target space $\bC$: 
\begin{itemize}
    \item For the Dirichlet boundary condition, $\bC[\theta]^\vee$ is essentially equivalent to $\bC[\theta]$ and $\CH^\bullet(\bC[\theta],\bC[\theta]^\vee)$ reproduces the polynomial bulk operators $\bC[\zeta,\theta]$. For example, 
\begin{equation}
    (1|\theta;)_0 =1 \qquad \qquad (\theta|1;)_0 =1 \qquad \qquad (\zeta|\theta;\theta)_0 =1
\end{equation}
etcetera.
    \item For Neumann b.c., disk correlation functions with boundary operators in $\bC[\zeta]$ detect distributional bulk states. For example, 
\begin{equation}
    (\delta(\zeta)|1;)_0 =1 \qquad \qquad (\partial^n \delta(\zeta)|\zeta^n;)_0 =n! \qquad \qquad (\theta \delta(\zeta)|1;\zeta)_0 =1
\end{equation}
etcetera. Polynomial bulk operators can be detected if we use distributional boundary operators, but this brings us back to $\HH^\bullet(A,A)$.
\end{itemize}

The $\CH^\bullet(A,A^\vee)$ complex is not equipped with a product or bracket, but there are mixed products and brackets with elements of $\CH^\bullet(A,A)$: intuitively, symmetries act on disk correlation functions. 

\subsection{Equivariant bulk local operators and $\HC^\bullet(A)$}
It is also possible to consider disk correlation functions with a bulk local operator but without boundary insertions, or with all integrated boundary insertions. An important subtlety is that such configurations are rotationally invariant and thus involve rotation-invariant bulk local operators. In a BRST setting, this notion needs to be refined to {\it rotation-equivariant} bulk local operators. In String Theory, rotation-equivariant bulk local operators are precisely the building blocks of closed string states \cite{Zwiebach:1992ie}! \footnote{In conformal gauge, this is the statement that the BRST cohomology of string states is computed by removing the ghost zero mode $c_0 - \bar c_0$ corresponding to rotations and imposing rotational symmetry/level matching by hand. It is a {\it relative} BRST cohomology, due to the fact that the rotation sub-group of the diffeomorphisms group is compact.} For example, in the B-model, rotation-equivariant bulk local operators correspond to divergence-free polyvector fields, which are the building blocks of BCOV theory.

These disk correlation functions thus map rotation-equivariant bulk operators to {\it cyclic} (i.e. $\bZ_n$-invariant) maps $A^{\otimes n} \to \bC$, which we denote as 
\begin{equation}
        (c|\bullet,\cdots, \bullet)_n \,.
\end{equation}
This leads to the definition of the cyclic cohomology complex $\CC^\bullet(A)$, with cohomology $\HC^\bullet(A)$ \cite{loday2013cyclic}. This is roughly the cyclic-invariant part of $\HH^{\bullet-1}(A,A^\vee)$. This intuition is made precise by the Connes construction, which involves a map $B$ on $\HH^{\bullet}(A,A^\vee)$ that is the analogue of the divergence operation on polyvector fields.   

For the B-model on $\bC$, divergence-free polynomial polyvector fields consist of $\bC[\zeta] \oplus \theta \bC$. Distributional divergence-free polyvector fields consist of $\bC[\partial]\delta(\zeta)$.
\begin{itemize}
    \item For the Dirichlet boundary condition, we have disk correlation functions 
    \begin{equation}
        (\theta|)_0 =1 \qquad \qquad (\zeta^{n-1}|\theta, \cdots, \theta)_n =\frac{1}{n} \,.
    \end{equation}
An useful string theory intuition to understand these formulae is that a polyvector field $n \zeta^{n-1}$ can be identified with the 1-form $n \zeta^{n-1}d\zeta$, whose primitive $\zeta^n$ evaluated at the location $z$ of a Dirichlet brane gives the disk correlation function deformed by a $z\theta$ boundary interaction. 
    \item For Neumann b.c., we have disk correlation functions  
\begin{equation}
    (\partial^n \delta(\zeta)|\zeta^n)_1 =n!\,.
\end{equation}
\end{itemize}

There is an action of $\CH^\bullet(A,A)$ on $\CC^\bullet(A)$, describing the action of symmetries on disc correlation functions. 

\subsection{Branes, modules and tensor products}
We can also use a specific reference D-brane $B$ to probe other D-branes $B'$, through the properties of boundary-changing local operators. The spaces $M_{BB'}$ and $M_{B'B}$ of boundary-changing local operators are naturally left- and right- ($A_\infty$) modules
for $A$, by composition of local operators along the boundary. 
These modules encode many properties of $B'$. 

Boundary local operators on $B'$ can be probed in multiple  ways:
\begin{itemize}
    \item They can act on boundary-changing local operators. Parsing definitions, this gives a chain map from local operators on $B'$ to endomorphisms of $A_\infty$ modules for $M_{BB'}$ and for $M_{B'B}$. This will typically capture normalizable operators on $B'$.
We can also vary the choice of $B'$, building a whole functor from the $A_\infty$ category of branes to the dg-categories of left- and right- $A$-modules. \footnote{As we briefly mentioned for the Hochschild cohomology, we should really consider notions of morphisms appropriate to {\it unital} algebras and modules.}
\item We can also compose boundary-changing local operators to produce (possibly distributional) operators on $B'$. Parsing definitions, this gives a chain map from the derived tensor product $M_{B'B} \otimes_A M_{BB'}$ \footnote{We use the symbol $\otimes_{A}$ to denote derived tensor product $\otimes_{A}^{\mathbb{L}}$ throughout the paper} to a space of (possibly distributional) local operators on $B'$.
\item If disk correlation functions are available, they will instead give a pairing between a space of (possibly distributional) local operators on $B'$ and $M_{B'B} \otimes_A M_{BB'}$. 
\end{itemize}
These maps and structures can be combined with the relations to bulk operators, and there are several Homological Algebra constructions which encode these combinations. They guarantee, say, that $\HH^{\bullet}(A,A)$ maps to the Hochschild cohomology of local operators on $B'$, etc. 

\subsection{B-model probe branes}
As an illustration, we can describe branes in the B-model with target $\bC$ using the Dirichlet brane as a reference, with $A=\bC[\theta]$. The Neumann brane is described by trivial one-dimensional modules. As the disk correlation functions in this B-model have a one unit of ghost number anomaly,  if we set $M_{DN}=\bC$ in ghost number $0$, then $M_{ND}=\bC[1]$ is supported on ghost number $1$. 

The $A_\infty$ endomorphisms of $\bC$ as a $\bC[\theta]$-module are a collection of maps $\bC[\theta]^{\otimes n} \to \bC$ which lower the ghost number by $n$, deforming the module action. The maps $\theta^{\otimes n} \to 1$ represent the $\zeta^n$ local operators on the Neumann b.c.

The derived tensor product $M_{ND}\otimes_A M_{DN}$, on the other hand, 
includes summands of the form $M_{ND}\otimes A^{\otimes n}\otimes M_{DN}$
in ghost number shifted by $-n$, with a specific differential. Here the elements $\bC[1] \otimes \theta^{\otimes n}\otimes \bC$ represent the 
distributional local operators $\partial^n \delta(z)$. For example, a local operator defined by going from Neumann to Dirichlet and back to Neumann b.c. clearly forces $\zeta \to 0$ and is identified with $\delta(z)$. 

In general, a $\bC[\theta]$ $A_\infty$ module $M_{DB}$ consists of some vector space $V$ and a collection of maps $V \to V$ of ghost number $1$ representing the action of $n$ $\theta$'s. We can collect all the maps into a polynomial differential $d(\zeta)$ on $V$, which satisfies $d(\zeta)^2=0$ by the module axioms. It is not hard to see that this coincides with the standard ``tachyon condensation'' presentation of B-model D-branes as a complex on $\bC$. We can think about that as an enriched Neumann brane, with Chan-Paton factor $V$ and boundary interaction $\cI^\partial(\zeta) = d(\zeta)$. 

In this situation, the dual module $M_{BD}$ will coincide with 
$V^\vee[1]$ with transpose differential. The derived tensor product should coincide with distributional maps on $V$ with a $[d(\zeta),\bullet]$ differential.  

The case of reference Neumann boundary conditions can be analyzed in a similar fashion. Here we simply report the description of Dirichlet local operators:
\begin{equation}
    M_{DN} \otimes_A M_{ND} = \bC \otimes_{\bC[\zeta]} \bC[1] = \bC[\theta]\,.
\end{equation}

\subsection{B-model on $\bC^n$}
The B-model on $\bC^n$ can be presented in terms of (super)fields $\zeta_i$ and $\theta^i$. 
Polynomial local operators form the algebra $\bC[\zeta_i,\theta^i]$ which is identified with holomorphic poly-vectorfields on $\bC^n$, with bracket induced from 
\begin{equation}
    \{\theta^i, \zeta_j\} = \delta^i_j\,.
\end{equation}
Rotation-equivariant local operators are identified with divergence-free holomorphic polyvector fields. General branes can be presented as enriched Neumann b.c., in terms of a (graded super-)vector space $V$ 
equipped with a nilpotent differential $d(\zeta)$.

These structures are all naturally recoverable by looking at a Dirichlet brane probe $\zeta_i |_\partial=0$, with boundary local operators $A=\bC[\theta^i]$. The Hochschild cohomology $\HH^{\bullet}(A,A)$ is well-behaved under tensor product and thus must be isomorphic to $\bC[\zeta_i,\theta^i]$, the polynomial local operators. 
The Hochschild cohomology $\HH^{\bullet}(A,A^\vee)$ is isomorphic to $\HH^{\bullet}(A,A)$, via the pairing $(ab)$ on $A$, with $(\theta^1 \cdots \theta^n)=1$ being the disk boundary 1-pt function. It is useful to think about $\HH^{\bullet}(A,A^\vee)$ as holomorphic $(\bullet,0)$ forms, by contracting the polyvector fields with the top form $d\zeta_1 \cdots d\zeta_n$.

Cyclic cohomology $\HC^{\bullet}(A)$ can be recovered from Connes construction, the role of $B$ being played by the $\partial$ operator acting on $(\bullet,0)$ forms. The result is $\partial$-closed holomorphic $(\bullet,0)$ forms, i.e. divergence-free polyvector fields. This computation is somewhat indirect, as it goes through the non-trivial tensor product quasi-isomorphism for $\CH^{\bullet}(A,A^\vee)$ and Connes construction. 

A useful perspective which we will explore further in the next section is that a $\partial$-closed holomorphic form $\alpha$ on $\bC^n$ has a primitive $\partial^{-1} \alpha$ which is what enters a D-brane's action. In particular, a cyclic homology element such as 
\begin{equation}
    (c|\theta^{i_1},\cdots, \theta^{i_n}) = c^{(i_1 \cdots i_n)}
\end{equation}
with a totally symmetric right hand side and other entries being $0$ represents a coupling of the Dirichlet brane to a function $c^{(i_1 \cdots i_n)}\zeta_{i_1} \cdots \zeta_{i_n}$ and thus the $\partial$-closed holomorphic $(1,0)$ form 
\begin{equation}
    \partial (c^{(i_1 \cdots i_n)}\zeta_{i_1} \cdots \zeta_{i_n})\,.
\end{equation}

We can present (but not justify here) a neat generalization of this statement. Introduce the algebra $\bC[\theta^i, d\theta^i]$ equipped with the obvious differential $d$. It is easy to see that a form $\partial^{-1} \alpha$, which is defined up to $\partial$-exact forms, defines a linear function on $d$-closed elements of $\bC[\theta^i, d\theta^i]$, which pairs up $\theta^i$ with $d\zeta_i$ and $d\theta^i$ with $\zeta_i$. We can extend that to an element of $\HC^{\bullet}(\bC[\theta^i])$ as follows: act with $d$ on all arguments, multiply them together and pair them with $\partial^{-1} \alpha$

Finally, other D-branes can be probed as $A_\infty$ modules for $\bC[\theta^i]$, giving directly the data of $V$ and $d(\zeta)$.

If we replace the Dirichlet branes in the analysis with some other probe branes, we may encounter various distributional local operators. The main actor in this paper will be D-branes wrapping a $\bC$ factor in $\bC^3$, which have an algebra $A= \bC[\theta^1,\theta^2,\zeta_3]$. While the Hochschild cohomology $\HH^{\bullet}(A,A)$ still recovers 
polynomial polyvector fields, $\HH^{\bullet}(A,A^\vee)$ and $\HC^{\bullet}(A)$ involve polyvector fields which are distributional in the $\zeta_3$ direction, interacting with the $\bC$ D-brane at a specific point. 

\subsection{B-model with target $\bC^2/\Gamma$}
Finally, we discuss an orbifold geometry which appears in generalizations of the Twisted Holography setup. The ADE singularities can be defined as the orbifold of $\bC^2$ by the action of a discrete subgroup $\Gamma$ of $SU(2)$ which fixes the origin. Correspondingly, the B-model on $\bC^2/\Gamma$ can be obtained from the B-model on $\bC^2$ by an orbifold by $\Gamma$. 

It is particularly useful to think about the behaviour of Dirichlet branes under the orbifold:
\begin{itemize}
    \item A Dirichlet brane in $\bC^2$ supported away from the origin combines with its $\Gamma$ images and survives the orbifold to give a Dirichlet brane in $\bC^2/\Gamma$ supported away from the origin.
    \item A Dirichlet brane at the origin of $\bC^2$ decomposes under the orbifold to a collection of ``exceptional branes'' supported at the origin of $\bC^2/\Gamma$. 
\end{itemize}
A Dirichlet brane supported away from the origin will decompose to a collection of exceptional branes at the origin.

Denote as $\bC[\Gamma]$ the vector space with a basis labelled by elements of $\Gamma$. The direct sum of $\Gamma$ images of a Dirichlet brane away from the origin of $\bC^2$, can be brought to the origin
to give a Dirichlet brane with CP factor $\bC[\Gamma]$. The local operators on such a brane $D_{\mathrm{tot}}$ form the matrix algebra 
\begin{equation}
    \mathrm{End}(\bC[\Gamma])[\theta^1, \theta^2]
\end{equation}
with elements of the form $e_{g,g'}$, $e_{g,g'}\theta^i$, $e_{g,g'}\theta^1 \theta^2$, with 
\begin{equation}
    e_{g,g'} e_{g'',g'''} = \delta_{g',g''} e_{g,g'''}\,.
\end{equation}

The orbifold projects that algebra to the $\Gamma$-invariant part
\begin{equation}
    A \equiv \mathrm{End}(\bC[\Gamma])[\theta^1, \theta^2]^\Gamma\,,
\end{equation}
where $\Gamma$ acts simultaneously on the $\theta^i$ as an $SU(2)$ rotation and on the basis element by group multiplication. 
The individual exceptional branes in $D_{\mathrm{tot}}$ are identified as idempotent elements in the ghost number $0$ part $\mathrm{End}(\bC[\Gamma])^\Gamma$. 

It is useful to identify 
$\bC[\Gamma]$ as the group algebra. It contains every irreducible representation $R_i$ of $\Gamma$ exactly $\mathrm{dim} \; R_i$ times. 
Indeed, it is a know result (see e.g. \cite{serre1977linear}) that $\bC[\Gamma]$ decomposes as 
\begin{equation}
    \bC[\Gamma] = \oplus_i R_i \otimes R_i^\vee
\end{equation}
under the left- and right- actions of $\Gamma$.

Then by Schur's lemma, $\mathrm{End}(\bC[\Gamma])^\Gamma$
coincides with 
\begin{equation}
    \mathrm{End}(\bC[\Gamma])^\Gamma= \oplus_i 1_{R_i} \otimes \mathrm{End}(R_i) \,.
\end{equation}
We interpret this as giving a decomposition 
\begin{equation}
    D_{\mathrm{tot}} = \oplus_i D_i \otimes R_i\,,
\end{equation} 
where the exceptional branes $D_i$ are labelled by irreducible representations of $\Gamma$ and appear $\mathrm{dim} R_i$ times in 
$D_{\mathrm{tot}}$.

We can then decompose the whole algebra $A$ into a category of exceptional branes $D_i$: 
\begin{equation}
    \mathrm{End}(\bC[\Gamma])[\theta^1, \theta^2]^\Gamma = \oplus_{i,j} \Hom(R_i,R_j)[\theta^1,\theta^2]^\Gamma \otimes \Hom(R_i,R_j)\,.
\end{equation}
The first factor can be computed by decomposing 
\begin{equation}
    R_j[\theta^1,\theta^2] = \oplus_i R_i \otimes \Hom(R_i,R_j)[\theta^1,\theta^2]^\Gamma
\end{equation}
into irreducible representations of $\Gamma$, with coefficients $\Hom(R_i,R_j)[\theta^1,\theta^2]^\Gamma$. 

We learn that $\Hom(R_i,R_j)[\theta^1,\theta^2]^\Gamma$ contains the local operators from $D_i$ to $D_j$. 
In particular, at ghost numbers $0$ and $2$ we have a single element between $D_i$ and $D_i$, while at ghost number $1$ we have a generator for every time $R_i$ enters in the tensor product of $R_j$ and the fundamental representation of $\Gamma$. This is the number of edges in the ``affine ADE quiver'' associated to $\Gamma$. 

An orbifold has two effects on local operators: it projects to $\Gamma$-invariants but it can also add new twisted sectors, which in the original theory are local operators living at the end of topological line defects which implement the action of elements in $\Gamma$.

The Hochschild cohomology of $A$ reproduces, non-trivially, this statement. It includes the $\Gamma$-invariant part of the Hochschild cohomology of $\bC[\theta^1, \theta^2]$, but also twisted sectors localized at the origin of the ADE singularity \cite{FARINATI2005415}. 

Another useful perspective on branes in the orbifold theory is that they can be understood in terms of a vector space $V$ equipped with a $\Gamma$ action and a differential $d(\zeta)$ compatible with that action. In particular, setting $d(\zeta)=0$ gives us a variant $N_i$ of Neumann b.c. for every irreducible representation $R_i$, which has a one-dimensional space of junctions with $D_i$ and zero-dimensional with other $D_j$ Dirichlet branes. 

\section{Homological algebra in the $\lambda \to 0$ limit. }\label{sec:hom}
We are now equipped to review the computation of the tree-level (aka planar, $\lambda \to 0$) BRST cohomology of single-trace operators. Recall that the chiral algebra describes the world-volume theory of $N$ D$1$ branes in $\bC^3$. Concretely, the D$1$ brane is a boundary condition $B$ in the B-model which combines Dirichlet b.c. $\zeta_1 |_\partial = \zeta_2 |_\partial = 0$ in two directions and Neumann $\theta_3 |_\partial = 0$ in the third. Boundary local operators on $B$ form the algebra $\bC[\theta^1,\theta^2,\zeta_3]$.

The chiral algebra fields are couplings for a general deformation of the stack of D-branes:
\begin{equation}
    \Phi(\theta^1,\theta^2,\zeta_3) \equiv c(\zeta_3) + X(\zeta_3) \theta^1 + Y(\zeta_3) \theta^2+ b(\zeta_3) \theta^1 \theta_2\,.
\end{equation}
We denote this sort of object, which pairs up the fields and their derivatives with the corresponding boundary vertex operators, a {\it generating field}. The tree-level BRST transformations of the generating field are neatly expressed in terms of algebra structure: 
\begin{equation}
    Q_0 \Phi = \Phi \Phi\,.
\end{equation}
Conversely, any tree-level BRST differential for a gauge theory with single-trace action can be interpreted as arising from an $A_\infty$ algebra and thus a world-sheet dg-TFT. The 2d chiral gauge theory arises from $\bC[\theta^1,\theta^2,\zeta_3]$.

The deformation $\Phi$ of the stack of D-branes changes the coupling of the closed string modes to the D-branes, i.e. the disk bulk 1-pt functions. In an Homological Algebra language, 
the disc 1-pt function of a bulk rotation-equivariant local operator $\fc$ becomes
\begin{equation}
    \cO_\fc \equiv \frac{1}{n \hbar} \Tr (\fc|\Phi, \cdots, \Phi)_n \, .
\end{equation}
Here, the map $\fc$ acts on the algebra elements and the fields are brought out of the map by linearity, up to Koszul signs. They are then composed as matrices and traced. We included a factor of $\hbar^{-1}$ because of the disk topology. 

Essentially by definition of the cyclic cohomology complex, the action of $Q_0$ on $\Phi$ is intertwined with the action of the differential $Q_\CC$ on $\fc$. We have thus gained an immediate identification:
\begin{itemize}
    \item The cohomology of single-trace local operators coincides with $\HC^\bullet(\bC[\theta^1,\theta^2,\zeta_3])$.
    \item The identification encodes the coupling of a closed string state to the stack of D-branes. 
\end{itemize}
The closed string state is a divergence-free polyvector field which is distributional in the $\zeta_3$ direction and couples to the stack of D1 branes at a point via a specific single-trace operator built from the world-volume fields on the D-brane. 

In the BCOV description of the B-model string theory, the distributional holomorphic polyvector field is mapped to a form $\alpha$ and $\partial^{-1} \alpha$ is restricted to the D-brane world-volume and coupled to the fields there. The space of divergence-free polyvector fields in $\HC[\bC[\theta^1,\theta^2,\zeta_3]]$ reproduce the single-trace operators in the four towers together with their derivatives:
\begin{itemize}
    \item An $\cA_{a,b}(z)$ single-trace operator 
is induced from  
\begin{equation}
    \partial^{-1} \alpha= \zeta_1^a \zeta_2^b \delta(\zeta_3-z) d\zeta_3 
\end{equation}
i.e. a divergence-free polyvector field 
\begin{equation}
    \beta= \delta(\zeta_3-z)  (a \zeta_1^{a-1} \zeta_2^b \partial_{\zeta_2}- b \zeta_1^{a} \zeta_2^{b-1}\partial_{\zeta_1}) \,.
\end{equation}
    \item An $\cB_{a,b}(z)$ single-trace operator is induced from 
\begin{equation}
    \alpha= \zeta_1^a \zeta_2^b \delta(\zeta_3-z) d\zeta_1 d\zeta_2 d\zeta_3 \,.
\end{equation}
\item An $\cC_{a,b}(z)$ single-trace operator is induced from 
\begin{equation}
    \partial^{-1} \alpha= \zeta_1^a \zeta_2^b \delta(\zeta_3-z)\,.
\end{equation}
\item A $\cD_{a,b}(z)$ single-trace operator is induced from 
\begin{equation}
    \alpha= \zeta_1^a \zeta_2^b \delta(\zeta_3-z)d\zeta_1 d\zeta_2 + (\cdots)\delta'(\zeta_3-z)d\zeta_3\,,
\end{equation}
where the ellipses denotes an appropriate 1-form to make it $\partial$-closed. There is a mixing with $\partial \cA_{a,b}(z)$ which can be resolved by imposing a quasi-primary condition.
\end{itemize}
In particular, this proves that the four towers exhaust the single-trace cohomology! 

Another route to produce representatives of $Q_0$ cohomology classes is to define the extended algebra $\bC[\theta^1,\theta^2,\zeta_3,d\theta^1,d\theta^2,d\zeta_3]$, equipped with the de Rham operator $d$. We can thus define the generating field $d\Phi$, which transforms as 
\begin{equation}
    Q_0 d\Phi = [\Phi, d\Phi]\,.
\end{equation}
Accordingly, 
\begin{equation}
    \frac{1}{\hbar n} \Tr (d\Phi)^n
\end{equation}
is $Q_0$-closed. It is also $d$-closed, and we can expand it into a basis of closed forms in the $\theta^1$,$\theta^2$ and $\zeta_3$ variables. The coefficients of the expansion will be a basis of $\HC^\bullet(\bC[\theta^1,\theta^2,\zeta_3])$.

\subsection{The Global Symmetry Algebra at tree-level}
At tree-level, the modes in the global symmetry algebra $\fL_0$ act on other single-trace operators by a single Wick contraction and thus by mapping an adjoint field to a sequence of adjoint fields:
\begin{equation}
    o:\quad \Phi \to \{o|\}_0 +\{o|\Phi\}_1 + \{o|\Phi,\Phi\}_2+ \cdots  \, .
\end{equation}
For example, the modes of $\frac{1}{n \hbar}\Tr X^n$ acts as 
\begin{equation}
   \left(\frac{1}{n \hbar}\Tr X^n\right)_k:  \quad Y(\zeta_3) \to \zeta_3^k X(\zeta_3)^{n-1}
\end{equation}
i.e. the maps
\begin{equation}
   \left\{\left(\frac{1}{n \hbar}\Tr X^n\right)_k | \theta_1 \zeta_3^{k_1}, \cdots, \theta_1 \zeta_3^{k_{n-1}}\right\}_{n-1} = \theta_2 \zeta_3^{k_1+ \cdots + k_{n-1} + k}\,.
\end{equation}
This type of transformations define elements in the $\HH^{\bullet}(\bC[\theta^1,\theta^2,\zeta_3],\bC[\theta^1,\theta^2,\zeta_3])$ and can thus be directly compared to normalizable local operators in the world-sheet theory (B-model on $\bC^3$), i.e. to polynomial polyvector field in $\bC^3$. 

A natural perspective on this is that a mode of $\fL_0$ could be added to the BRST charge of the chiral algebra, leading to a deformation of the BRST transformation of $\Phi$ and thus of the $A_\infty$-algebra structure on $\bC[\theta^1,\theta^2,\zeta_3]$, which is an element of $\HH^{\bullet}(\bC[\theta^1,\theta^2,\zeta_3],\bC[\theta^1,\theta^2,\zeta_3])$.

We expect $\fL_0$ to actually correspond to the $\lambda \to 0$ limit of divergence-free holomorphic polynomial polyvector fields in $SL(2,\bC)$. These are the same as divergence-free holomorphic polynomial polyvector fields in $\bC^3$ which satisfy a certain growth condition at large $\zeta_3$. The latter condition can be removed by looking at all non-negative modes of single-trace operators, which define a Lie algebra at $\lambda \to 0$. The restriction to the modes to $\fL_0$ can be expressed geometrically by promoting $\zeta_3$ to a $\bC P^1$ coordinate and the whole geometry to $\cO(-1) \oplus \cO(-1) \to \bC P^1$ with coordinates $\zeta_1$, $\zeta_2$, $\zeta_1 \zeta_3$ and $\zeta_2 \zeta_3$. This is the natural geometry where $N$ D1 branes would reproduce the chiral algebra on the sphere.

The divergence-free condition is trickier to understand. It must be associated to the fact that not all possible deformations of the BRST charge should be expressible as modes of single-trace operators. We do not have a good Homological Algebra understanding of this condition beyond checking that it is satisfied by the images of the maps
\begin{equation}
    \HC^\bullet(\bC[\theta^1,\theta^2,\zeta_3]) \to \HH^\bullet(\bC[\theta^1,\theta^2,\zeta_3],\bC[\theta^1,\theta^2,\zeta_3])
\end{equation}
given by taking the modes of single-trace operators. We leave this question as an open problem.

\subsection{The mesons at $\lambda \to 0$.}
Neumann branes $P$ in $\bC^3$ have junctions to $B$ described by $\bC[\zeta_3]$. We identify the (anti)fundamental matter fields $I(\zeta_3)$ and $J(\zeta_3)$ we introduced before as world-volume fields associated to these boundary-changing local operators. 

As discussed at greater length in the next section, the $Q_0$ cohomology of mesonic operators is dual to the derived tensor product 
\begin{equation}
    \bC[\zeta_3]\otimes_{\bC[\theta^1,\theta^2,\zeta_3]} \bC[\zeta_3] \simeq \bC[\zeta_1,\zeta_2,\zeta_3]\,.
\end{equation}
It is identified with a space of local operators on $P$ which can enter a disk correlation function with a $B$ segment, i.e. functions of the form $\zeta_1^a \zeta_2^b\partial^n \delta(\zeta_3)$. This matches the $\cM$ tower of mesons.

Modes of the open symmetry algebra $\fP_0$ map naturally into the (derived) endomorphisms of $\bC[\zeta_3]$ as a $\bC[\theta^1,\theta^2,\zeta_3]$
module, and coincide with polynomial vertex operators $\bC[\zeta_1, \zeta_2, \zeta_3]$ on $P$. The restriction to modes in the correct range gives holomorphic functions which extend to $\cO(-1) \oplus \cO(-1) \to \bC P^1$. 

The action of $\fL_0$ onto $\fP_0$ gives a perspective on matching $\cL_0$ to polyvector fields which we have seen extends nicely to non-zero $\lambda$: 
\begin{enumerate}
    \item Modes of operators in the $\cB_{a,b}$ tower are directly mapped by the mesonic part of the BRST differential to modes in $\fP_0$.
    \item Modes of operators in the $\cA_{a,b}$ and $\cD_{a,b}$ towers act as derivations on $\fP_0$.
    \item Modes of operators in the $\cB_{a,b}$ tower added to the BRST differential will modify the product structure constants of $\fP_0$. 
\end{enumerate}
These three statements map $\cL_0$ into the Hochschild cohomology $\HH^{\bullet}(\fP_0,\fP_0)$.

\subsection{Determinants at $\lambda \to 0$.}
The basic giant graviton brane at $\lambda \to 0$ is a probe brane $D$ which has Dirichlet b.c. $\zeta_1 |_\partial= \zeta_3|_\partial=0$ and Neumann $\theta^2|_\partial=0$. The junctions to $B$ are controlled by $\bC[\theta^1]$. The $m$ parameter can be introduced by setting $\zeta_1 |_\partial=-m$ instead, and $u$ by fixing a linear combination $(\zeta_1 + u \zeta_2)|_\partial$.

As we have seen, the integral defining determinant operators can be presented in a BV formalism. The auxiliary fermions $\psi$ and $\bar \psi$ and their anti-fields $u$ and $\bar u$ enter in a BV action (\ref{eq:detBV}) which generalizes $\bar \psi X \psi$. The BRST differential is consistent with the identification of $\psi + \theta^1 u$ and $\bar \psi + \theta^1 \bar u$ as the open string fields stretched between the branes. \footnote{Several aspects of the construction and computation of open modifications described in this Section originally emerged in unpublished work with Kasia Budzik.}

The computation of the tree-level BRST cohomology of determinant modifications reduces to the computation of the dual to the derived tensor product
\begin{equation}
    \bC[\theta_1] \otimes_{\bC[\theta^1,\theta^2,\zeta_3]}\bC[\theta_1] \,,
\end{equation}
which coincide with the space $\bC[\theta^1, \zeta^2, \theta_3]$ of polynomial boundary local operators on $D$. The $\zeta^2$ variable is clearly dual to $Y$ insertions. An explicit description of the cohomology of determinant modifications goes beyond the scope of our discussion.

Open modifications, instead, are dual to 
\begin{equation}
    \bC[\zeta_3] \otimes_{\bC[\theta^1,\theta^2,\zeta_3]}\bC[\theta_1] \,,
\end{equation}
which is the space $\bC[\zeta_2]$ of junctions between $P$ and $D$. These reproduce the $I Y^n \psi$ open modifications of determinant operators we employed in explicit calculations. We now give a dg-TFT interpretation of these calculations.

In order to study a non-trivial saddle for a collection of determinant operators, we would start from a collection of branes 
\begin{align}
    \zeta_3|_\partial &= z_i \cr
    \zeta_1|_\partial + u_i \zeta_2|_\partial &=0\,.
\end{align}
We can describe each D-brane as a deformation of 
a basic D-brane $D$ by 
\begin{equation}
   \cI^\partial_{z,u} =  z_i \theta_3 + u_i \zeta_2 \theta_1\,.
\end{equation}
Turning on a general $\rho$ corresponds to a further boundary interaction $\cI^\partial_\rho =  \rho \theta_1$. The separation of the D-branes in the $\zeta_3$ direction obstructs that, via a BRST anomaly
\begin{equation}
    \{\cI^\partial_{z,u},\cI^\partial_\rho \}= (z_i - z_j) \rho_{ij} \theta_1 \theta_3\,.
\end{equation}
Once we turn on $\lambda$, we expect this anomaly will cancel against an extra $\lambda(u_i-u_j) (\rho^{-1})_{ij}$, leading to the saddle equations. 
This can be made concrete by computing the deformation of the $\bC[\theta^1, \zeta^2, \theta_3]$ algebra due to the bulk back-reaction,
e.g. by computing the planar corrections to the space of modifications or, more indirectly, as we did originally: compute the deformation of $\bC[\zeta_2]$ to a $\fP_\lambda$ module and use it to describe the deformation of $D$. 

\section{Two-dimensional chiral gauge theories at large $N$}\label{sec:HAcalculations}
Formally, a two-dimensional chiral gauge theory is defined by coupling a matter chiral algebra with Kac-Moody symmetry $G$ to a 2d chiral gauge field, i.e. a gauge field which only has a $(0,1)$ form component. Upon gauge-fixing, this definition results in a 2d chiral algebra presented as the cohomology of a certain BRST complex we discuss below. The BRST complex is well-defined only if the matter Kac-Moody currents have a specific level which cancels a one-loop gauge anomaly.

We have reviewed the supersymmetric chiral $SU(N)$ gauge theory defined by taking bosonic matter in two copies of the adjoint representation.
This is a protected sub-sector of four-dimensional ${\cal N}=4$ $SU(N)$ gauge theory and also the world-volume theory of $N$ D$1$ branes in the $\bC^3$ B-model, up to the decoupled $U(1)$ center-of-mass degrees of freedom. 

Four-dimensional ${\cal N}=2$ SCFTs also have protected subsectors \cite{Beem:2013sza,Beem:2017ooy,Bonetti:2016nma,Dedushenko:2019mnd}. Four-dimensional gauge theories with gauge group $G$
and matter transforming in a symplectic representation $R$ have a protected subsector consisting of a 2d chiral gauge theory with gauge group $G$ and bosonic matter in representation $R$. 

The anomaly cancellation condition for 4d SCFTs or 2d chiral gauge theories with bosonic matter only is rather restrictive. 
For example, quiver gauge theories with special unitary gauge group must be modelled on affine ADE quivers. The corresponding ``ADE'' chiral gauge theories appear on the world-volume of 
$N$ D$1$ branes in a B-model with target space $\bC \times \frac{\bC^2}{\Gamma}$, up to $U(1)$ factors in the gauge group. Twisted Holography relates such ADE chiral gauge theories to the B-model on $SL(2,\bC)/\Gamma$.

The definition of 2d chiral gauge theory allows the introduction of fermionic matter fields as well, transforming into an orthogonal representation $R_f$ of $G$. Fermions and bosons contribute to the anomaly with opposite signs and thus the choice of gauge group and matter representation is much less constrained. For example, we could (and will) consider an $SU(N)$ gauge theory with $2n+2$ adjoint bosons and $2n$ adjoint fermions. 

The theories with fermionic matter do not appear to be protected sectors of 4d SUSY theories. Any statement we may derive about the 't Hooft expansion of such theories will not encode a protected part of a standard holographic correspondence. We are also not aware of 3d CY geometries such that D$1$ branes would support such chiral theories. A 't Hooft analysis of general 2d chiral gauge theories will thus likely lead us to unexplored corners of String Theory, {\it if} 't Hooft completeness holds for this class of gauge theories.

\subsection{The large $N$ expansion of 2d chiral gauge theories}

For conciseness, in the remainder of this section we will take the fields to be a collection of $N \times N$ matrices, which could be organized further into adjoint or bifundamental representations of one or more $U(k_i N)$ groups with $k_i \in \bZ$. There are important differences between $U(N)$ and $SU(N)$, but they are immaterial in the planar limit. Anomaly cancellation may also require the addition of order $1$ fields which transform as $SU(N)$ scalars. These are also immaterial in the planar limit. (Anti)fundamental degrees of freedom will be discussed separately.

This assumption could be easily relaxed to allow for more general ranks $N_i$, with minimal changes to our formulae below. Generalizations to $SO(k_i N)$ and $Sp(2 k_i N)$ gauge groups are also possible, as well as matter in various two-index representations of the gauge groups, but require some considerations about unorientable worldsheets. In these cases, local operators are no longer computed by cyclic cohomology but instead by the so-called Dihedral cohomology \cite{LODAY198893,loday2013cyclic}. We briefly discuss these cases in Appendix \ref{appendix:ortho_symp}. 

We organize correlation functions, OPEs, etc. in a 't Hooft expansion just as we did in the standard example. 

\subsection{A hidden algebra}
We denote the Grassmann parity of a symbol $x$ as $|x|$ and its ghost number as $\mathrm{gh}[x]$. In the absence of free fermions, the Grassmann parity of fields coincides with the ghost number modulo $2$. If free fermions are present, we instead need to allow the Grassmann parity to be distinct from the ghost number and thus work with graded super vector spaces. 

The free fields we work with include: 
\begin{itemize}
    \item A collection of $bc$ systems with scaling dimensions $\Delta_c =0$ and $\Delta_b=1$ and ghost numbers $1$ and $-1$ respectively. Both sets of fields are fermionic,\footnote{It may be possible to extend the formalism to include super-groups. The ghosts corresponding to fermionic generators in $G$ would then be bosons.} i.e. $|c|=|b|=1$. 
    \item A collection of symplectic bosons and free fermions with scaling dimension $\frac12$ and ghost number $0$. We denote them collectively as $Z$.  
\end{itemize}

We now discuss a crucial observation: the entire field content and BRST symmetry of a 2d chiral gauge theory built from $N \times N$ matrices of free fields can be encoded into a {\it 2d-cyclic \footnote{Here, $2d$ refers to the degree of the cyclic pairing, not the dimension of the algebra}, finite-dimensional, graded associative super-algebra} $A$. Vice versa, any such algebra defines a 2d chiral gauge theory of $N \times N$ matrices at tree level. An anomaly cancellation condition is required at one loop.

The super vector space $A$ can be introduced as a way to package all of the fields into a single generating field $\Phi$, an $N \times N$ matrix valued in $A$ with $|\Phi|=\mathrm{gh}[\Phi]=1$:
\begin{equation}
	\Phi(z) \equiv a_{0,u} c^u(z) + a_{1,\alpha} Z^\alpha(z) + a_{2}^u b_u(z)\,.
\end{equation}
Here the $u$ and $\alpha$ indices run over the collections of $bc$ ghosts, symplectic bosons and free fermions defining the chiral algebra under consideration. Accordingly, we denoted as $a_{0,u}, a_{1,\alpha}, a_{2}^u$ a basis of 
\begin{equation}
	A = A_0 \oplus A_{1} \oplus A_{2}\,,
\end{equation} 
where $A_i$ are the components of $A$ of ghost number $i$. The scaling dimension of different components of $\Phi$ can be encoded in an operator $\Delta$ acting on $A_i$ as $\frac{i}{2}$. 
 
We can also denote individual components of $\Phi$ collectively as $\phi^a \in \mathfrak{gl}(N)$ and the basis elements of $A$ as $a_a$:
\begin{equation}
    \Phi(z) = a_a \phi^a\,.
\end{equation}
In practice, we will do our best to minimize any references to individual component of $\Phi$ except in examples. Working with the generating field $\Phi$ has considerable conceptual and practical advantages.

The OPE of elementary fields can be written concisely as 
\begin{equation}\label{eq_al:ope_abstract}
	\Phi^i_j(z) \otimes \Phi^k_t(w) \sim  \delta^i_t \delta^k_j\hbar \frac{\eta}{z-w}\,,
\end{equation}
where we wrote explicitly the $U(N)$ indices $i,j,k,t$. In the following we will leave $U(N)$ indices implicit when possible.

The numerator $\eta \in A \otimes A$ is a graded-symmetric tensor which collects the two-point functions. It has non-zero components
\begin{equation}
	\eta_u^v = \delta_u^v \qquad \qquad \eta^{\alpha \beta} = \omega^{\alpha \beta}\,.
\end{equation}
Expanding out the concise OPE, we recover the familiar OPE of 
a collection of $bc$ systems and symplectic bosons/free fermions:
\begin{align}
    b_u(z) c^v(w) &\sim \hbar\frac{\delta_u^v}{z-w} \cr
    Z^\alpha(z) Z^\beta(w) &\sim \hbar\frac{\omega^{\alpha \beta}}{z-w}\,.
\end{align}
We could also write 
\begin{equation}\label{eq_al:ope_with_ab_indices}
	\phi^a(z) \phi^b(w) \sim \hbar \frac{\eta^{ab} }{z-w}
\end{equation}
with $\eta = \eta^{ab} a_a \otimes a_b$.

In concrete OPE calculations, we will encounter expressions where $\eta^{ab}$ is contracted with pairs of $A$ basis elements scattered through the expression. We find it useful to borrow the Sweedler notation from the theory of Hopf algebras and write $\eta = \eta^{(1)} \otimes \eta^{(2)}$ as a stand-in for the full expansion in a basis for the tensor product. If $l$ number of contractions occur, we use pairs $\eta_i^{(1)} \otimes \eta_i^{(2)}$ with $i = 1, \cdots, l$ to keep track of the different contractions. 

We denote the (graded symmetric) pairing dual to $\eta$ simply as $(a_a a_b) \in \bC$, so that 
\begin{align}
    \eta^{(1)}(\eta^{(2)} a) &= a \cr
    (a \eta^{(1)})\eta^{(2)} &= a
\end{align}
for all $a \in A$.

We can now write concise expressions for the ghost number current and the stress tensor. Notice that $\Tr(\Phi \Phi)\equiv (-1)^{|\phi^b||a_a|} \, \Tr \phi^a \phi^b (a_a a_b) = 0$ because the symmetry properties of $(a_a a_b)$, when non-zero, are opposite to these of $\phi^a \phi^b$. 

The ghost number current can be written as 
\begin{equation}
    J_{\mathrm{gh}} = \frac{1}{\hbar}\Tr(\Phi \Delta \Phi) = \frac{1}{\hbar} \Tr c^u b_u\,.
\end{equation}
In particular, 
\begin{equation}
    J_{\mathrm{gh}}(z) \Phi(w) \sim \frac{1-2\Delta}{z-w}\, \Phi(w)\,.
\end{equation}
We can also write the Stress Tensor $T(z)$ as 
\begin{equation}
    T = \frac{1}{2\hbar}\Tr(\partial \Phi \Phi) +\frac12 \partial J_{\mathrm{gh}}(z)=\Tr(\partial \Phi \Delta \Phi)\,.
\end{equation}

\subsection{An algebra structure from the BRST differential}

We will discuss the BRST current momentarily. At first, we can focus on the tree level part $Q_0$ of the BRST transformations, i.e. the part involving a single Wick contraction. The action of $Q_0$ maps a field to a sum of (matrix) products of fields:
\begin{align}\label{eq:al_tree-level_BRST}
	Q_0 c^u &= f^u_{vw} c^v c^w \cr
	Q_0 Z^\alpha &= f^\alpha_{v \beta} \left[c^v Z^\beta-Z^\beta c^v \right] \cr
	Q_0 b_u &= f^w_{vu} \left[c^v b_w+b_w c^v \right] + f_{u\alpha \beta} Z^\alpha Z^\beta\,.
\end{align} 
It is easy to see that the structure constants on the right hand side equip $A$ with the structure of an associative algebra. The $c$ ghost for the diagonal $U(N)$ gauge action equips $A$ with an unit. 

The algebra structure preserves the weight and ghost number. It allows us to write a simple transformation rule 
\begin{equation}
	Q_0 \Phi(z) = \Phi(z) \Phi(z)
\end{equation}
extended by the Leibniz rule to products of fields. Associativity is closely related to $Q_0^2=0$ (remember that $\Phi$ is fermionic):
\begin{equation}
	Q_0^2 \Phi(z) = (\Phi(z) \Phi(z))\Phi(z) - \Phi(z)(\Phi(z) \Phi(z))\,.
\end{equation}

The BRST current is a cubic expression in the elementary fields. It has a very concise expression
\begin{equation}
    J_{\mathrm{BRST}} = \frac{1}{3 \hbar} \Tr (\Phi\Phi\Phi)\,.
\end{equation}
Here we denote as $(\bullet)$ a linear map $A \to \bC$ such that 
the composition $(\bullet \bullet)$ with the product on $A$ coincides with the pairing dual to $\eta$. In particular, $(\bullet)$ is a graded trace supported on $A_2$. 

We will denote an associative algebra equipped with a trace with these properties as a 2d-cyclic algebra. Its relationship with the standard definition of Calabi-Yau algebra is reviewed in Appendix \ref{appendix:CY}. Conversely, any such an algebra $A$ can be used to define a free chiral algebra equipped with a BRST differential of this form. 

The full BRST differential acting on general local operators includes both a $Q_0$ term with a single Wick contraction and a 1-loop term with two Wick contractions. There is a potential 1-loop BRST anomaly which further constrains the form of $A$.

\subsection{Back to $\bC^3$}
As an example, we consider the case of the supersymmetric chiral gauge theory with gauge group $U(N)$. The collection of matrix-valued fields consists of a single $bc$ system and a single set of symplectic bosons $X$, $Y$. These can be collected into a generating field
\begin{equation}
	\Phi(z) = c(z) + \theta_1 X(z) + \theta_2 Y(z) + \theta_1 \theta_2 b(z)  
\end{equation}
valued in the algebra $A = \bC[\theta_1, \theta_2]$ of polynomials in two anti-commuting fermionic variables $\theta_\alpha$. 

We have already encountered this parametrization in Section \ref{sec:HAcalculations}. It identifies $\Phi(z)$ with the open string field for the stack of D$1$ branes supported on $\bC \in \bC^3$. 
In particular, the algebra $A$ is simply the algebra of boundary local operators for the Dirichlet boundary conditions  in the transverse directions. 

The scaling dimension operator can be written as 
\begin{equation}
	\Delta = \frac12 \theta_\alpha \partial_{\theta_\alpha}\,.
\end{equation}
We can write the pairing as 
\begin{equation}
	\eta = (\theta_1 - \theta_1')(\theta_2 - \theta_2')\,,
\end{equation} 
where the unprimed and primed variables denote the two factors in the tensor product $A \otimes A$, i.e. we identified 
\begin{equation}
    \bC[\theta_1, \theta_2] \otimes \bC[\theta_1, \theta_2] = \bC[\theta_1, \theta_2,\theta'_1, \theta'_2]\, .
\end{equation} 
The corresponding trace is 
\begin{equation}
    (\theta_1 \theta_2)=1 \, ,
\end{equation}
and $0$ otherwise.

\subsection{The ADE chiral algebra and the B-model}
Recall that the ADE quiver has nodes labelled by the representations of the discrete group $\Gamma$ and edges controlled by the tensor product with the fundamental representation. For example, for $\Gamma = \mathbb{Z}_k$ we have one-dimensional representations $R_i$ with $0\leq i<k$ modulo $k$  
and a necklace quiver. 

The ADE gauge theory has ranks equal to the dimensions $\mathrm{dim}R_i$ and matter fields $(X_e, Y_e)$ for each edge $e$. It is easy to recognize that 
\begin{equation}
    A = \mathrm{End}(\bC[\Gamma])[\theta^1, \theta^2]^\Gamma
\end{equation}
coincides with the algebra of local operators on $D_{\mathrm{tot}}$.

Indeed, the ADE chiral algebra is the world-volume theory of $N$ $\bC \times D_{\mathrm{tot}}$ branes in the B-model with target 
\begin{equation}
    \bC \times \frac{\bC^2}{\Gamma}\, .
\end{equation}

The trace on $A$ is simply the matrix trace combined with the 
trace on $\bC[\theta^1, \theta^2]$.

\subsection{A small generalization}
The notion of 2d-cyclic associative algebra can be generalized to 
that of 2d-cyclic $A_\infty$ algebra. Schematically, we may imagine a very general tree-level BRST transformation rule: 
\begin{equation}
    Q_0 \Phi = \{\Phi\} +  \{\Phi,\Phi\}+  \{\Phi,\Phi,\Phi\}+ \cdots
\end{equation}
where 
\begin{equation}
    \{\bullet, \cdots, \bullet\}: A^{\otimes n} \to A
\end{equation}
are multi-linear maps which change the overall ghost number by $2-n$. These maps generalize the associative product encountered in the rest of this section. Essentially by definition, they equip $A$ with the structure of an $A_\infty$ algebra. The existence of a BRST current 
\begin{equation}
   J_{\mathrm{BRST}} =\frac12 \Tr(\{\Phi\}\Phi) +  \frac13\Tr(\{\Phi,\Phi\}\Phi)+ \frac14 \Tr(\{\Phi,\Phi,\Phi\}\Phi)+ \cdots
\end{equation}
such that structure constants are cyclic symmetric make $A$ into a 2d-cyclic $A_\infty$ algebra. We also note, following \cite{Kontsevich2008NotesOA}, that a 2d-cyclic algebra is equivalent to a 2d Calabi-Yau algebra, which can be used to define an abstract dg-TFT.

\subsection{Algebras and branes}
Suppose now that we are given some abstract dg-TFT $T_2$ 
with a ghost number anomaly of $2$ and a D-brane $B$ with a finite-dimensional boundary (possibly $A_\infty$) algebra $A$ which admits a trace, i.e. gives finite disc correlation functions, and has the correct scaling properties. We have seen how such a D-brane can probe normalizable local operators and a category of D-branes in $T_2$ via Homological Algebra constructions applied to $A$.

We can combine $T_2$ and the B-model with target $\bC$ to make a world-sheet theory suitable to define a B-model-like String Theory. We can consider a stack of $N$ D-branes of the form $B \times \bC$ in that String Theory. By construction, the world-volume theory of such D-branes can be identified with the 2d chiral gauge theory we associated to $A$.

In such a situation, we would expect the 't Hooft expansion of the 2d chiral gauge theory to be dual to a modified String Theory, deformed by the back-reaction of the $N$ branes.

The question we explore in the rest of the paper is: can we characterize this back-reaction algebraically, even if $T_2$ does not have a sigma-model interpretation and BCOV theory is not available? Even better, 
can we somehow {\it define} the deformed String Theory if all we have is an algebra $A$ with the correct properties, perhaps by giving a dg-TFT description of the corresponding world-sheet theory? 

\section{Local operators at tree level and cyclic cohomology} \label{sec:closed}
General local operators are built as normal-ordered polynomials in the fields and their derivatives. The action of the tree-level differential $Q_0$ \eqref{eq:al_tree-level_BRST} does not change the number of derivatives present in a monomial. 
A useful warm-up is to consider the $Q_0$-cohomology of single-trace local operators which do not contain derivatives, analogous to the $\cA$ and $\cB$ towers in the canonical example. We will then characterize the whole cohomology of $\mathrm{Obs}_0$.

\subsection{The first tower}
\label{sec:first_tower}
We will now introduce a useful notation which allows us to express all calculations in terms of the 2d-cyclic algebra $A$. Consider an expression of the form 
\begin{equation}
\label{eq:towerone}
	\cO_\fc(z) \equiv \frac{1}{\hbar \ell(\fc)} \Tr (\fc|\Phi(z), \cdots, \Phi(z))\,,
\end{equation}
where $\fc$ denotes a cyclic-symmetric (with signs) multi-linear map 
\begin{equation}
	(\fc|\bullet, \cdots, \bullet): \left(A^{\otimes \ell(\fc)}\right)^{\mathbb{Z}_{\ell(\fc)}} \to \bC \,,
\end{equation} 
and $\ell(\fc)$ is the number of inputs in $\fc$.\footnote{The signs insure compatibility with the cyclicity of the trace. If we rotate the trace and bring the last entry to the beginning, we pay a Koszul price from passing the fields $\phi^a$ across each other. The Koszul parity of the $\phi^a$ is opposite to the Koszul parity of the $a_a$ elements, so we get a $-1$ factor for each pair of  {\it bosonic} elements in $A[[s]]$. 
 } By linearity, we can expand 
\begin{equation}
	\cO_\fc(z) \equiv \sum_{a_1, \cdots a_{\ell(\fc)}}\pm \frac{1}{\hbar |\fc|}  (\fc|a_{a_1}, \cdots, a_{\ell(\fc)})\Tr \phi^{a_1} \cdots \phi^{a_{\ell(\fc)}}(z)
\end{equation}
and recognize the matrix elements of $\fc$ as coefficients of a generic linear combination of single-trace local operators. The overall factor of $(\hbar \ell(\fc))^{-1}$ is introduced for later convenience. 
As these operators do not contain derivatives, they are manifestly quasi-primary operators in the chiral algebra. 

We can easily compute the action of $Q_0$: 
\begin{equation}
\label{eq:toweroneQ0}
	\cO_{Q_0\,\fc}(z) \equiv Q_0\,\cO_{\fc}(z) = \frac{1}{\hbar} \Tr (\fc|\Phi(z)\Phi(z), \cdots, \Phi(z))\, .
\end{equation}
Symmetrizing carefully, 
\begin{align}
\label{eq:cyclic}
	&(Q_0\,\fc|a_1,\cdots, a_{n+1}) = (\fc|a_1 a_2, \cdots, a_{n+1}) - (\fc|a_1,a_2 a_3, \cdots, a_{n+1})+ \cr 
	&+ (\fc|a_1,a_2, a_3 a_4, \cdots, a_{n+1}) + \cdots \pm (-1)^{n} (\fc|a_{n+1} a_1,a_2, \cdots, a_n)\,.
\end{align}
We recognize the differential defining the {\it cyclic cohomology} complex $\CC^\bullet(A)$ for $A$. This tower of local operators is thus labelled by classes in the cyclic cohomology $\HC^\bullet(A)$.

There is a small subtlety which we should address here. The space of local operators in the gauge theory should be built as the {\it relative} BRST cohomology: the ghost $c$ is only allowed to appear in local operators through its derivatives and $G$-invariance is imposed by hand. A naive calculation which ignores this point will produce some extra cohomology classes of scaling dimension $0$ built as polynomials in the $c$ ghosts, as well as the derivatives of these classes. That extra cohomology can be removed by hand, as it is the only cohomology in the sector with scaling dimension $0$.\footnote{If $A_0$ consists of the identity $1_A$ only, i.e. the gauge group is $U(N)$, a relative cohomology calculation only requires us to restrict to functions $\fc$ which vanish on $1_A$. This restriction defines the relative cyclic cohomology complex $\CC_{\mathrm{rel}}^\bullet(A)$. In a more general situation, the correct procedure would be to promote $A$ from an algebra to a category whose objects label individual gauge groups, as we saw in the ADE example. Then the definition of relative cyclic cohomology for a category automatically keeps track of the requirement of $G$-invariance. 
We leave this generalization implicit for conciseness. }

\subsection{The second tower}
\label{sec:second_tower}
Next, we can look at operators involving a single derivative, analogous to the $\cC$ and $\cD$ towers in the standard example (and first derivatives of the other two):
\begin{equation}
\label{eq:towertwo}
	\cO_\fh(z) \equiv \frac{1}{\hbar} \Tr (\fh|\partial \Phi(z);\Phi(z), \cdots,\Phi(z))\,,
\end{equation}
where the multilinear map $\fh$ is not cyclic symmetric. As
\begin{equation}
	Q_0 \partial \Phi = \Phi \partial \Phi + \partial \Phi \Phi
\end{equation}
we have 
\begin{align}
\label{eq:towertwoQ0}
	\cO_{Q_0\,\fh}(z) &\equiv Q_0\,\cO_{\fh}(z) =  \frac{1}{\hbar} \Tr (\fh|\Phi \partial \Phi, \cdots, \Phi(z)) +  \frac{1}{\hbar} \Tr (\fh|\partial \Phi \Phi, \cdots, \Phi(z)) + \cr &-\frac{1}{\hbar} \Tr (\fh|\partial \Phi, \Phi(z)\Phi(z),\cdots, \Phi(z)) + \cdots \, .
\end{align}
i.e. 
\begin{align}
\label{eq:hoch}
	&(Q_0\,\fh|a_1;\cdots, a_{n+1}) = (\fh|a_1 a_2; \cdots, a_{n+1}) - (\fh|a_1;a_2 a_3, \cdots, a_{n+1})+ \cr 
	&+ (\fh|a_1,a_2, a_3 a_4, \cdots, a_{n+1}) + \cdots \pm (-1)^{n} (\fh|a_{n+1} a_1,a_2, \cdots, a_n)\,.
\end{align}
Notice that this differential is identical in form to the one we wrote for the cyclic cohomology complex,
but it acts here on maps which are not cyclic invariant. It defines the {\it Hochschild cohomology} complex $\CH^\bullet(A,A^\vee)$ valued in the dual $A^\vee$ of $A$.\footnote{As the pairing identifies $A$ and $A^\vee$, this is essentially the same as the standard Hochschild cohomology complex $\CH^\bullet(A)$. The physical meaning of the latter, though, is slightly different.} 
This tower of local operators is thus labelled by classes in the Hochschild cohomology $\HH^\bullet(A,A^\vee)$. Again, we can avoid the issue of relative vs absolute BRST cohomology by 
restricting to cohomology classes of scaling dimension greater than $1$.

Some of the operators we have identified are actually derivatives of operators in the first tower. The operation of taking a derivative, i.e. the $L_{-1}$ Virasoro generator, gives a standard morphism $I: \CC^\bullet(A)\to \CH^\bullet(A,A^\vee)$, which embeds the space of cyclic maps into all possible maps. 

We can look for quasi-primary operators of the form $\cO_\fh(z)$ by looking at the action of the $L_1$ Virasoro generator\footnote{Here we mean the action of the $L_1$ centered at the location of the field, $L_1 \Phi(z) = \oint dw (w-z)^2 T(w) \Phi(z)$, as is usual in determining if an operator is (quasi-)primary.}. At tree level, $L_1$ simply maps $\partial \Phi \to -2\Delta \Phi$ (and $\Phi \to 0$). Accordingly, 
\begin{equation}
	L_1 \cO_\fh(z) = -\frac{2}{\hbar} \Tr (\fh|\Delta\Phi(z);\Phi(z), \cdots,\Phi(z))\,.
\end{equation}
We thus encounter a map $\CH^\bullet(A,A^\vee)\to \CC^\bullet(A)$ which composes $\fh$ with $(-2\Delta)$ at the first argument and then applies a  (graded) cyclic symmetrization to the result. The kernel of this map gives the space of quasi-primary operators in the second tower. 

Operators of the form $\mathcal{O}_{\mathfrak{h}}$ should either be a derivative of the first tower or generate a new Verma module (a quasi primary). We conclude that we have an equivalent characterization of the quasi-primary operators in the second tower, as the quotient $\HH^\bullet(A,A^\vee)/I(\HC^{\bullet}(A))$.

The morphism $I$ is part of Connes Periodicity long exact sequence. The long exact sequence also involve certain ``periodicity maps'' $S$ which controls the kernel of $I$. The map $L_1$ above certifies that $S$ is trivial as long as we ignore operators of scaling weight $0$ and thus the long exact sequence collapses to a collection of short exact sequences, identifying the quotient $\HH^n(A,A^{\vee})/\HC^n(A)$ with $\HC^{n-1}(A)$ by Connes $B$ operator. \footnote{More precisely, the degree (scaling dimension) of cyclic cochain induced from the degree of $A$ is preserved by the differential. We can split the cyclic cohomology according to the degree
	\begin{equation}
		\HC^{\bullet}(A) = \bigoplus_{w \geq 0}\HC^{\bullet}(A)^{(w)}\,.
	\end{equation}
	In particular, the degree zero part $\HC^{\bullet}(A)^{(0)}$ is the same as the cyclic cohomology for $A_0$, $\HC^{\bullet}(A)^{(0)} = \HC^{\bullet}(A_0)$. 
	
	We show in Appendix \ref{appendix:vanish_S} that the Connes' periodicity map $S$ vanishes on the  positive degree part of cyclic cohomology $\HC^{\bullet}(A)^{(\geq 1)}$. Therefore, the Connes' long exact sequence reduces into a collection of short exact sequences
	\begin{equation}\label{short_ex_Hoch_cyc}
		0\longrightarrow \HC^{n}(A)^{(\geq 1)} \overset{I}{\longrightarrow} \HH^n(A,A^{\vee})^{(\geq 1)} \overset{B}{\longrightarrow} \HC^{n-1}(A)^{(\geq 1)} \longrightarrow 0\,.
	\end{equation}
	As a result, we have the following isomorphism
	\begin{equation}
		\HH^n(A,A^{\vee})^{(\geq 1)}/\HC^{n}(A)^{(\geq 1)} \overset{B}{\cong} \HC^{n-1}(A)^{(\geq 1)}\,.
	\end{equation}}

This identification suggests that we can map an element of the first tower of length $n-1$ to a quasi-primary in the second tower of length $n$. When a dual geometric picture is available, this corresponds to solving $\partial^{-1}\alpha$ for a divergence free vector $\alpha$. Here, we can solve it with the help of the stress tensor and the cup product on Hochschild cohomology. The stress tensor $T$ is a canonical member of the second tower of quasi-primary local operators. It corresponds to a function 
\begin{align}
	(T|  a_{0,v}; a_{2}^u) &= \delta^u_v \cr 
	(T|  a_{1}^\alpha;  a_{1}^\beta) &= \frac12 \omega_{\alpha \beta}
\end{align}
i.e. 
\begin{equation}
    (T|a;b) = (a \Delta b)\,.
\end{equation}
The cup product is conventionally defined on the Hochschild complex $\CH^{\bullet}(A,A)$ as follows
\begin{equation}\label{eq:cup}
    (f\cup g) (a_1,,\dots,a_{n+m})= f(a_1,\dots,a_n)g(a_{n+1},\dots,a_{n+m})
\end{equation}
for $f\in \CH^{n}(A,A),g\in \CH^{m}(A,A)$. We can translate this operation to $\CH^{\bullet}(A,A^\vee)$ through the identificatioin $A\cong A^{\vee}$. We find that, given a cyclic map $\mathfrak{c}\in \HC^{n}(A)$, its cup product with the stress tensor $\mathfrak{c}\cup T$ is given by
\begin{equation}\label{eqn:Hoch_to_cyc}
    (\mathfrak{c}\cup T)(a_1,\dots,a_{n+1}) = (\mathfrak{c}\mid a_{n+1}\Delta a_1,a_2,\dots,a_n)\,.
\end{equation}
Using the expression for $B$ \cite{loday2013cyclic}, we can check that the following identity holds
\begin{equation}
    B(\mathfrak{c}\cup T) = (-1)^{|\mathfrak{c}|}(\Delta_{\mathfrak{c}}+1)\mathfrak{c}\,.
\end{equation}
The above identity also follows from the fact that $(\HH^{\bullet}(A,A^{\vee}),B,\cup,\{,\})$) forms a BV algebra \cite{tradler2002bv}, analogous to the BV structure on Polyvector fields. We have
\begin{equation}
    B(\mathfrak{c}\cup T) = B(\mathfrak{c})\cup T + (-1)^{|\mathfrak{c}|}\mathfrak{c}\cup BT + (-1)^{|\mathfrak{c}|}\{\mathfrak{c},T\}\,.
\end{equation}
Then (\ref{eqn:Hoch_to_cyc}) follows from $\{\mathfrak{c},T\} = \Delta_{\mathfrak{c}}\mathfrak{c}$, $B(T) = 1$. The map $\mathfrak{c} \to \mathfrak{c}\cup T$ provides for us the identification between cyclic cohomology element $\HC^{\bullet-1}(A)$ and quasi-primary of the second tower $\HH^\bullet(A,A^\vee)/\HC^{\bullet}(A)$. For example, the stress tensor $T$ itself can be thought of as coming from the map $(\bullet): A \to \bC$ under this identification. 

\subsection{Operators with any number of derivatives}
Next, we will adopt a notation which allows us to deal transparently with derivatives of fields. The expression $\Phi(z+s) \in A[[s]]$ is a useful generating function for derivatives of $\Phi(z)$:
\begin{equation}
	\Phi(z+s) = \sum_{n=0}^\infty \frac{s^n}{n!} \partial_z^n \Phi(z)\,.
\end{equation}
The super vector space $A[[s]]$ plays the role of $V$ from the general discussion in the Introduction: it is dual to the collection of ``letters'' 
$\frac{1}{n!}\partial_z^n \phi^a$ which can occur in a single-trace local operator.

We can denote single-trace operators built from $\Phi$ and its derivatives concisely as 
\begin{equation}
	\cO_C(z) = \frac{1}{\hbar \ell(C)} \Tr (C|\Phi(z+s), \cdots, \Phi(z+s))
\end{equation}
with $(C|\cdots)$ defined as a multi-linear map 
\begin{equation}
	(C|\bullet, \cdots, \bullet): \left(A[[s]]\otimes \cdots \otimes A[[s]]\right)^{\mathbb{Z}_{\ell(C)}} \to \bC \,.
\end{equation} 
Explicitly, 
\begin{align}
	\cO_C(z) &= \frac{1}{\hbar \ell(C)} \Tr \left(C\Big|\,a_{a_1} \frac{s^{n_1}}{n_1!} \partial^{n_1} \phi^{a_1}(z), \cdots, a_{a_{|C|}} \frac{s^{n_{\ell(C)}}}{n_{\ell(C)}!} \partial^{n_{|C|}} \phi^{a_{|C|}}(z)\right) = \cr
 &= \pm \,\frac{ 1}{\hbar \ell(C)} \left(C\Big|\,a_{a_1} \frac{s^{n_1}}{n_1!} , \cdots, a_{a_{|C|}} \frac{s^{n_{\ell(C)}}}{n_{\ell(C)}!} \right)\Tr \partial^{n_1} \phi^{a_1}(z) \cdots \partial^{n_{\ell(C)}} \phi^{a_{\ell(C)}}(z) \, .
\end{align}
Hence the matrix elements of $C$ are essentially the coefficients in a general linear combination of single-trace operators built from $\ell(C)$ fields.

For example, the stress tensor
\begin{equation}
	T = -\frac{\omega_{\alpha \beta}}{2 \hbar} \Tr Z^\alpha \partial Z^\beta-\frac{1}{\hbar} \Tr b_u \partial c^u
\end{equation}
corresponds to a function $(T|\bullet, \bullet)$ with several non-zero entries 
\begin{align}
	(T|  a_{0,v} s, a_{2}^u) &= \delta^u_v \cr 
	(T|  a_{2}^u,  a_{0,v} s) &= -\delta^u_v \cr 
	(T|  a_{1}^\alpha s,  a_{1}^\beta) &= \frac12 \omega_{\alpha \beta}\cr 
	(T|  a_{1}^\alpha,  a_{1}^\beta s) &= -\frac12 \omega_{\alpha \beta}\,.
\end{align}
The single-trace local operators form a representation of the global conformal symmetry algebra. The global conformal generators $L_{-1}$, $L_0$ and $L_1$ act as vector fields on $\Phi(z)$. The action on the functionals $C$ follows from the action of the same vector fields on $A[[s]]$. For example, $L_{-1}$ adds a derivative on $\Phi$, i.e. acts as $\partial_s$ on the generating function. The other generators also use the information about the weight:
\begin{align}
	L_{-1} &= \partial_s \cr
	L_0 &= s \partial_s + \Delta \cr
	L_1 &= -s^2 \partial_s - 2 s \Delta\,.
\end{align}
In particular, quasi-primary operators are described by functionals annihilated by $L_1$. 

As the scaling dimensions are non-negative half-integers (in particular, $L_0$ is diagonalizable) and the number of operators of a given dimension is finite, the $SL(2)$ representation theory is quite restrictive. Quasi-primaries of positive dimension generate Verma modules consisting of their derivatives. Operators of dimension $0$ generate Verma modules which can contain quasi-primaries of dimension $1$. This only affects the spurious cohomology classes of dimension $0$ built from $c$ ghosts only.  

The tree-level BRST differential is easily identified with the differential for the cyclic cohomology complex $\CC^\bullet(A[[s]])$, which can be organized by the total number of derivatives appearing in the operator. 
We have already characterize the two towers of cohomology with $0$ and $1$ derivatives. We will now argue that all other cohomology consists of derivatives of operators in the two towers. 

\subsection{A cohomology computation}
It is useful to go back to our definition of the differential $C \to Q_0 C$. As we saw in the analysis of the second tower, the explicit formula for $Q_0$ maps cyclic-invariant maps to cyclic-invariant maps, but also makes sense on generic maps and defines the Hochschild cohomology complex $\CH^\bullet(A[[s]],A[[s]]^\vee)$. 

An important property of Hochschild cohomology complex is a good behaviour under tensor product: the complex $\CH^\bullet(A \otimes A',M \otimes M')$ is quasi-isomorphic to the complex $\CH^\bullet(A,M) \otimes \CH^\bullet(A',M')$.
This quasi-isomorphism is non-trivial \cite{LE2014Hochtensor}. Applied to the case at hand, this gives a quasi-isomorphism of complexes
\begin{equation} \label{eq:hhtensor}
    \CH^\bullet(A[[s]],A[[s]]^\vee) \simeq \CH^\bullet(A;A^\vee) \otimes \CH^\bullet(\bC[[s]],\bC[[s]]^\vee) \, .
\end{equation}
In order to proceed further, we need to recall the Connes construction relating Hochschild cohomology and 
cyclic cohomology.


There is an odd map $B$ which acts on $\CH^\bullet(A[[s]];A[[s]]^\vee)$ in the opposite direction as $Q_0$, roughly given by inserting the unit $1$ to each slot of the map. It is nilpotent and anti-commutes with $Q_0$. The relationship between the cyclic and Hochschild cohomology can be made precise by considering a bi-complex defined using $B$ \cite{loday2013cyclic}. This bi-complex can be written as follows
\begin{equation}
    (\CH^\bullet(A[[s]];A[[s]]^\vee)[v], Q_{\CH}+vB)\,,
\end{equation}
where $v$ is a formal parameter of degree $2$. Then $\CC^\bullet(A[[s]])$ is quasi-isomorphic to the above complex. The cohomology at order $v^0$ is simply given by the kernel of $B$ on the Hochschild cohomology. In general, cohomology at higher order of $v$ are non-zero, but in our case, the algebra is assumed to have a scaling degree. According to the short exact sequence in (\ref{short_ex_Hoch_cyc}), the positive degree part of cyclic cohomology\footnote{Here we give $s$ degree $1$, so only constant $c$ modes are in the zero degree part.}, which is the part we are interested in, is the kernel of Connes' operator $B$. This construction commutes with the action of global conformal transformations, so we expect an analogous relation for quasi-primaries.

We can combine this analysis with the fact the Hochschild cohomology factor with tensor product (\ref{eq:hhtensor}).
\begin{equation}
    \CH^\bullet(A[[s]];A[[s]]^\vee) \simeq \CH^\bullet(A;A^\vee) \otimes \CH^\bullet(\bC[[s]];\bC[[s]]^\vee) \, ,
\end{equation}
with Connes operator $B + \partial_{\bC}$, where $\partial_{\bC} = \frac{\partial}{\partial (\partial_s)}\frac{\partial}{\partial s}$ is the divergence operator on the $s$ plane and $B$ is the Connes' operator on $\CH^\bullet(A;A^\vee)$. 

We thus have a quasi-isomorphism relating the cyclic complex $\CC^\bullet(A[[s]])$ and the following complex 
\begin{equation}
    (\CH^\bullet(A;A^\vee)\otimes \CH^\bullet(\bC[[s]];\bC[[s]]^\vee)[v],Q_{\CH} + v(B + \partial)) \, .
\end{equation} 
Moreover, the relative cyclic cohomology can be identified with the $(B + \partial)$-invariant part of $\HH^{\bullet}(A,A^{\vee})\otimes \HH^\bullet(\bC[[s]];\bC[[s]]^\vee)$.

The tower of cohomology with no derivatives corresponds to elements of the form 
\begin{equation}
    \beta \delta(s) \in \HH^{\bullet}(A,A^{\vee}) \delta(s)
\end{equation}
which are invariant under $B$, i.e. with the expected $\HC^{\bullet}(A,A^{\vee}) \delta(s)$ from Section \ref{sec:first_tower}. 

Operators with one derivative correspond to elements of the form 
\begin{equation}
    \beta \delta'(s) - B \beta \delta(s)\partial_s 
\end{equation}
and are in correspondence with $\HH^{\bullet}(A,A^{\vee}) \delta'(s)$.
This is the answer we computed before in Section \ref{sec:second_tower}. 

The novel step is a characterization of operators in $\mathrm{Obs}_0$ containing more derivatives: they take the form 
\begin{equation}
    \beta \delta^{(k+1)}(s) - B \beta \delta^{(k)}(s)\partial_s 
\end{equation}
and are thus always derivatives of other operators. 

\subsection{A tree-level holographic dictionary for local operators}
The algebra $A[[s]]$ has a straightforward dg-TFT interpretation if $A$ does: it is the space of boundary local operators  for a brane of the form $B[A] \times \bC$ in a world-sheet theory combining $T[A]$ and the B-model with target $\bC$. Formally, these are the D-branes which support the chiral algebra as a world-volume theory. 

The coupling of closed string states to a D-brane is controlled by disc amplitudes. The disk amplitudes take as an input a closed string vertex operator and a cyclic collection of open string vertex operators. Accordingly, a closed string vertex operator for which the $B[A] \times \bC$ disk amplitudes are well defined can be mapped to a cyclic multi-linear function on $A[[s]]$. It is easy to see that the BRST operator 
acting on the closed string vertex operator maps to the differential in the cyclic complex. 

A standard entry of the dg-TFT dictionary is thus that there is chain complex from the space of such closed string vertex operators to $\CC^\bullet(A[[s]])$. The latter is also identified the space of single-trace local operators, giving a natural entry in the tree level holographic dictionary. 

Recall that closed string vertex operators are ``rotation-equivariant'' vertex operators in the 
dg-TFT. In a more conventional string theory setup, rotation-equivariant vertex operators are defined 
by a BRST complex restricted to rotationally-symmetric vertex operators. Standard vertex operators can also be inserted in disk correlation functions, but one of the boundary vertex operators remains unintegrated 
and cyclic symmetry may be absent. A standard entry of the dg-TFT dictionary is that there is chain complex from the space of standard bulk vertex operators to $\CH^\bullet(A[[s]],A[[s]]^\vee)$. The Connes complex 
describes the relation between standard and rotation-equivariant vertex operators in the dg-TFT. 

The quasi-isomorphism (\ref{eq:hhtensor}) expresses the fact that the space of vertex operators in the product of two dg-TFTs should be equivalent to the product of the spaces of vertex operators for the individual theories. The same is not true for rotation-equivariant vertex operators, which can combine factors of opposite worldsheet spin. The vertex operators in the B-model with target $\bC$ which 
are described by $\CH^\bullet(\bC[[s]],\bC[[s]]^\vee)$ are distributional in nature, as appropriate for 
representing insertions of local operators in a tree-level holographic dictionary. 

From a dg-TFT perspective, we can describe the tree-level holographic dictionary as follows:
\begin{itemize}
    \item An operator $\cO_\fc(z)$ in the first tower maps to a world-sheet operator of the form 
    \begin{equation}
         \fc \otimes \delta(\zeta-z) \, .
    \end{equation}
    This is a cyclic cohomology element in the full theory combining $T[A]$ and the B-model with target $\bC$. 
    \item An operator $\cO_\fh(z)$ in the second tower maps to a ($B$+divergence)-closed world-sheet operator of the form 
    \begin{equation}
        \fh \otimes \delta'(\zeta-z)  - B\fh \otimes  \delta(\zeta-z)\partial_\zeta \, .
    \end{equation}
    Quasi-primaries correspond to $\fh$ which vanish when acted upon by $\Delta$ and mapped to cyclic cohomology.
\end{itemize}
As discussed in Section \ref{sec:second_tower}, we can also construct quasi-primaries of the second tower from a cyclic map $\mathfrak{c}'$. It takes the form
$$
(\mathfrak{c}'\cup T)\otimes\delta'(\zeta - z) -  (\Delta_{\mathfrak{c}'}+1)\mathfrak{c}'\otimes \delta(\zeta - z)\partial_{\zeta}\,.
$$

This generalizes the standard example.  

\section{The global symmetry algebra}
\label{sec:general_GCA}
In this section we will analyze the tree-level limit $\fL_0$ of the global symmetry algebra of single-trace operators. Our objective is to compare it with the global symmetry algebra of the dual $\lambda \to 0$ worldsheet dg-TFT.

The algebra has a well-defined $\lambda \to 0$ limit. Much as for the BRST generator, the tree-level action of $\fL_0$ generators on local operators is encoded in the transformation of individual fields:
\begin{equation}
    [O, \Phi(s)] = \{O|\Phi(s),\cdots, \Phi(s)\}\,,
\end{equation}
where we employ a multi-linear map 
\begin{equation}
    \{O|\bullet,\cdots, \bullet\}: A[[s]]^{\otimes \ell(O)} \to A[[s]]
\end{equation}
to describe the structure constants of the transformation. Such a map can only be a symmetry at tree level if it commutes with $Q_0$. 

In the homological algebra language, we can study the complex defined by these maps with a differential $[Q_0,O]$. 
Essentially by definition, this is the Hochschild cohomology complex $\CH^\bullet(A[[s]]) \equiv \CH^\bullet(A[[s]],A[[s]]) $. Recall that $\CH^\bullet(A[[s]])$ is a dg-Lie algebra, with Lie bracket $[O,O']$ induced by the commutator of the corresponding transformations. Important examples of elements of $\CH^\bullet(A[[s]])$ are 
the global conformal generators 
\begin{align}
	L_{-1} &= \partial_s \cr
	L_0 &= s \partial_s + \Delta \cr
	L_1 &= -s^2 \partial_s - 2 s \Delta\,.
\end{align}

Not all such transformations will arise as modes of single-trace operators! We thus only have a dg-Lie algebra map
\begin{equation}
    \fL_0 \to \CH^\bullet(A[[s]])\,.
\end{equation}
The tree-level action of symmetries on local operators follows directly from these definitions and matches a well-known action of $\CH^\bullet(A[[s]])$ on $\CC^\bullet(A[[s]])$. The complex $\CH^\bullet(A[[s]])$ gives a standard dg-TFT description of the symmetries of the theory combining
$T[A]$ and the B-model on $\bC$.

We can use the nice properties of Hochschild cohomology under tensor product to produce another chain map:
\begin{equation}
    \fL_0 \to \CH^\bullet(A)[[s,\partial_s]]
\end{equation}
into polynomial polyvector fields on $\bC$ valued in $\CH^\bullet(A)$. In order to make sense of this statement, we should recall another property of $\CH^\bullet(A)$: it is also endowed with a cup product $\cup$ (see Equation \ref{eq:cup}) distinct from the bracket. This product allows one to express the Lie bracket on $\CH^\bullet(A)[[s,\partial_s]]$ as a combination of Lie brackets and products on $\CH^\bullet(A)$ and on $\bC[[s,\partial_s]]$.

A piece of the quasi-isomorphism between the complexes $\CH^\bullet(A)[[s,\partial_s]]$ and $\CH^\bullet(A[[s]])$ can be made rather explicit. Denote as $\fh$ an element of $\CH^\bullet(A)$. We can build a collection of elements of $\CH^\bullet(A[[s]])$ as 
\begin{equation}
    \{\fh|\Phi(s), \cdots, \Phi(s)\} s^n
\end{equation}
i.e. as a map which acts on $A[[s]]$ by acting on $A$ and combining the powers of $s$ in the arguments with $s^n$ to go in the output. This works essentially because translations commute with $Q_0$ and gives a map 
$\CH^\bullet(A)[[s]] \to \CH^\bullet(A[[s]])$. We will recover the second half of the quasi-isomorphism momentarily.

\subsection{The mode algebra}
We denote the Global Symmetry Algebra modes as 
\begin{equation}
	O_{n;C} \equiv \frac{1}{\hbar |C|} \oint_{|z|=1} \frac{dz}{2 \pi i} z^n \, \Tr (C|\Phi(z+s), \cdots, \Phi(z+s))
\end{equation}
with maximum $n$ at $2 \Delta_C -2$, i.e. the weight of $C$ minus $2$. Notice that the map ${\cal O}_C \to O_{0;C}$ annihilates descendants. These modes form an irreducible representation of dimension $2 \Delta_C -1$ under the global conformal group: 
\begin{align}
	[L_{-1}, O_{n;C}] &= - n O_{n-1;C} \cr
	[L_{0}, O_{n;C}] &= (\Delta_C -1- n) O_{n;C} \cr
	[L_{1}, O_{n;C}] &= (n+2 - 2 \Delta) O_{n+1;C} \,.
\end{align}

The tree-level action of modes $O_{n,\fc}$ of the first tower of quasi-primary operators on $\Phi$ is easily described as the Wick contraction has a single pole: 
\begin{equation}
    [O_{n,\fc}, \Phi(s)] =  (\fc|\Phi(s), \cdots, \eta^{(1)}) s^n\eta^{(2)}\,.
\end{equation}
This is clearly the combination of the map $\CC^\bullet(A) \to \CH^\bullet(A)$ given by contraction with $\eta$ and of the above-described collection of maps $\CH^\bullet(A) \to \CH^\bullet(A[[s]])$.

The action of the modes of the second tower is a bit more complicated
\begin{align}
    [O_{n,\fh}, \Phi(s)] &=  \partial_s \left((\fh|\eta^{(1)}; \cdots, \Phi(s)) s^n\right)\eta^{(2)} -  
    (\fh|\partial \Phi(s);\eta^{(1)}, \cdots, \Phi(s)) s^n \eta^{(2)} + \cr
    &\cdots \pm (\fh|\partial \Phi(s); \Phi(s), \cdots,\eta^{(1)}) s^n \eta^{(2)}\,.
\end{align}
This expression must give a second collection of maps $\CH^\bullet(A)\partial_s \to \CH^\bullet(A[[s]])$
which annihilate the image of $I$. Putting all together, we have identified the BRST cohomology in
$\fL^\bC_0$ with $\CC^\bullet(A)[[s,\partial_s]]$ and given a chain map
\begin{equation}
    \CC^\bullet(A)[[s,\partial_s]] \to  \CH^\bullet(A[[s]]) \simeq \CH^\bullet(A)[[s,\partial_s]]\,,
\end{equation}
which captures the image of $\fL^\bC_0$ in $\CH^\bullet(A[[s]])$. These should be analogues of divergence-free polynomial polyvector fields in $\bC \times X_2$. 

The construction is compatible with the action of the global conformal algebra. 
We can start from a lowest weight element in $\CH^\bullet(A)[[\partial_s]]$ defined as a zero mode of a 
quasi-primary local operator with a given half-integral $L_0$ eigenvalue $1-d$. We can then build up a representation of dimension $2d-1$ by acting repeatedly with $L_1$ until we hit an element which is annihilated by $L_1$.

For example, if we start from a lowest weight element of the form $\beta \in \CH^\bullet(A)$,
we will build a sequence 
\begin{align}
    & \beta \cr
    & (2d-2) s\beta \cr
    & (2d-2) (2d-3) s^2 \beta \cr
    &\cdots
\end{align}
If we start from $\beta \partial_s$ we will build a sequence 
\begin{align}
    & \beta \partial_s \cr
    & (2d-2) s \partial_s\beta+ 2\Delta \beta \cr
    &\cdots
\end{align}
This gives the image of $\fL_0$ in $\CH^\bullet(A[[s]])$.

\subsection{Polyvector fields on $\bC P^1$}
There is another useful geometric perspective on the problem. The complex 
$\CH^\bullet(A)[[s,\partial_s]]$ can be promoted to a complex of vector bundles on $\bC P^1$ by extending it over $s=\infty$ with the help of the 
scaling weight $\Delta$, so that $L_1$ and $L_{-1}$ are exchanged by $s \to s^{-1}$. 

Then the elements of $\CH^\bullet(A)[[s,\partial_s]]$ which fit into finite-dimensional irreps of the global conformal group are simply globally-defined holomorphic polyvector fields on $\bC P^1$ valued in $\CH^\bullet(A)$.

In geometric situations, these are the global symmetries of a B-model defined on a 3d CY geometry
\begin{equation}
    X_2(-1) \to \bC P^1\,,
\end{equation}
which is related to the $\lambda \neq 0$ dual geometry $Y_3$ by a conifold transition induced by $N$ branes wrapping the $\bC P^1$ base. It is a natural way to engineer sphere correlation functions. 

The planar corrections to the global symmetry algebra for a generic $A$ should generalize this conifold transition to a non-geometric setting. 

\subsection{The action of general local operators}
For completeness, we can describe the action of the modes of a general $\cO_C$ local operator. 
The OPE of all derivatives can be organized in a generating function \begin{equation}
    \Phi(z_1+s) \otimes \Phi(z_2+s) \in A[[s]] \otimes A[[s]]\,,
\end{equation}
which is proportional to
\begin{align}
	\sum_{n,m} \binom{n}{m} (-1)^m \frac{s ^{m}\eta^{(1)} \otimes s^{n-m} \eta^{(2)}}{(z_1-z_2)^{n+1}} \,.
\end{align}
We can now introduce some notation to lighten the complexity of this expression. The tensor product 
$A[[s]] \otimes A[[s]]$ can be described naturally in terms of polynomials in two variables $s_1$, $s_2$
so that the sum collapses to the obvious 
\begin{align}
	\frac{\eta}{z_1-z_2 + s_1 - s_2} &= \sum_n \frac{(s_2-s_1)^n \eta}{(z_1-z_2)^{n+1}} \,.
\end{align}
Furthermore, the numerator $E_n \equiv (s_2-s_1)^n \eta$
can be written as $E_n^{(1)} \otimes E_n^{(2)}$ just as we did with $\eta$, leaving implicit the sum over $m$ and over summands in $\eta$.

In conclusion, we write
\begin{equation}
	\Phi^i_j(z+s) \otimes \Phi^k_t(w+s) \sim \delta^i_t \delta^k_j \hbar \sum_n \frac{E_n}{(z-w)^{n+1}}
\end{equation}
and treat $E_n$ as we would $\eta$.

We compute
\begin{equation}
    [O_{n;C},\Phi] = \Tr (C|\Phi(s), \cdots, E_n^{(1)})E_n^{(2)}\,,
\end{equation}
which gives the explicit dg-Lie algebra map from $\fL^{\bC}_0$ to $\CH^\bullet(A[[s]])$.

\subsection{A conformal-invariant presentation of the global symmetry algebra}
The global conformal symmetry constrains the form of the global symmetry algebra $\fL_\lambda$. 

For example, consider the zero modes $O_{0,C_1}$ and $O_{0,C_2}$ of quasi-primary operators $\cO_{C_i}$ of scaling dimensions $\Delta_i$. The commutator is also lowest weight and thus must be the zero mode of some quasi-primary operator $\cO_{[C_1,C_2]_0}$ of scaling dimension $\Delta_{C_1} +\Delta_{C_2} -1$:
\begin{equation}
	[O_{0;C_1},O_{0;C_2}] = O_{0;[C_1,C_2]_0} \, .
\end{equation}
This relation determines the whole spin $\Delta_{C_1} +\Delta_{C_2} -2$ part of the $[O_{n;C_1},O_{m;C_2}]$ 
commutation relations. 

Analogously, we can verify that $[O_{1;C_1},O_{0;C_2}]-[O_{0;C_1},O_{1;C_2}]$ is also annihilated by $L_{-1}$ and thus 
we can define 
\begin{equation}
	O_{0;[C_1,C_2]_1} =[O_{1;C_1},O_{0;C_2}]-[O_{0;C_1},O_{1;C_2}]  \, .
\end{equation}
This determines the whole spin $\Delta_{C_1} +\Delta_{C_2} -3$ part of the $[O_{n;C_1},O_{m;C_2}]$ 
commutation relations.

We can systematically define 
\begin{align}
	 O_{0;[C_1,C_2]_2} &= [O_{2;C_1},O_{0;C_2}]-2 [O_{1;C_1},O_{1;C_2}] +[O_{0;C_1},O_{2;C_2}] \cr
	O_{0;[C_1,C_2]_3} &= [O_{3;C_1},O_{0;C_2}]-3 [O_{2;C_1},O_{1;C_2}] +3[O_{1;C_1},O_{2;C_2}]-[O_{0;C_1},O_{3;C_2}]  
\end{align}
etcetera. These relations capture the spin $\Delta_{C_1} +\Delta_{C_2} -n-1$ part of the $[O_{n;C_1},O_{m;C_2}]$ 
commutation relations. For example, with these notations we find $[T,T]_0=0$, $[T,T]_1 =2 T$.

This notation is useful because we can recover the full commutator $[O_{n,C_1},O_{m,C_2}]$ recursively from the $SL(2)$ symmetry and the $[C_1,C_2]_n$ brackets. It also turns out to simplify explicit calculations. 
Indeed, consider the schematic contribution of a term in the OPE which scales as $(z_1-z_2)^{-n-1}$:
\begin{align}
	&\oint_{|z_2|=1}  \frac{dz_2}{2 \pi i} z_2^{n_2} \oint_{|z_1-z_2|=\epsilon} \frac{dz_1}{2 \pi i} \frac{z_1^{n_1}}{(z_1-z_2)^{n+1}}f(z_1)g(z_2) = \frac{1}{n!}\oint_{|z|=1}  \frac{dz}{2 \pi i} z^{n_2}g(z) \partial^n_{z} (z^{n_1}f(z)) \, .
\end{align}
Then the contributions to $[\bullet, \bullet]_k$ simplifies to
\begin{equation}
    \frac{1}{(n-k)!} g(z) \partial^{n-k}_{z} f(z)
\end{equation}
and in particular is only non-vanishing for $n \geq k$. 

The structure of $\fL_0$ in this presentation is thus rather simple:
\begin{itemize}
    \item The tree-level OPE of two operators from the first tower involves a single Wick contraction and can only generate a simple pole. Only $[\fc_1,\fc_2]_0$ is thus non-vanishing. 
    \item The tree-level OPE of operators from the two towers contain a simple pole and a double pole. The latter allows a non-vanishing $[\fc,\fh]_1$ contribution belonging to the first tower. Both contribute to an $[\fc,\fh]_0$ belonging to the second tower.
    \item The tree-level OPE of operators from the second tower contributes 
    both an $[\fh_1,\fh_2]_2$ in the first tower and an $[\fh_1,\fh_2]_1$ in the second. The $[\fh_1,\fh_2]_0$ terms contains two derivatives and must thus vanish in cohomology
\end{itemize}

The tree-level OPE of two general operators is 
\begin{align}
	\cO_{C_1}(z_1) \cO_{C_2}(z_2) &\sim \sum_{n=0}^\infty \frac{1}{(z_1-z_2)^{n+1}} \cdot \cr &\frac{1}{\hbar}\Tr (C_1|\Phi(z_1+s), \cdots, E_n^{(1)})(C_2|E_n^{(2)}, \cdots, \Phi(z_2+s) ) \,.
\end{align}
We compute 
\begin{align}
	 \cO_{[C_1,C_2]_k}(z) &= \sum_{n=0}^\infty \frac{1}{n!}  \frac{1}{\hbar}\Tr \left[\partial^{n}_z (C_1|\Phi(z+s), \cdots, E_{n+k}^{(1)})\right](C_2|E_{n+k}^{(2)}, \cdots, \Phi(z+s) ) 
\end{align}
up to total derivatives. 

Somewhat implicitly, this expression defines the maps $[C_1,C_2]_k$. They are a collection of brackets on the quasi-primary part of $\CC^\bullet(A[[s]])$
which encode $\fL_0$. They provide a (poor) alternative computational method to mapping the quasi-primaries to the algebra of holomorphic polyvector fields on $\bC P^1$ valued in $\CH^\bullet(A)$.

In order to be a bit more explicit, we can observe that away from dimension $0$, the tree-level $Q$ must pair up whole Verma modules. We can define a ``quasi-primary complex'' $\CC_{\mathrm{qp}\geq 1}^\bullet(A[[s]])$
consisting of quasi-primaries in the cyclic cohomology complex.\footnote{Quasi-primaries of dimension $1$ which are total derivatives do not contribute zero modes and should be removed from $\CC_{\mathrm{qp}\geq 1}^\bullet(A[[s]])$.} We expect $\CC_{\mathrm{rel}}^\bullet(A[[s]])$ and $\CC^\bullet(A[[s]])$ to be quasi-isomorphic as chain complexes up to a collection of dimension $0$ fields built from the $c$ ghost only, without derivatives. If that is the case, $\CC_{\mathrm{rel},\mathrm{qp}}^\bullet(A[[s]])$ and $\CC_{\mathrm{qp}\geq 1}^\bullet(A[[s]])$ will be equivalent. 

The $[C_1, C_2]_k$ operations must then be well-defined on the $\CC_{\mathrm{qp}\geq 1}^\bullet(A[[s]])$ complex.

\section{Tree Level Flavour}
\label{sec:generla_flavor}
In this section we enrich the chiral algebra by some extra (anti)fundamental fields. We will ignore for now anomaly cancellation issues. 
We can collect the (anti)fundamental fields into two generating fields $I$ and $J$ valued in auxiliary vector spaces $\widetilde M$ and $M$. 
We can also include multiple copies of these fields, adding indices to get $I^r$ and $J_r$.

\subsection{Mesonic local operators}
In the planar approximation, the main actors are mesonic operators:  
\begin{equation}
	\cM_P(z) \equiv \frac{1}{\hbar} (P|I(z+s);\Phi(z+s), \cdots, \Phi(z+s);J(z+s))
\end{equation}
where $(P|\cdots)$ is a map $\widetilde M[[s]] \otimes A[[s]] \cdots \otimes M[[s]] \to\bC$. 

The tree-level BRST differential endows $\widetilde M$ with the structure of a right $A$-module and $M$ with the structure of a left $A$-module, so that $Q_0 I = I\Phi$ and $Q_0 J = \Phi J$. The action of 
$Q_0$ on functionals 
\begin{equation}
	\cM_{Q_0 \, P}(z) \equiv Q_0 \, \cM_{P}(z)
\end{equation}
receives contributions from both the variation of $\Phi$ and the variation of $I$ and $J$. Essentially by definition, it is dual to the bar complex for a (derived) tensor product $\widetilde M[[s]] \otimes_{A[[s]]} M[[s]]$ and thus mesons are 
labelled by linear functionals on that tensor product. 

The tensor product can be simplified drastically via a canonical quasi-isomorphism 
\begin{equation}
	\widetilde M[[s]] \otimes_{A[[s]]} M[[s]] \simeq (\widetilde M \otimes_{A} M)[[s]]
\end{equation}
so we expect quasi-primaries to be labelled by linear functionals on $\widetilde M \otimes_{A} M$. We can present them explicitly as 
\begin{equation}
	\cM_\fp(z) \equiv \frac{1}{\hbar} (\fp|I(z);\Phi(z), \cdots, \Phi(z);J(z))
\end{equation}
where $(\fp|\cdots)$ is a map $\widetilde M \otimes A \cdots \otimes M \to\bC$. 

In a dg-TFT setting, a probe brane $P$ in $T[A]$ can be encoded into the spaces of junctions from $P$ to $B[A]$ and vice versa. These are $\wt M$ and $M$ respectively. The brane $P$ can be combined with a brane wrapping $\bC$ in the extra direction to produce junctions labelled by $\wt M[[s]]$ and $M[[s]]$ respectively. The dual to $\widetilde M[[s]] \otimes_{A[[s]]} M[[s]]$
can be identified with a space of open string states attached to this brane,
providing one entry of the tree-level holographic dictionary. 

\subsection{Space-filling probe branes}
In practice, the most typical option for fundamental matter is 
some symplectic bosons or fermions coupled to the overall gauge group. Then $M$ and $\wt M$ are supported in ghost number $1$ and only $A_0$ acts non-trivially, encoding the action of the gauge group on the matter fields. 

In these situations, the derived tensor product $\wt M \otimes_A M$ will be quite large. Roughly, it involves bosonic generators of ghost number $0$ dual to $A_1$ and should be comparable in size to cyclic cohomology. The corresponding probe branes should be thought of as space-filling and are useful probes of the holographic dual geometry. 

A canonical possibility is to take $M=A_0[1]$, the algebra itself  shifted to ghost number $1$, and $\wt M = A_0^\vee[1]$ to have a natural pairing. We can denote the corresponding canonical space-filling brane as $O[A]$. 

The mesons for that brane are dual to 
\begin{equation}
    A_0^\vee \otimes_A A_0
\end{equation}
and include $IJ$ bilinears labelled by $A_0$ itself.

Further deformations of the mesonic BRST differential, possibly allowing for fundamental fields of non-zero ghost number, will give more general modules to be interpreted as generic D-branes in the transverse geometry $X_2$. 

\subsection{Open symmetry algebra}\label{subsec:open_symmetry_algebra}
We can define the open version of the global symmetry algebra from the modes of mesons. Including multiple flavours, we get modes we can denote as 
\begin{equation}
	(O_{n;P})^r_s\,.
\end{equation}
The linearized commutators between these modes have an interesting index structure: 
\begin{equation}\label{eq:ad_mode_algebra_product}
	[ (O_{n_1;P_1})^r_s,(O_{n_2;P_2})^u_v] = \delta^u_s (O_{n_1;P_1} \cdot O_{n_2;P_2})^r_v \pm  \delta^r_v (O_{n_2;P_2} \cdot O_{n_1;P_1})^u_s \,,
\end{equation}
which defines an algebra structure $\fP_\lambda$ on the global $O_{n;P}$ modes. In the presence of $k$ copies of the fudamentals, the open global symmetry algebra is $\mathfrak{gl}_k[\fP_\lambda]$ (up to a small mixing with $\fL_\lambda$ we discuss momentarily).

 We denote the OPE numerator pairing (anti)fundamental fields as $\mu\in M \otimes \wt M$, with the same $\mu^{(1)} \otimes \mu^{(2)}$ notation used for $\eta$. 
 
At tree level, the OPE between mesons only receives contributions from contractions of an fundamental and an anti-fundamental fields.  
Restricting to $O_{n;\fp}$, the OPE has only a simple pole and the mode algebra is controlled by the zero modes, reproducing the natural product
\begin{equation}
\begin{aligned}
    	&(\fp_1 \fp_2|\widetilde m_0;a_1, \cdots, a_{n_1+n_2};m_{n_1+n_2+1}) \\
     &\equiv (\fp_1|\widetilde m_0;a_1, \cdots, a_{n_1};\mu^{(1)})(\fp_2|\mu^{(2)};a_{n_1+1}, \cdots, a_{n_1+n_2};m_{n_1+n_2+1})
\end{aligned}
\end{equation}
on $(\widetilde M \otimes_{A} M)^\vee$.

Alternatively, we can focus on the transformations of $I$ and $J$ induced by 
$O_{n;P}$. Parsing definitions, these coincide respectively with $A_\infty$ endomorphisms of $M[[s]]$ and $\wt M[[s]]$ and compose accordingly. We thus get maps: 
\begin{equation}
    \fP_0 \to \mathrm{End}_{A[[s]]}(M[[s]]) \qquad \qquad \fP_0 \to \mathrm{End}_{A[[s]]}(\wt M[[s]])^{\mathrm{op}}\,.
\end{equation}
The spaces of endomorphisms are quasi-isomorphic to $\mathrm{End}_{A}(M)[[s]]$, etcetera

Effectively, we have mapped $\fP^\bC_0$ to functions on $\bC$ valued in $\mathrm{End}_{A}(M)$. We expect the part $\fP_0$ of $\mathrm{End}_{A}(M)[[s]]$ 
which consists of finite-dimensional representations of the global conformal algebra to be identified with global holomorphic sections of  $\mathrm{End}_{A}(M)$ as a vector bundle over $\bC P^1$. 

We can specialize to the case where $M$ and $\wt M$ are $A_0$ modules, and in particular to $O[A]$, i.e. $M = A_0[1]$. We can denote the corresponding mode algebra as $\mathfrak{O}_\lambda[A]$. It contains a copy of $A_0$ from the zero modes of $IJ$ mesons.

The algebra $\mathrm{End}_{A}(A_0)$ is a kind of Koszul dual to $A$. It is analogous to an algebra of holomorphic functions on the transverse space $X_2$. 

As we discussed in the standard example, $\fL_\lambda$ maps naturally to the Hochschild cohomology of $\fP_\lambda$.

\section{Determinant-like local operators} \label{sec:det}
The analysis is completely parallel to our standard example: 
a determinant-like operator is defined by an auxiliary integral 
on some (anti)fundamental 0-dimensional fields.
Fields and anti-fields alike can be 
assembled into generating fields $\widetilde \Psi$ and $\Psi$.

A typical BV action will take the form 
\begin{equation}
   S_{\BV} = (\wt \Psi; m \Psi) +  (\wt \Psi; \Phi(z) \Psi)\,,
\end{equation}
where $\wt \Psi$ and $\Psi$ are valued in dg-modules $\wt M$ and $M$ for $A$ equipped with a pairing $(\bullet;\bullet)$ of ghost number $3$.

If the auxiliary fields have a standard form, the modules will be supported in degree $1$ and $2$, with $M_1$ being paired to fundamental fields and being dual to $\wt M_2$ and vice versa. We denote the differential as $m$, as it gives rise to a quadratic mass term for the auxiliary fields. 

These modules are promoted to $A[[s]]$ modules with $s$ acting as multiplication by $z$. Coupling to derivatives of fields can be implemented by more general  $A[[s]]$ modules. We can take multiple copies of the fields to describe powers or products of determinants. 

The BV action must satisfy the master equation 
\begin{align}
(Q_{\mathrm{BRST}} + \hbar \Delta_{\BV})e^{S_{\BV}}\,.
\end{align}
It is easy to see that the assumptions above ensure that this is the case at tree-level: $Q_0$ produces a $(\wt \Psi; \Phi(z)\Phi(z) \Psi)$ term 
which is cancelled by $\{S_{\BV},S_{\BV}\}_\BV$. 

A general example would have $M = A_0[\theta]$, $\wt M = A_0^\vee[\theta]$.
The module structure is given by specifying a linear function $\mu$ on $A_1$, specifying how they map to a multiple of $\theta$. This selects which linear combination of the $Z$ fields we are taking the determinant of and generalizes the ``$u$'' parameter in the standard example.

\subsection{A BV algebra of mesons}
We can describe algebraically the space of ``determinant modifications'' at tree level, i.e. 
\begin{equation}
	\cM_D \equiv \frac{1}{\hbar} (D|\widetilde \Psi;\Phi(z+s), \cdots, \Phi(z+s);\Psi)
\end{equation}
inserted in the zero-dimensional auxiliary integral to modify the local operator, leading to $\fD_\lambda$ as in the canonical example. 

The linear and bi-linear terms in the action of $Q_{\mathrm{BRST}} + \hbar \Delta_{\BV}$ on a modified determinant equip $\fD_\lambda$ with the structure of a dg-algebra.

At tree level, the tensor product computing the cohomology of $\fD_\lambda$ simplifies to 
\begin{equation}
	\widetilde M \otimes_{A[[s]]} M = \left(\widetilde M \otimes_{A} M\right)[ds]\,,
\end{equation}
so that we expect to have two towers of modifications.

The product should arise from the BV Laplacian contracting a field and an anti-field in different modifications: 
\begin{equation}
    \cM_D \cM_{D'} \to \frac{1}{\hbar} (D|\widetilde \Psi;\Phi(z+s), \cdots, \Phi(z+s);\mu^{(1)})(D'|\mu^{(2)};\Phi(z+s), \cdots, \Phi(z+s);\Psi)
\end{equation}
leading to the familiar product on $\widetilde M \otimes_{A[[s]]} M$. 

\subsection{Open modifications}
In the presence of both fundamental fields and determinants, there will be mixed mesonic operators $\fM_\lambda$ and $\wt \fM_\lambda$.

At tree level, these are controlled by tensor products such as $\wt M_D \otimes_{A[[s]]} M_P[[s]] \simeq \wt M_D \otimes_A M_P$. For the simplest space-filling branes, that becomes $\wt M \otimes_A A_0$

Geometrically, we have a giant graviton brane $D$ which is Dirichlet in $\bC$ and a space-filling brane $P$ which is Neumann in $\bC$. At tree level, open strings stretched between the two are controlled by the boundary-changing local operators in the $T[A]$ factor of the theory. 

\section{A non-commutative example at tree-level} \label{sec:nctree}
Consider now the case of an $U(N)$ gauge theory with $2n+2$ bosonic adjoint fields and $2n$ fermionic ones. This has an $OSp(2n|2n+2)$ global symmetry. We can denote the matter fields collectively as $Z_a(z)$, with an OPE proportional to the ortho-symplectic form $\omega_{ab}$.

The corresponding algebra $A$ has ghost number $1$ generators $\theta^A$,
which satisfy 
\begin{equation}
    \theta^a \theta^b = \omega^{ab} u
\end{equation}
for the only ghost number $2$ generator $u$, dual to the $b(z)$ field. To be more precise, our algebra $A$ is defined as the quotient algebra
\begin{equation}\label{def_nonc_A}
    A = \mathbb{C}\langle\theta^a,u\rangle/(\theta^a\theta^b = \omega^{ab} u,\theta^a u =0),
\end{equation}
where the notation $\mathbb{C}\langle \dots \rangle$ represents the non-commutative algebra freely generated by the variables in the angle bracket.

The trace map $(\bullet):A \to \bC$ is simply defined by $(u) = 1$ and $(\text{others}) = 0$.

We would like to identify $A$ as the algebra of boundary local operators for some ``Dirichlet'' brane in a dg-TFT $T[A]$, but we do not have any alternative definition of the theory. 

Standard fundamental matter gives mesons which are computed from 
\begin{equation}
    (\bC \otimes_A \bC)^\vee \equiv A^!\,.
\end{equation}
Through the presentation \eqref{def_nonc_A} of $A$, we can view it as a quadratic-linear algebra. By the standard technique of the quadratic Koszul duality \cite{loday2012algebraic}, we find that its Koszul dual can be computed by the complex
\begin{equation}\label{eqn:Koszul_nc}
    (\mathbb{C}\langle \zeta_a, \nu \rangle, d)\,,
\end{equation}
with $d\nu = \omega^{ab} \zeta_a \zeta_b,d\zeta_a = 0$. Computing the cohomology eliminates the variable $\nu$ and imposes the relation \begin{equation}\label{eqn_nc_rel}
    \omega^{ab} \zeta_a \zeta_b =0\,.
\end{equation} 
Thus we obtain the Koszul dual algebra 
\begin{equation}\label{eq:nc_tree_level_algebra_of_mesons}
A^! = \mathbb{C}\langle \zeta_a \rangle/(\omega^{ab} \zeta_a \zeta_b =0)\,,    
\end{equation}
which is a non-commutative algebra for $n > 0$.   

The interpretation of this result is intuitive: the meson operators
\begin{equation}
    I^A Z_{a_1} \cdots Z_{a_n} J_B
\end{equation}
are $Q_0$-closed but a contraction of $\omega^{ab}$ with a pair of consecutive indices gives an exact operator. The $Q_0$ image of an operator with an extra $b$ insertion plays a role analogous to the $d$ image of $\nu$ in \eqref{eqn:Koszul_nc}. 

This nicely sets the stage for a computation of the $Q_0$ cohomology on single-trace operators. If $A$ was commutative, the Hochschild cohomology would be given by the tensor product of the algebras $\mathbb{C}[\theta^a]$ and $\mathbb{C}[\zeta_a]$. In the non-commutative case, the Hochschild cohomology $\HH^{\bullet}(A)$ can also be computed by the tensor product $A\otimes A^{!}$, but now equipped with a differential. We can write this complex $(A\otimes A^{!},Q_{\CH})$ as 
\begin{equation}
    0 \rightarrow A^! \overset{Q_{\CH}}{\rightarrow} \bigoplus_{a}A^!\theta^a \overset{Q_{\CH}}{\rightarrow} A^!u \rightarrow 0\,,
\end{equation}
where $Q_{\CH}f(\zeta) = \sum [\zeta_a,f(\zeta)]\theta^a$, and $Q_{\CH}(f(\zeta)\theta^a) = [\omega^{ab}\zeta_b,f(\zeta)]u$. This complex can be understood as the analog of polyvector fields. For $n = 0$, we immediately see that $Q_{\CH} = 0$, and the above complex reproduces the space of polyvector fields on $\bC^2$. In the following, we focus on the case when $n >0$.

To compute the cyclic cohomology, it will be convenient to use the trace pairing to dualize $A$ in the above complex. We write it as
\begin{equation}
    0 \rightarrow A^!u^* \overset{Q_{\CH}}{\rightarrow} \bigoplus_{a}A^!d\zeta_a \overset{Q_{\CH}}{\rightarrow} A^! \rightarrow 0\,.
\end{equation}
The differential now takes the following form
\begin{equation}
\begin{aligned}
        Q_{\CH}(f(\zeta)u^*) &= \sum\omega_{ab}[\zeta_a,f(\zeta)]d\zeta_b\,,\\
        Q_{\CH}(g(\zeta)d\zeta_a) &= [\zeta_a,g(\zeta)]\,.
\end{aligned}
\end{equation}
After this identification, we can now write the Connes $B$ operator as $B = d\zeta_a\frac{\partial}{\partial \zeta_a}$. 

In general, the cyclic cohomology is computed as a bi-complex given by $(A\otimes A^{!}[v],Q_{\CH}+vB)$. However, our algebra has the important property that it carries a scaling degree. Then, according to our discussion around Equation \eqref{short_ex_Hoch_cyc}, the positive scaling degree part of the cyclic cohomology can be identified with the kernel of $B$. 

We first find that the kernel of the map $Q_{\CH}: A^!u^* \rightarrow\bigoplus_{a}A^!d\zeta_a $ is the center $Z(A^{!})$ of the non-commutative algebra $A^{!}$. In the commutative case, the center is the whole algebra, and they correspond to the $\mathcal{B}^{(n)}$ tower in the standard example. In the non-commutative case, the center is much smaller. This implies that many of the original $\mathcal{B}^{(n)}$ tower operators no longer exist. For example $Q\Tr bZ_c = \omega^{ab}\Tr Z_aZ_bZ_c$, so that $\Tr bZ_c $ is not BRST closed. The center $Z(A^{!})$ contains at least $\bC$, corresponding to
\begin{equation}
    \mathcal{B}^{(0)} = \Tr b\,.
\end{equation}
However, we do not know if there exists any other non-trivial element in the center $Z(A^{!})$.

The kernel of $B$ on $A^!$ only gives $\bC$, which correspond to $\Tr c$ which we don't consider. This also exclude any possibility with $\Tr(cZ_{a}\dots )$. We thus find that the remaining operators in the (relative) cyclic cohomology, or the first tower, is given by the $\mathcal{A}^{(n)}$ tower 
\begin{equation}
      \frac{1}{n \hbar} \fc^{a_1,\cdots, a_n} \Tr Z_{a_1} \cdots Z_{a_n}\,.
\end{equation}
These are the cyclic words on $\zeta_a$ modulo the relation \eqref{eqn_nc_rel}. The quotient by the relation can be realized by considering $Q_0\Tr(bZ^{i_1}\dots Z^{i_n})$.

According to the discussion in Section \ref{sec:second_tower}, the cyclic cohomology $\HC^{\bullet}(A)$ also characterizes quasi-primaries in the second tower. For the same reason that the original $\mathcal{B}^{(n)}$ tower collapses, most of the original $\mathcal{D}^{(n)}$ tower operators are killed by BRST invariance. The $\bC$ in the center $Z(A^!)$ gives the stress tensor
\begin{equation}
    \mathcal{D}^{(0)} = \frac{1}{2\hbar}\Tr(\omega^{ab}Z_a\partial Z_b + 2b\partial c)\,,
\end{equation}
which survives in the BRST cohomology. 
We also have the following $\mathcal{C}^{(n)}$ tower
\begin{equation}
      \frac{1}{n \hbar} \fc^{a_1,\cdots, a_n} \Tr \partial c Z_{a_1} \cdots Z_{a_n}\,,
\end{equation}
with the same condition on $\fc^{a_1,\cdots, a_n}$ as before. Naively, we don't need the cyclic condition on $\fc^{a_1,\cdots, a_n}$. However, we can check that two operators $\Tr \partial c Z_{a_1} \cdots Z_{a_n}$, $(\pm)\Tr \partial c Z_{a_2} \cdots Z_{a_n} Z_{a_1}$ related by a cyclic permutation on $Z_{a_i}$ differ by a BRST exact element $Q_0\Tr(\partial Z_{a_1} Z_{a_2}\dots Z_{a_n})$.

Finally, we could consider some determinant operators $\det (m + u^a Z_a)$. 
This leads to an $A$-module $M_u = \bC[\theta]$, defined by mapping $\theta^a \to u^a \theta$ and $u \to 0$, with differential $m \theta$. We define an auxiliary linear map $p(\zeta_a) = u^a\theta$ on the subspace $(A^{!})_{1} = \oplus_a\bC\zeta_a$. Then $\ker p$ is a subspace of $(A^{!})_{1}$. We define $M_u^! = \langle \ker p\rangle$ as the subalgebra of $A^{!}$ generated by $\ker p$.

We are particularly interested in the space of open modifications, computed from 
\begin{equation}
    (\bC \otimes_A M_u)^\vee \equiv M_u^!\,.
\end{equation}
This can be derived using the standard Koszul resolution $(A\otimes A^{\text{\textexclamdown}},d_{\text{Kos}})$ of $\bC$ \footnote{Here, $A^{\text{\textexclamdown}}$ is the Koszul dual coalgebra of $A$, which is the linear dual of the Koszul dual algebra $A^!$. We refer to \cite{loday2012algebraic} for a general discussion of this and the Koszul complex.}. We find that $(\bC \otimes_A M_u)^\vee$ can be computed by the complex $(A^![\theta^*],d)$ with differential $d\zeta_a = u^a\theta^*$. Cohomology of this complex gives us $M_u^!$. This is an $A^!$ module defined by the left ideal generated from $m+u^a \zeta_a$, with a natural identification as $I Z_{a_1} \cdots Z_{a_n} \psi$ 
modifications. 

\section{Categorical Back-reaction} \label{sec:back}
As long as a planar 1-loop anomaly cancellation condition holds, we can follow the canonical example and define a dg-Lie algebra $\fL_\lambda$ from global modes of single-trace operators. Given choices $P$ of fundamental chiral matter and $D$ of determinant operators, we can define dg-algebras $\fP_\lambda$ and $\fD_\lambda$ from global modes of mesons and determinant modifications, as well as $\fM_\lambda$ and $\wt \fM_\lambda$ bimodules of open determinant modifications. Any anomalies introduced by the fundamental fields will at most curve some of the algebras.  

Our general strategy is to tentatively {\it define} a dual world-sheet 
theory $T_\lambda$ and D-branes $P_\lambda$ and $D_\lambda$ from this data. The definition of that data via BRST anomalies essentially guarantees the axiomatic properties expected from them. This includes maps from $\fL_\lambda$ into the Hochschild cohomology of the algebras and modules, the module action themselves, etc. 

In the remainder of the paper we will begin the work of making these constructions explicit.

We begin by reviewing the one-loop anomaly cancellation. Recall that the  
BRST differential is the zero mode of a BRST current, which we can concisely write as 
\begin{equation}
    J_{\mathrm{BRST}} = \frac{1}{3 \hbar} (\Phi \Phi \Phi)\,.
\end{equation}
The full action of $Q$ on local operators has a tree-level $Q_0$ and 1-loop $\hbar Q_1$ parts, involving $1$ or $2$ Wick contractions with $J_{\mathrm{BRST}}$.  
We would like $Q_0^2=0$, $\{Q_0, Q_1\} = 0$ and $Q_1^2=0$ separately, so we have a BRST symmetry for all values of 
$\hbar$. 

We can actually require the stronger condition that $Q J_{\mathrm{BRST}}=0$.
This condition has a tree-level part $Q_0 J_{\mathrm{BRST}}=0$, involving one Wick contraction, and a 1-loop part $Q_1 J_{\mathrm{BRST}}=0$ involving two. 

The tree-level condition is 
\begin{equation}
	\Tr (\Phi \Phi \eta^{(1)})(\eta^{(2)}\Phi \Phi) = \Tr (\Phi \Phi \Phi \Phi) =0
\end{equation}	
and is identically satisfied by cyclicity of $()$. 

At 1-loop, we have a planar contribution
\begin{equation}
	\Tr (\partial \Phi  \eta_1^{(1)}\eta_2^{(1)})(\eta_2^{(2)}\eta_1^{(2)}\Phi)\,.
\end{equation}	
We named the two $\eta$ tensors to avoid confusion. There is a second term involving a different relative order of Wick contractions, leading to a double trace. We will ignore it, as it does not obstruct the construction of the linearized planar structures we are interested in.

The ``planar anomaly'' cancellation condition on $A$ is thus 
\begin{equation}
	(a \eta_1^{(1)}\eta_2^{(1)})(\eta_2^{(2)}\eta_1^{(2)}b)=0\,,
\end{equation}	
i.e. 
\begin{equation}
	(a \eta^{(1)} \eta^{(2)}b)=0\,,
\end{equation}	
i.e. 
\begin{equation}
	\eta^{(1)} \eta^{(2)} =0\,.
\end{equation}
In the canonical example, this holds because of a cancellation between bosons and fermions. 

\subsection{Adding (anti)fundamental chiral fields}
The BRST current gains a second term in the presence of (anti)fundamental fields:
\begin{equation}
    J_{\mathrm{BRST}}^{IJ} = \frac{1}{\hbar} (I \Phi J)\,.
\end{equation}
If we ignore non-planar contributions, the potential BRST anomaly comes from terms with two Wick contractions. This leads to something like
\begin{equation}
    (\tilde m \eta^{(1)} \mu^{(1)})(\mu^{(2)}\eta^{(2)} m)=0\,,
\end{equation}
which is guaranteed by again by $\eta^{(1)} \eta^{(2)}=0$. We are thus free to add any fundamental matter at the planar level. 

We now describe the general structure of linear planar corrections, in greater generality than the rest of the paper. 

We will work with a matrix super-field $\Phi$ of ghost number $1$, valued in an auxiliary space $V$ which plays a role analogous to that of $A[[s]]$ in the main text. 

\subsection{Tree level}
The most general form of a tree-level differential acting via a Leibniz rule is
\begin{equation}
    Q^{(0)} \Phi = \{Q|\Phi\} + \{Q|\Phi, \Phi\} + \{Q|\Phi, \Phi, \Phi\}+ \cdots
\end{equation}
where 
\begin{equation}
    \{Q|\bullet, \cdots, \bullet\} : V^{\otimes (n+1)} \to V
\end{equation}
is a collection of brackets on $V$. The condition $Q^2=0$ 
gives a set of quadratic relations which makes $\{Q|\bullet, \cdots, \bullet\}$ into an $A_\infty$ algebra. 

A collection $L$ of maps $V^{\otimes (n+1)} \to V$ is, by definition, an element of the Hochschild cohomology complex $HH^\bullet(V,V)$. We can denote the transformation 
\begin{equation}
    L \Phi = \{L|\Phi\} + \{L|\Phi, \Phi\} + \{L|\Phi, \Phi, \Phi\}+ \cdots
\end{equation}
by the same symbol. The bracket $\{\bullet, \bullet\}$ on the Hochschild complex is defined in such a manner that 
\begin{equation}
    [L,L'] \Phi = \{\{L, L'\}|\Phi\} + \{\{L, L'\}|\Phi, \Phi\} + \{\{L, L'\}|\Phi, \Phi, \Phi\}+ \cdots
\end{equation}
and is a sum over all possible ways of inserting a map into the other and vice versa. In this notation, the differential on the Hochschild cohomology complex is simply:
\begin{equation}
    Q L = \{Q,L\} \,.
\end{equation}
The RHS of this equation corresponds diagrammatically to:
\begin{figure}[h!]
    \centering
    \includegraphics[width=\linewidth]{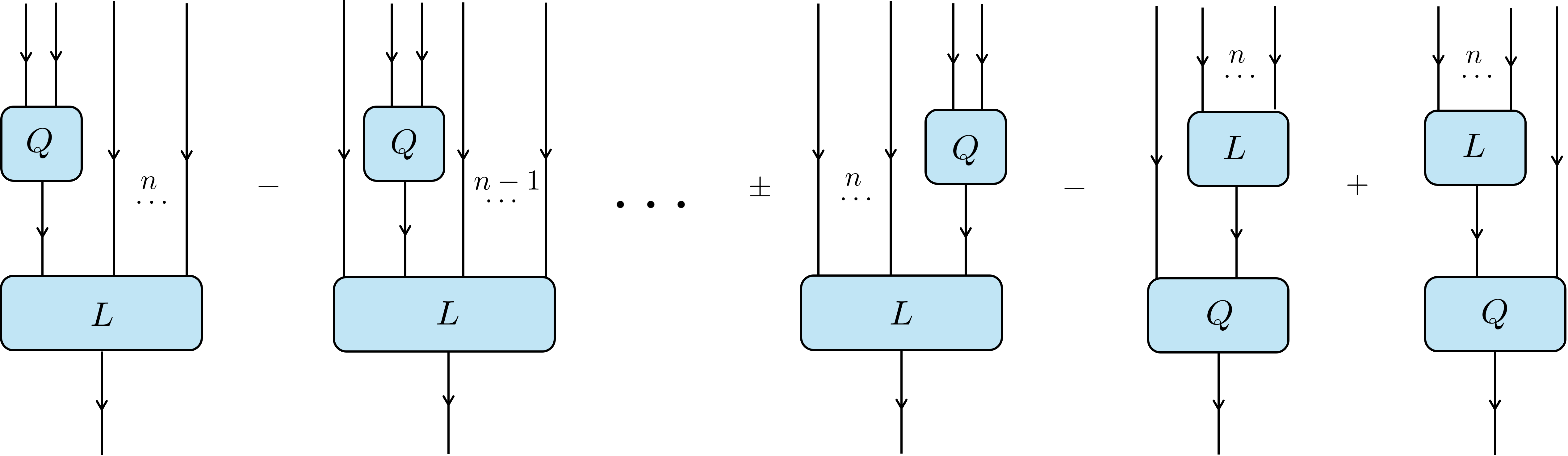}
    \caption{Illustration of how the Hochschild differential $Q$ acts on a generic $L$. Black lines represent $\Phi$'s.}
    \label{fig:enter-label}
\end{figure}

It is reasonable to identify $\HH^\bullet(V,V)$ as the ``tree level symmetry algebra'' of the underlying theory: classical field redefinitions which are compatible with the BRST differential. 

Notice that deformations of the differential are controlled by a quadratic MC equation
\begin{equation}
    Q V +\{V,V\}=0 \, .
\end{equation}

The definitions can be extended to the case where $Q\Phi$ includes a constant source term $\{Q|\}$, i.e. $V$ is a {\it curved} $A_\infty$ algebra. 

\subsection{Planar transformations}
At the planar level, the BRST differential is a sum of terms which act on $m+1$ consecutive fields in a trace or meson and replaces them with a sum of products of $n+1$ fields. 
We can write that as 
\begin{equation}
    Q^{(m)} \Phi^{\otimes(m+1)} = \{Q|\Phi\}_m + \{Q|\Phi, \Phi\}_m + \{Q|\Phi, \Phi, \Phi\}_m+ \cdots
\end{equation}
where now 
\begin{equation}
    \{Q|\bullet, \cdots, \bullet\}_m : V^{\otimes (n+1)} \to V^{\otimes (m+1)}\,.
\end{equation}
Denote the space of collections of such maps as $\BC^{\bullet, \bullet}(V)$. This is equipped with an obvious bracket such that 
\begin{equation}
    \sum_{k=0}^m [L^{(k)}_1,L^{(m-k)}_2] \Phi = \{\{L, L'\}|\Phi\}_m + \{\{L, L'\}|\Phi, \Phi\}_m + \{\{L, L'\}|\Phi, \Phi, \Phi\}_m+ \cdots
\end{equation}
The bracket is a sum of terms where the output of one operation is inserted in a consecutive sequence of slots 
in the other operation in all possible ways, and vice versa (see for example Figure~\ref{fig:planar L}). 

\begin{figure}[h]
    \centering
    \includegraphics[width=0.7\linewidth]{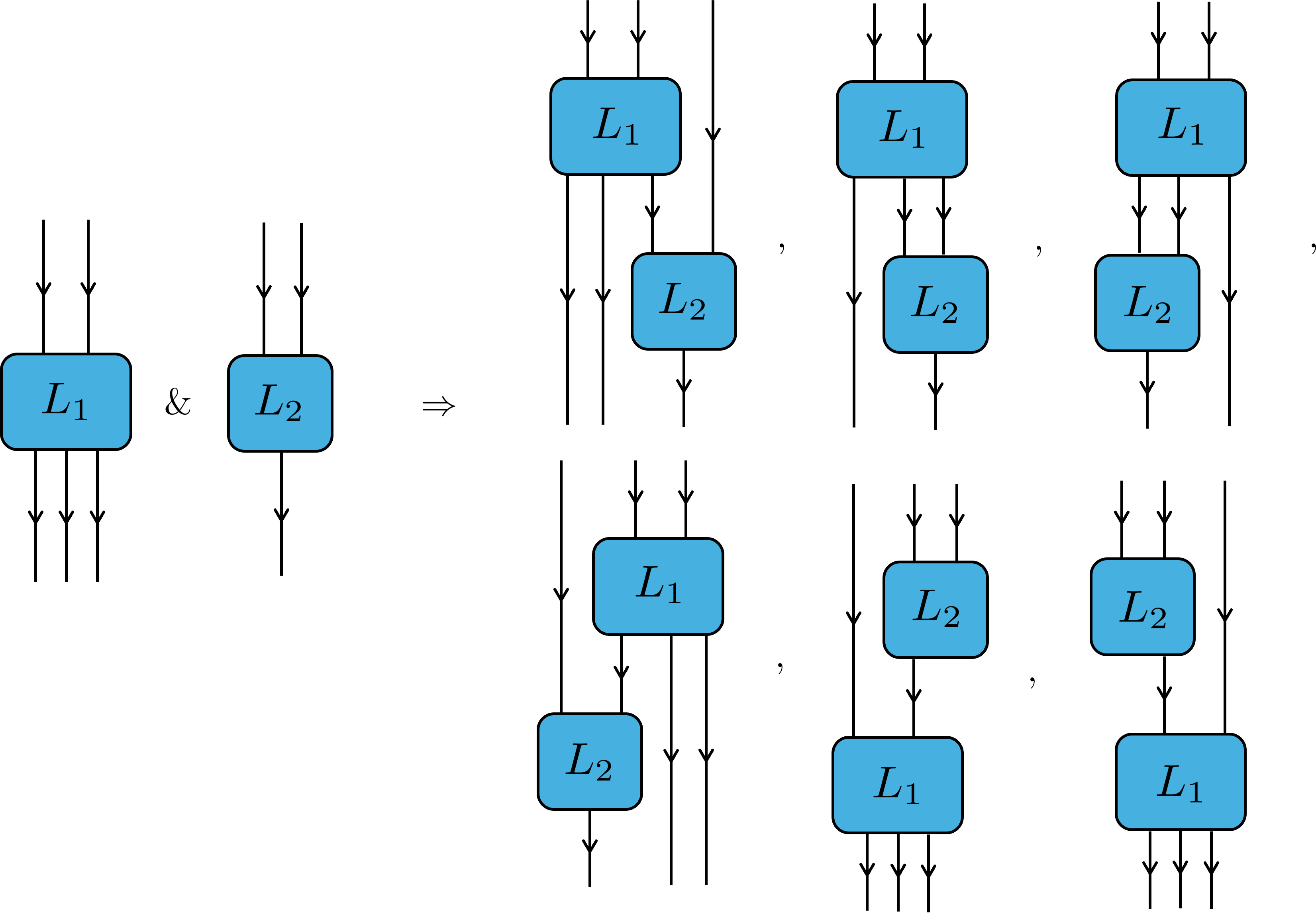}
    \caption{Illustration of the resulting combinations of two specific maps $L_1$ and $L_2$. Black lines represent $\Phi$'s.}
    \label{fig:planar L}
\end{figure}

We thus have quadratic relations $\{Q,Q\}=0$. We denote this structure as a ``planar algebra''. 

The BRST differential gives a differential on $\BC^{\bullet, \bullet}(A)$:
\begin{equation}
    Q L = \{Q,L\} \, .
\end{equation}
We denote the resulting complex as the ``planar Hochschild cohomology complex'' $\BH^{\bullet, \bullet}(V,Q)$ for $(V,Q)$. 

Again, deformations of the planar algebra structure are controlled by a quadratic MC equation
\begin{equation}
    Q L +\{L,L\}=0 \, .
\end{equation}

On general grounds, a planar symmetry of the underlying theory should manifest itself as a deformation of the planar algebra and thus an element of $\BH^{\bullet, \bullet}(V,Q)$.
We are led to identify $\BH^{\bullet, \bullet}(V,Q)$ as a ``formal planar symmetry algebra'' of the underlying theory/planar algebra. It is not obvious that $\BH^{\bullet, \bullet}(V,Q)$ should coincide with the actual planar symmetry algebra, which e.g. could be a sub-algebra of $\BH^{\bullet, \bullet}(V,Q)$. 
At first sight, $\BH^{\bullet, \bullet}(V,Q)$ is ``too big''. Still, the ability to do algebra computations in $\BH^{\bullet, \bullet}(V,Q)$ should be invaluable. 

The definitions can be extended to the case where $Q^{(m)}$ includes a constant source term $\{Q|\}_m$, i.e. $V$ is a {\it curved} planar algebra. 

\subsection{Modules}
Fundamental flavours in a theory will transform at tree level as \begin{equation}
    Q_M^{(0)} J = \{Q_M|;J\} + \{Q_M|\Phi; J\} + \{Q_M|\Phi, \Phi; J\}+ \cdots
\end{equation}
where 
\begin{equation}
    \{Q_M|\bullet, \cdots; \bullet\} : V^{\otimes n}\otimes M \to M \, .
\end{equation}
Suppose we have a tree level BRST differential $Q$ that defines a $A_\infty$ algebra structure on $V$. The condition $(Q+Q_M)^2=0$ gives a set of quadratic relations which makes $M$ equipped with $\{Q_M|\bullet, \cdots, \bullet;\bullet\}$ into an $A_\infty$ module of $V$.

Given two modules $M_1$ and $M_2$, we can consider the collection $L_{12}$ of maps $V^{\otimes n}\otimes M_1 \to M_2$. We can denote the transformation 
\begin{equation}
    L_{12} J_2 = \{L_{12}|;J_1\} + \{L_{12}|\Phi;J_1\} + \{L_{12}|\Phi, \Phi;J_1\}+ \cdots
\end{equation}
by the same symbol. The composition $\bullet \circ \bullet$ can be defined in such a manner that 
\begin{equation}
\label{eq:tree modules L}
    L_{12}L_{23} J_3 = \{L_{12} \circ L_{23}|;J_1\} + \{L_{12} \circ L_{23}|\Phi;J_1\} + \{L_{12} \circ L_{23}|\Phi, \Phi;J_1\}+ \cdots
\end{equation}
and is given by inserting $L_{12}$ into the corresponding slot of $L_{23}$ (see Figure~\ref{fig:tree modules}).

\begin{figure}[h]
    \centering
    \includegraphics[width=0.9\linewidth]{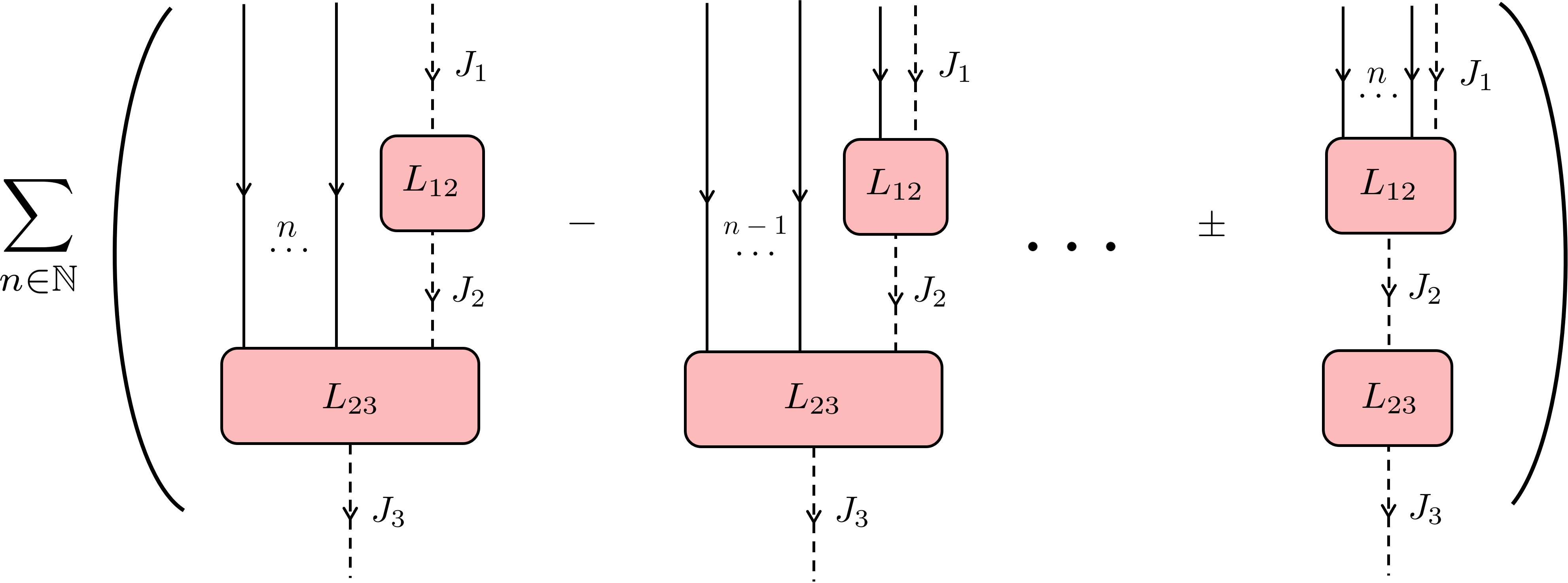}
    \caption{Diagrams corresponding to the RHS of Equation~\eqref{eq:tree modules L}. Solid lines represent $\Phi$'s, dashed lines represent $J$'s.}
    \label{fig:tree modules}
\end{figure}

Note that we can also compose a map $Q:V^{\otimes \bullet} \to V
$ with $L_{12}$ such that $Q\circ L_{12}$ is a sum of over all possible ways of inserting $Q$ into $L_{12}$. Using this notation, we can define a differential on the space of maps $\{L_{12}\}$
\begin{equation}
    Q L_{12} = (Q + Q_1)\circ L_{12} - (-1)^{\cdots} L_{12}\circ Q_2\,.
\end{equation}
By definition, an $A_\infty$ morphism from $M_1$ to $M_2$ is a map $L_{12}$ such that $QL_{12} = 0$.

\subsection{Planar modules}
At the planar level, the BRST symmetry will now act on a fundamental field $J$ together with a collection of consecutive $\Phi$ fields before it:
\begin{equation}
    Q_M^{(m)} \Phi^{\otimes m} \otimes J = \{Q_M|;J\}_m + \{Q_M|\Phi; J\}_m + \{Q_M|\Phi, \Phi; J\}_m+ \cdots
\end{equation}
encoded in maps 
\begin{equation}
    \{Q_M|\bullet, \cdots; \bullet\}_m : V^{\otimes n}\otimes M \to V^{\otimes m}\otimes M \, .
\end{equation}
We call such a pair a $(M,Q_M)$ a ``planar module'' for the planar algebra $(V,Q)$ if it satisfy $(Q+Q_M)^2$. 

We get for free a category $\mathrm{PMod}_V$ of planar modules. Given two planar modules $M_1$, $M_2$, we consider collections of maps 
\begin{equation}
    L_{12}: V^{\otimes n}\otimes M_1 \to V^{\otimes m}\otimes M_2\,.
\end{equation}
These can be composed in an obvious manner (see for example Figure~\ref{fig:planar modules}). 

\begin{figure}[h]
    \centering
    \includegraphics[width=0.45\linewidth]{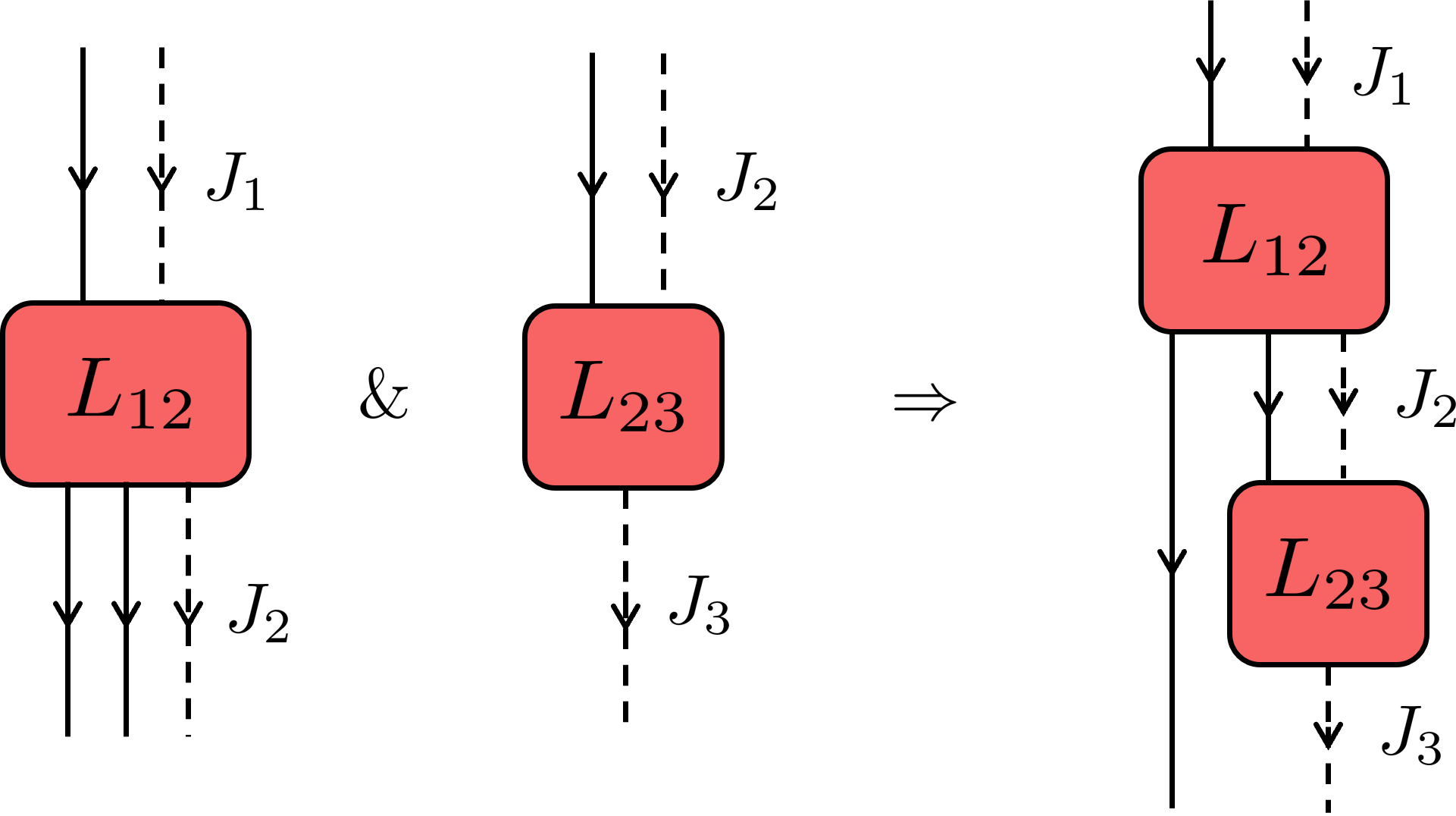}
    \caption{Illustration of the resulting combinations of two specific maps $L_{12}$ and $L_{23}$. Solid lines represent $\Phi$'s, dashed lines represent $J$'s.}
    \label{fig:planar modules}
\end{figure}

Furthermore, 
we have a differential defined in the obvious way:
\begin{equation}
    Q L_{12} = (Q+Q_1) L_{12} - (-1)^{\cdots} L_{12}(Q+ Q_2)\,.
\end{equation}
Then a morphism in the category of planar modules is defined to be a map $L_{12}$ that $QL_{12} = 0$ \footnote{The derived category $\mathrm{PMod}_V$ should be defined by a suitable localization with respect to the class of quasi-isomorphisms.}.

Again, $\mathrm{PMod}_V$ seems to be a formal version of an actual category which describes planar calculations in the presence of fundamental fields: a consistent way to include fundamental fields will map to an object of $\mathrm{PMod}_V$ and deformations of that will map to its morphisms. In other words, there should be an actual planar category which deforms the category of $A_\infty$ modules which appear at tree level and which has a functor to $\mathrm{PMod}_V$ which is likely not surjective. 

The planar-level statement of open-closed duality is that 
there should be a back-reacted world-sheet theory reproducing planar calculations. A category 
is a natural way to present a 2d TFT. It would be interesting to know if the formal category 
$\mathrm{PMod}_V$ could give an algebraic description of the back-reacted world-sheet theory.  As a test of this idea, we are led to the following conjectures:
\begin{itemize}
    \item ``Easy conjecture'': the planar global symmetry algebra $\BH^{\bullet, \bullet}(V)$ acts on $\mathrm{PMod}_V$, i.e. there is a nice map from $\BH^{\bullet, \bullet}(V)$ to the Hochschild cohomology of $\mathrm{PMod}_V$.
    \item ``Unlikely conjecture'': the planar global symmetry algebra $\BH^{\bullet, \bullet}(V)$ is quasi-isomorphic (in the category of $L_\infty$ algebra) to the Hochschild cohomology of $\mathrm{PMod}_V$.
\end{itemize}
The first conjecture, if true, should follow from some careful diagram-chasing. 

\section{The Planar global symmetry algebra}\label{sec:planar_symmetry_algebra}
We now specialize the general considerations to the case of the global symmetry algebras. We should first discuss planar corrections to the BRST cohomology of single-trace operators (ignoring multi-traces).


In this approximation, the BRST differential is the sum of two parts $Q_0$ and $\lambda Q_1$. The latter acts on two consecutive fields in the single-trace operator and maps it to a single field (with an extra derivative acting on it):
\begin{equation}
    \Phi(s) \otimes \Phi(s') \to \left(\frac{\Phi(s)-\Phi(s')}{s-s'} \eta_1^{(1)}\eta_2^{(1)}\right) \eta_2^{(2)}\otimes \eta_1^{(2)}  = \frac{\Phi(s)-\Phi(s')}{s-s'} \eta^{(1)}\otimes \eta^{(2)}\,.
\end{equation}


Overall, we obtain a deformation $\CC_\lambda^\bullet[A]$ of the 
cyclic cohomology complex $\CC^\bullet[A[[s]]]$. We propose this as a dg-TFT description of the collection of distributional local operators in the back-reacted world-sheet theory.

(Planar, single trace) cohomology classes will be polynomials in $\lambda$,
starting with a $Q_0$ cohomology class and continuing with corrections which compensate for the action of $\lambda Q_1$ on the leading term. It may or not be possible to complete a $Q_0$ class to a full class. We can express the obstruction as the action of $Q_1$ on $\HC^{\bullet}[A[[s]]]$.

\subsection{The planar mode algebra}
The modes in the planar algebra act on single-trace local operator at the linearized level as a sum of transformations which act on a sequence of consecutive symbols in the trace. We can thus represent them as elements in the planar algebra $\BC^{\bullet, \bullet}(A[[s]])$ and $\fL_\lambda$ as a dg-Lie sub-algebra of 
the planar algebra.

The planar part of the action of some $O_{n;C}$
mode is computed by taking $m+1$ Wick contractions 
of consecutive fields in $\cO_C$ and in the target. 
Concretely: 
\begin{equation}\label{eq:ad_planar_mode_action_with_contour}
\begin{aligned}
        &\Phi(s_1) \otimes \cdots \otimes \Phi(s_{m+1}) \\
        &\to \oint \frac{dz}{2 \pi i} z^n \left(C|\Phi(z+s),\cdots, \frac{\eta^{(1)}_{m+1}}{z+s-s_{m+1}},\cdots,\frac{\eta^{(1)}_{1}}{z+s-s_{1}}\right) \eta^{(2)}_1 \otimes \cdots \eta^{(2)}_{m+1}\,,
\end{aligned}
\end{equation}
i.e. 
\begin{equation}\label{eq:ad_planar_mode_action_without_contour_old}
    \Phi(s_1) \otimes \cdots \otimes \Phi(s_{m+1}) \to \sum_{n_i|n=m+\sum_i n_i}\left(C|\Phi(s),\cdots, E^{(1)}_{n_{m+1}},\cdots,E^{(1)}_{n_{1}}\right) E^{(2)}_{n_{1}} \otimes \cdots E^{(2)}_{n_{m+1}}\,.
\end{equation}
gives the action of $O^{(m)}_{n;C}$. In Appendix \ref{app:planar action}
 we provide a different expression for the planar action of the mode algera.
\subsection{The planar fundamental algebra}
By the same token, we can define the action of the modes of mesons, collected into the planar algebras $\fP_\lambda$ and in particular $\mathfrak{O}_\lambda(A)$. They are embedded into the category $\mathrm{PMod}_{A[[s]]}$.

Recall that the tree-level cohomology of mesonic operators is dual to $A_0^\vee \otimes_A A_0$ up to a ghost number shift. The mesons built from the $c$ ghost only are computed by $(A_0^\vee \otimes_{A_0} A_0)^\vee = A_0$ and are recognized as gauge-invariant bilinears of the schematic form $IJ$.

The mesons of the schematic form $I Z J$ are roughly labelled by elements of $A_1$. Each gives rise to two modes in $\mathfrak{O}_\lambda(A)$,
which are somewhat analogous to coordinates on the dual geometry. At tree level, the product of such elements lands on modes of mesons of the schematic form $IZZJ$. As in the canonical example, this leads to bilinear relations associated to the $Q_0$ image of $IbJ$ mesons. These bilinear relations are deformed at order $\lambda$ by $IJ$ zero modes. 

This is a non-commutative version of the geometric transition from the resolved conifold to $SL(2,\bC)$. 

\section{The non-commutative algebra at planar level} \label{sec:ncplanar}
We now go back to the example in Section \ref{sec:nctree} of the $U(N)$ gauge theory with $OSp(2n|2n+2)$
flavour symmetry. Recall that in this example $\omega^{ab} \omega_{ab} = -2$.

\subsection{The global algebra of meson modes.}
We introduce the notation 
\begin{equation}
    Z^{[a_1}Z^{a_2} \cdots Z^{a_m]}
\end{equation}
to denote a ``non-commutative traceless part'' of the product of $Z$'s. This is a linear combination of products which vanishes when we contract $\omega$ with any consecutive pairs of indices. These elements provide us a basis of the non-commutative algebra $A^{!} = \bC\langle\zeta_a\rangle/(\omega^{ab}\zeta_a\zeta_b = 0)$ that we find in Section \ref{sec:nctree}.

We can define it recursively:
\begin{align}
    Z^{[a}Z^{a_1} \cdots Z^{a_m]}&\equiv Z^a Z^{[a_1}Z^{a_2} \cdots Z^{a_m]}+ f_{1,m} \omega^{a a_1} \omega_{bc} Z^b Z^{[c}Z^{a_2} \cdots Z^{a_m]} \cr 
    &+ f_{2,m} \omega^{a_1 a_2} \omega_{bc} Z^b Z^{[c}Z^{a}Z^{a_3} \cdots Z^{a_m]}\cr 
    &+ f_{3,m} \omega^{a_2 a_3} \omega_{bc} Z^b Z^{[c}Z^{a}Z^{a_1}Z^{a_4} \cdots Z^{a_m]} \cr &+ \cdots
\end{align}
We need 
\begin{align} \label{eq:cartan}
    1 - 2 f_{1,m} + f_{2,m} &=0 \cr
    f_{1,m}- 2 f_{2,m} + f_{3,m} &= 0 \cr
    \cdots
\end{align}
i.e. $f_{i,m} =1-\frac{i}{m+1}$.

Explicitly,
\begin{align}
    Z^{[a_1}Z^{a_2]} &= Z^{a_1} Z^{a_2} + \frac12 \omega^{a_1 a_2} \omega_{b_1 b_2} Z^{b_1} Z^{b_2}  \label{eq_ad:Z2_traceless}\\
    Z^{[a_1}Z^{a_2}Z^{a_3]}&= Z^{a_1} Z^{a_2} Z^{a_3} + \frac23 \omega^{a_2 a_3} \omega_{b_1 b_2} Z^{a_1} Z^{b_1} Z^{b_2} + \frac23 \omega^{a_1 a_2} \omega_{b_1 b_2} Z^{b_1} Z^{b_2} Z^{a_3}+ \cr &+\frac13 \omega^{a_1 a_2} \omega_{b_1 b_2} Z^{a_3} Z^{b_1} Z^{b_2}+\frac13 \omega^{a_2 a_3} \omega_{b_1 b_2} Z^{b_1} Z^{b_2} Z^{a_1}  \, . \label{eq_ad:Z3_traceless}
\end{align}
Etcetera. 

Mesons built from this combinations:
\begin{equation}\label{eqn_towerf_nc}
     \frac{1}{\hbar} IZ^{[a_1} \cdots Z^{a_k]} J
\end{equation}
are $Q$-closed at the linearized planar level and not exact. They are explicit representatives for the abstract $A^!$ representatives in $Q_0$-cohomology and we expect to exhaust $\mathrm{Ops}^\partial_\lambda$.  

As explained in Section~\ref{subsec:open_symmetry_algebra}, the commutators between the modes of these mesons define the algebra of functions on the back-reacted non-commutative geometry. We denote them by
\begin{align}
u^{a_1...a_k}_{n_1...n_k} = \oint dz\, z^{n_1+...+n_k} \frac{1}{\hbar} IZ^{[a_1} \cdots Z^{a_k]} J\,,
\end{align}
where the bottom $n_i$ indices take value $0$ or $1$ and are symmetric, representing a spin $k/2$ representation of $SL(2,\bC)$ global conformal 
symmetry, and the top indices are non-commutative traceless. The two conditions can be expressed as the vanishing of contractions of consecutive indices with either $\omega_{a_i a_{i+1}}$ or $\epsilon^{n_i n_{i+1}}$.

The algebra they form can be obtained following lengthy yet conceptually simple 2d CFT calculations. At tree level, one simply concatenates two strings of $Z$'s and drops $\omega$ contractions as being $Q_0$ exact, so that 
\begin{equation}
    u^{a_1\cdots a_r}_{n_1 \cdots n_r} \cdot u^{a_{r+1}\cdots a_{r+s}}_{n_{r+1} \cdots n_{r+s}} = u^{a_1\cdots a_{r+s}}_{n_1 \cdots n_{r+s}} + O(\lambda)\,.
\end{equation}
Therefore, we can identify the tree level algebra as the subalgebra of $A^{!}[z]$ generated by $u^a_0 = \zeta^a,u^a_1 = \zeta^az$. This is a non-commutative generalization to the algebra of functions on the singular conifold.

Now we discuss non-planar corrections, which deform the tree level algebra and play a role analogous to the conifold transition. Each power of $\lambda$ reduces the number of $Z$'s by two and thus multiplies $u$'s with two fewer indices. 

A straightforward calculation gives 
\begin{equation}
    u^{a_1}_{n_1} \cdot u^{a_2}_{n_2} = u^{a_1 a_2}_{n_1 n_2} - \frac{\lambda}{2} \epsilon_{n_1n_2} \omega^{a_1 a_2}\,,
\end{equation}
and thus 
\begin{align}
    \omega_{a_1 a_2} u^{a_1}_{n_1} \cdot u^{a_2}_{n_2} &= \lambda \epsilon_{n_1 n_2}  \cr
     \epsilon^{n_1 n_2}u^{a_1}_{n_1} \cdot u^{a_2}_{n_2} &= \lambda \omega^{a_1 a_2}\,,
\end{align}
which nicely generalize the equation defining $SL(2,\bC)$ in the canonical example. 

These equations and constraints essentially fix the whole algebra. For example, consider the triple product ansatz
\begin{equation}
    u^{a_1}_{n_1} \cdot u^{a_2}_{n_2}\cdot u^{a_3}_{n_3} = u^{a_1 a_2 a_3}_{n_1 n_2 n_3} + \lambda \left[A_{n_1}^{a_1 a_2 a_3} \epsilon_{n_2 n_3}+B_{n_3}^{a_1 a_2 a_3} \epsilon_{n_1 n_2} \right]\,.
\end{equation}
We only need two terms on the right hand side because the tensor product of three fundamental representations of $SL(2,\bC)$ contains two copies of the fundamental representation. 

The ansatz must satisfy 
\begin{align}
     \epsilon^{n_1 n_2}u^{a_1}_{n_1} \cdot u^{a_2}_{n_2} \cdot u^{a_3}_{n_3}&= \lambda \left[A_{n_3}^{a_1 a_2 a_3}-2B_{n_3}^{a_1 a_2 a_3} \right] =\lambda \omega^{a_1 a_2} u^{a_3}_{n_3}\cr
     \epsilon^{n_2 n_3}u^{a_1}_{n_1} \cdot u^{a_2}_{n_2} \cdot u^{a_3}_{n_3}&= \lambda \left[-2A_{n_1}^{a_1 a_2 a_3} +B_{n_1}^{a_1 a_2 a_3} \right]=\lambda \omega^{a_2 a_3}u^{a_1}_{n_1}\,,
\end{align}
which implies 
\begin{equation}
    u^{a_1}_{n_1} \cdot u^{a_2}_{n_2}\cdot u^{a_3}_{n_3} = u^{a_1 a_2 a_3}_{n_1 n_2 n_3} -\frac{\lambda}{3} \left[\epsilon_{n_1 n_2}(2 \omega^{a_1 a_2} u^{a_3}_{n_3} + \omega^{a_2 a_3} u^{a_1}_{n_3})+\epsilon_{n_2 n_3}(\omega^{a_1 a_2} u^{a_3}_{n_1} + 2\omega^{a_2 a_3} u^{a_1}_{n_1}) \right] \, .
\end{equation}
In turn, we derive 
\begin{align}
    u^{a_1}_{n_1} \cdot u^{a_2 a_3}_{n_2 n_3} &= u^{a_1 a_2 a_3}_{n_1 n_2 n_3} -\frac{\lambda}{6} \left[2\epsilon_{n_1 n_2}(2 \omega^{a_1 a_2} u^{a_3}_{n_3} + \omega^{a_2 a_3} u^{a_1}_{n_3})+\epsilon_{n_2 n_3}(2\omega^{a_1 a_2} u^{a_3}_{n_1} + \omega^{a_2 a_3} u^{a_1}_{n_1}) \right] \cr
    u^{a_1 a_2}_{n_1 n_2}\cdot u^{a_3}_{n_3} &= u^{a_1 a_2 a_3}_{n_1 n_2 n_3} -\frac{\lambda}{6} \left[\epsilon_{n_1 n_2}(\omega^{a_1 a_2} u^{a_3}_{n_3} + 2 \omega^{a_2 a_3} u^{a_1}_{n_3})+2 \epsilon_{n_2 n_3}(\omega^{a_1 a_2} u^{a_3}_{n_1} + 2\omega^{a_2 a_3} u^{a_1}_{n_1}) \right] \,.\label{eq:ad_2x1_prod}
\end{align}
We can generalize these expressions, but it is useful to pick a different 
basis of fundamental representations in the tensor product, via the identity 
\begin{equation}
    \epsilon_{n_1 n_2} C_{n_3} +  \epsilon_{n_2 n_3} C_{n_1}+ \epsilon_{n_3 n_1} C_{n_2}=0\,,
\end{equation}
so that 
\begin{align}
    u^{a_1}_{n_1} \cdot u^{a_2 a_3}_{n_2 n_3} &= u^{a_1 a_2 a_3}_{n_1 n_2 n_3} -\frac{\lambda}{6} \left[\epsilon_{n_1 n_2}(2 \omega^{a_1 a_2} u^{a_3}_{n_3} + \omega^{a_2 a_3} u^{a_1}_{n_3})+\epsilon_{n_1 n_3}(2\omega^{a_1 a_2} u^{a_3}_{n_2} + \omega^{a_2 a_3} u^{a_1}_{n_2}) \right] \cr
    u^{a_1 a_2}_{n_1 n_2}\cdot u^{a_3}_{n_3} &= u^{a_1 a_2 a_3}_{n_1 n_2 n_3} -\frac{\lambda}{6} \left[\epsilon_{n_1 n_3}(\omega^{a_1 a_2} u^{a_3}_{n_2} + 2 \omega^{a_2 a_3} u^{a_1}_{n_2})+\epsilon_{n_2 n_3}(\omega^{a_1 a_2} u^{a_3}_{n_1} + 2\omega^{a_2 a_3} u^{a_1}_{n_1}) \right] \,.
\end{align}

Next, we can write an ansatz for a general product at order $O(\lambda^2)$:
\begin{equation}
    u^{a_1}_{n_1} \cdot \cdots \cdot u^{a_k}_{n_k} = u^{a_1 \cdots a_k}_{n_1 \cdots n_k} + \lambda \sum_{i=1}^{k-1} \epsilon_{n_i n_{i+1}} C(i)^{a_1 \cdots a_k}_{n_1 \cdots \hat n_i \hat n_{i+1} \cdots n_k}+ O(\lambda^2)
\end{equation}
and contract with $\epsilon^{n_j n_{j+1}}$ to get (after some relabeling of indices):
\begin{equation}
   \omega^{a_j a_{j+1}} u^{a_1 \cdots \hat a_j \hat a_{j+1} \cdots a_k}_{n_1 \cdots n_{k-2}} =    C(j-1)^{a_1 \cdots a_k}_{n_1 \cdots n_{k-2}}-2 C(j)^{a_1 \cdots a_k}_{n_1 \cdots n_{k-2}}+ C(j+1)^{a_1 \cdots a_k}_{n_1 \cdots n_{k-2}}\,,
\end{equation}
which is solved by inverting the $SU(k)$ Cartan matrix:
\begin{equation}
    C(i)^{a_1 \cdots a_k}_{n_1 \cdots n_{k-2}} = \sum_{j} \left[-\mathrm{min}(i,j) + \frac{ij}{k}\right]\omega^{a_j a_{j+1}} u^{a_1}_{n_1} \cdot \cdots \cdot u^{a_k}_{n_{k-2}}\,,
\end{equation}
so 
\begin{equation}
    u^{a_1}_{n_1} \cdot \cdots \cdot u^{a_k}_{n_k} = u^{a_1 \cdots a_k}_{n_1 \cdots n_k} + \lambda \sum_{i=1}^{k-1} \epsilon_{n_i n_{i+1}} \sum_{j=1}^{k-1} \left[-\mathrm{min}(i,j) + \frac{ij}{k}\right]\omega^{a_j a_{j+1}} u^{a_1 \cdots \hat a_j \hat a_{j+1} \cdots a_k}_{n_1 \cdots \hat n_i \hat n_{i+1} \cdots n_{k}} + O(\lambda^2)\, .
\end{equation}

With a bit of work, this expression is enough to derive a general formula for right multiplication 
\begin{equation}\label{eq:ad_full_right_proudct}
    u^{a_1 \cdots a_k}_{n_1 \cdots n_k} u^{a_{k+1}}_{n_{k+1}} = u^{a_1 \cdots a_{k+1}}_{n_1 \cdots n_{k+1}} - \frac{\lambda}{k(k+1)} \sum_{i=1}^{k} \epsilon_{n_{i} n_{k+1}} \sum_{j=1}^{k} j\omega^{a_j a_{j+1}} u^{a_1 \cdots \hat a_j \hat a_{j+1} \cdots a_{k+1}}_{n_1 \cdots \hat n_{i} \cdots n_{k}} \, .
\end{equation}
We check in Appendix \ref{app:non_com_OPE} that the above formula can also be verified by computing OPE of the mesonic operators.

The left multiplication is given by:
\begin{align}
u^{a_0}_{n_0}\cdot u^{a_1 \cdots a_k}_{n_1 \cdots n_k} = u^{a_0 \cdots a_k}_{n_0 \cdots n_k} - \frac{\lambda}{k(k+1)}\sum_{i=1}^{k}\epsilon_{n_0 n_i}\sum_{j=1}^{k}(k-j+1)\, \omega^{a_{j-1} a_{j}}u_{n_1 \cdots\hat{n_i}\cdots n_k}^{a_0 \cdots\hat{a}_{j-1} \hat{a}_{j} \cdots a_k}\,.
\end{align}

We will also need
\begin{align}
    & \epsilon^{n_k n_{k+1}}u^{a_1 \cdots a_k}_{n_1 \cdots n_k} u^{a_{k+1}}_{n_{k+1}} =\frac{\lambda}{k+1} \sum_{j=1}^{k} j\omega^{a_j a_{j+1}} u^{a_1 \cdots \hat a_j \hat a_{j+1} \cdots a_{k+1}}_{n_1 \cdots n_{k-1}} \cr
    & \epsilon^{n_0 n_1} u^{a_0}_{n_0}\cdot u^{a_1 \cdots a_k}_{n_1 \cdots n_k} = \frac{\lambda}{k+1}\sum_{j=1}^{k}(k-j+1)\, \omega^{a_{j-1} a_{j}}u_{n_2\cdots n_k}^{a_0 \cdots\hat{a}_{j-1} \hat{a}_{j} \cdots a_k}\,,
\end{align}
as well as 
\begin{align}
    & \omega_{a_k a_{k+1}} u^{a_1 \cdots a_k}_{n_1 \cdots n_k} u^{a_{k+1}}_{n_{k+1}} = \frac{\lambda}{k} \sum_{i=1}^{k} \epsilon_{n_{i} n_{k+1}} u^{a_1 \cdots a_{k-1}}_{n_1 \cdots \hat n_{i} \cdots n_{k}} \cr
    & \omega_{a_0 a_{1}} u^{a_0}_{n_0}\cdot u^{a_1 \cdots a_k}_{n_1 \cdots n_k} =\frac{\lambda}{k}\sum_{i=1}^{k}\epsilon_{n_0 n_i} u_{n_1 \cdots\hat{n_i}\cdots n_k}^{a_2 \cdots a_k} \,.
\end{align}
\subsection{Single traces}
According to our discussion in Section \ref{sec:nctree}, the single-trace operators in this example also contain the four towers $\mathcal{A}, \mathcal{B}, \mathcal{C},$ $\mathcal{D}$, but with the $\mathcal{B}$ and $\mathcal{D}$ towers collapsed. For the $\mathcal{A}$ tower, one may attempt to build BRST-closed single-trace operators in a similar manner as for mesons, starting from a cyclic string of $Z$'s and removing traces. The analogue of Equations (\ref{eq:cartan}), though, is governed by the affine A-type Cartan matrix and cannot be solved. A single-trace operator which is closed at the linear planar level can be written as
\begin{equation}
    \frac{1}{m \hbar} A_{a_1 \cdots a_m} \Tr Z^{a_1} \cdots Z^{a_m}\,,
\end{equation}
where $A_{a_1 \cdots a_m}$ can be taken to be graded-cyclic symmetric and has to satisfy
\begin{equation}
    \omega^{a_1 a_2} A_{a_1 \cdots a_m}=0 \, .
\end{equation}

We can compute the action of a mode $L_i[A]$ of such an operator on an $IZJ$ meson. Only one Wick contraction is available, so only the zero mode acts non-trivially, giving
\begin{equation}
    [L_0[A], I Z^b J] = A_{a_1 \cdots a_{m-1} a_m} \omega^{a_m b} I Z^{a_1} \cdots Z^{a_{m-1}}J \, ,
\end{equation}
leading to
\begin{equation}
    [L_{n_1 \cdots n_{m-2}}[A], u^b_j] = A_{a_1 \cdots a_{m-1} a_m} \omega^{a_m b} u^{a_1 \cdots a_{m-1}}_{n_1 \cdots n_{m-2}j}\,.
\end{equation}

We expect this to be an infinitesimal automorphism of $\fP_\lambda$, but an explicit check takes a bit of work. We need to verify that 
\begin{align}
    [L_{n_1 \cdots n_{m-2}}[A], \epsilon^{j_1 j_2} u^{b_1}_{j_1} \cdot u^{b_2}_{j_2} ] &=0 \cr
    [L_{n_1 \cdots n_{m-2}}[A], \omega_{b_1 b_2} u^{b_1}_{j_1} \cdot u^{b_2}_{j_2} ] &=0\,.
\end{align}
Most of the terms in the first line drop off due to contractions of $A$ with $\omega$. The two remaining terms cancel thanks to cyclic invariance of $A$.

For the second identity,
\begin{align}
    &[L_{n_1 \cdots n_{m-2}}[A], \omega_{b_1 b_2}u^{b_1}_{j_1} \cdot u^{b_2}_{j_2} ] = \cr &= A_{a_1 \cdots a_{m-1} a_m} u^{a_1 \cdots a_{m-1}}_{n_1 \cdots n_{m-2}j_1} \cdot u^{a_m}_{j_2} -A_{a_1 \cdots a_{m-1} a_m} u^{a_m}_{j_1} \cdot u^{a_1 \cdots a_{m-1}}_{n_1 \cdots n_{m-2}j_2}\,,
\end{align}
the part linear in lambda drops off due to contractions of $A$ with $\omega$, the rest due to cyclic invariance. 

Commutators of such transformations will generate a large automorphism algebra, which may include modes from the second $\mathcal{D}$ tower of single-trace operators \footnote{As we have discussed in Section \ref{sec:nctree}, most of the $\mathcal{D}$ tower operators vanish. This leads to intricate constraints on the OPE/commutator of $\mathcal{A}$ tower operators,which would be interesting to explore further.}. For example, consider the cubic generator 
\begin{align}
    [L_{n_1}[A], u^b_j] &= A_{a_1 a_2 a_3} \omega^{a_3 b} u^{a_1 a_2}_{n_1 j} \cr 
    [L_{n_1}[A], u^{b_1 b_2}_{j_1 j_2} ]
    &= A_{a_1 a_2 a_3} \omega^{a_3 b_1}u^{a_1 a_2 b_2}_{n_1 j_1 j_2} + A_{a_1 a_2 a_3} \omega^{a_3 b_2}u^{b_1 a_1 a_2}_{j_1 j_2 n_1} + \cr
    &-\frac{\lambda}{3}A_{a_1 a_2 a_3} \omega^{a_3 b_1}  \omega^{a_2 b_2} \epsilon_{n_1 j_1}u^{a_1}_{j_2} -\frac{\lambda}{3} A_{a_1 a_2 a_3}  \omega^{a_1 b_1}\omega^{a_3 b_2} \epsilon_{n_1 j_2} u^{a_2}_{j_1}\,. 
\end{align}
Then 
\begin{align}
    &[[L_{(n_1}[A],L_{n_2)}[A']], u^b_j] = [A,A'] \omega^{a_4 b} u^{a_1 a_2 a_3}_{n_1 n_2 j} -\frac{\lambda}{6}\omega^{b_3 b} (\epsilon_{n_1 j} u^{a_1}_{n_2}+\epsilon_{n_2 j} u^{a_1}_{n_1} ) \cdot \cr &\cdot\left[A_{a_1 a_2 a_3} \omega^{a_3 b_1}  \omega^{a_2 b_2} A'_{b_1 b_2 b_3} - A'_{a_1 a_2 a_3} \omega^{a_3 b_1}  \omega^{a_2 b_2} A_{b_1 b_2 b_3}\right] \,,
\end{align}
with
\begin{equation}
    [A,A'] \equiv A_{a_1 a_2 b_1} \omega^{b_1 b_2} A'_{b_2 a_3 a_4} + A_{a_2 a_3 b_1} \omega^{b_1 b_2} A'_{b_2 a_4  a_1}+\omega^{b_1 b_2} A_{a_3 a_4 b_1}A'_{b_2 a_1 a_2} +\omega^{b_1 b_2} A_{a_4 a_1 b_1} A'_{b_2 a_2 a_3} 
\end{equation}
is cyclic invariant but not traceless and
\begin{align}
   & [\epsilon^{n_1 n_2}[L_{n_1}[A],L_{n_2}[A']], u^b_j] = \cr &= \lambda \left[A_{a_1 a_2 a_3} \omega^{a_3 b_1}  \omega^{a_2 b_2} A'_{b_1 b_2 b_3} + A'_{a_1 a_2 a_3} \omega^{a_3 b_1}  \omega^{a_2 b_2} A_{b_1 b_2 b_3} \right]\omega^{b_3 b} u^{a_1}_{j} 
\end{align}
is an $OSp$ infinitesimal rotation.

In order to process this expression further, notice that the BRST image of a $\Tr b Z^{[a_1} \cdots Z^{a_m]}$ operator gives the sum of a single-trace operator and a meson, traced over flavour indices. The single-trace operator would be exact in the absence of flavours. In the presence of flavours, its modes act on the meson BRST cohomology in the same way as the traced meson, i.e. by conjugation by the corresponding element of $\fP_\lambda$.

Accordingly, the single-trace cohomology in ghost number $0$ acts on $\fP_\lambda$ as derivations of $\fP_\lambda$ modulo inner derivations. This is the map to $\HH^{0}(\fP_\lambda, \fP_\lambda)$. For example, the inner derivations
\begin{align}
    [ u^{a_1 a_2}_{n_1 n_2}, u^b_j] &= u^{a_1 a_2 b}_{n_1 n_2 j} -u^{b a_1 a_2}_{n_1 n_2 j} -\frac{\lambda}{3} (\omega^{a_1 a_2} \delta^{b}_{c} +  \omega^{a_2 b} \delta^{a_1}_{c}+ \omega^{b a_1} \delta^{a_2}_{c})(\epsilon_{n_1 j}u^{c}_{n_2}+\epsilon_{n_2 j}u^{c}_{n_1})
\end{align}
can be added to the commutator of two cubic generators to simplify the answer and identify a global $SL(2,\bC)$ generator from the stress tensor contribution. 

We could continue the analysis further to map modes of the $\mathcal{C}$ tower $\Tr \partial c Z^{[a_1} \cdots Z^{a_m]}$ to deformations of $\fP_\lambda$, etcetera. 

\subsection{Determinants and modules}
Consider next a determinant operator of the form 
\begin{equation}
    \det (m + \mu_a Z^a)
\end{equation}
for some vector $\mu_a$ and basic open modifications of the form $I \psi$.

If we act on the modified determinant with $u^a_1$ we need a Wick contraction with the determinant itself, leading to the sort of term 
\begin{equation}
    \omega^{ab} \mu_b I \psi (\bar \psi \psi) \sim m^{-1} \omega^{ab} \mu_b I \psi\,,
\end{equation}
so that 
\begin{equation}
    u^a_1 I \psi = m^{-1} \omega^{ab} \mu_b I \psi\,.
\end{equation}
The action of $u^a_0$ produces more complicated modifications $I Z^a \psi + \cdots$, unless we contract with $\mu$:
\begin{equation}
   (\mu_a u^a_0) I \psi = - m I \psi\,.
\end{equation}
This is compatible with the algebra: 
\begin{equation}
  u^a_1 (\mu_b u^b_0) I \psi - u^a_0 (\mu_b u^b_1) I \psi= \omega^{ab} \mu_b I \psi\,.
\end{equation}

Next, consider 
\begin{equation}
    u^{a_1 a_2}_{n_1 1} I \psi = m^{-1} \omega^{a_2 b} \mu_b u^{a_1}_{n_1} I \psi + \frac{\lambda}{2} \epsilon_{n_1 1} \omega^{a_1 a_2} I \psi\,.
\end{equation}
More generally, the action of any generator except $u^{a_1 \cdots a_k}_{0 \cdots 0}$ can be recursively expressed in terms of the collection of modifications 
\begin{equation}
    u^{a_1 \cdots a_k}_{0 \cdots 0} I \psi\,,
\end{equation}
which is traceless for consecutive indices and satisfies 
\begin{equation}
    \mu_{a_k} u^{a_1 \cdots a_k}_{0 \cdots 0} I \psi = - m u^{a_1 \cdots a_{k-1}}_{0 \cdots 0} I \psi \,.
\end{equation}
We expect this to be a basis of $\fM_\lambda[\mu,m]$.

If we move the determinant to a location $z$, we have $\fM_\lambda[\mu,m,z]$:
\begin{equation}
    u^a_1 I \psi = m^{-1} \omega^{ab} \mu_b I \psi + z u^a_0 I \psi \, .
\end{equation}

The analysis is easily generalized to multiple determinants. We replace $m$ with a matrix $\rho$ and introduce diagonal matrices $z$ and $\mu$ representing positions and orientations of the determinants. We promote $\psi$ to a vector acted from the right by these matrices. We get 
\begin{align}
    u^a_1 I \psi &=  \omega^{ab} I \psi \mu_b \rho^{-1} + u^a_0 I \psi z\,, \cr
    u^a_0 I \psi \mu_a &= - I \psi \rho\,.
\end{align}

If we impose 
\begin{equation}
  u^a_1 u^b_0 I \psi \mu_b - u^a_0 u^b_1 I \psi \mu_b = \omega^{ab} I \psi \mu_b\,,
\end{equation}
we get 
\begin{equation}
  -(\omega^{ab} I \psi \mu_b \rho^{-1} + u^a_0 I \psi z) \rho - u^a_0 (\omega^{bc} I \psi \mu_c \rho^{-1} + u^b_0 I \psi z) \mu_b = \omega^{ab} I \psi \mu_b\,,
\end{equation}
i.e. we recover the saddle equations
\begin{equation}
  \omega^{bc} \mu_c \rho^{-1}  \mu_b= [\rho, z]
\end{equation}
as a condition for $\fM_\lambda[\mu,\rho,z]$ to exist. This module is the non-commutative analogue of the spectral curve. 

\section{Conclusion and open questions} \label{sec:conclude}
We employed the global symmetry algebra $\fP_\lambda$ of mesonic operators as a way to generalize the notion of the algebra of holomorphic functions on $SL(2,\bC)$ and ascribe an holographic dual nc-geometry to a generic chiral algebra which admits a 't Hooft expansion. 

The category of $\fP_\lambda$-modules seems a good description of a category of D-branes in the dual nc-geometry. In particular, it contains modules $\fM_\lambda$ associated to saddle points of correlation functions of determinant operators, i.e. D-branes which ``reach the holographic boundary'' at a collection of points. 

We leave the following points unresolved: 
\begin{enumerate}
    \item We expect $\fP_\lambda$ to be a smooth 3d Calabi-Yau algebra. A proof would require one to present an explicit class in $\HH_3(\fP_\lambda)$ (and should be an element in the negative cyclic homology $\HC_3^{-}(\fP_\lambda)$). This is a cyclic element with four entries valued in $\fP_\lambda$. It is not hard to write such a class for $SL(2,\bC)$ from the volume form and guess a generalization. We do not know, though, how to derive it from the chiral algebra. 
    \item We expect the category of $A_\infty$ modules of $\fP_\lambda$ to capture the category of branes in the dual nc-geometry. Essentially by construction, we have an $A_\infty$ morphism from the space of ``determinant modifications'' of a determinant operator in the chiral algebra to the space of $A_\infty$ endomorphisms of the corresponding module $\fM_\lambda$. Ideally, we would like to show that this is a quasi-isomorphism and generalize the statement to any pairs of determinants. 
    \item We expect the divergence-free part of $\HH^{\bullet}(\fP_\lambda, \fP_\lambda)$ to match the global symmetry algebra $\fL_\lambda$ of single-trace operators. We only have $L_\infty$ morphisms from the latter to the former. 
\end{enumerate}

It seems a harder problem to derive from the data of $\fP_\lambda$ the planar cohomology of single-trace operators or even the planar cohomology of mesons, i.e. the spaces of distributional vertex operators which build up the holographic dictionary. Indeed, quantities such as $\HC^{\bullet}(\fP_\lambda)$ capture fully distributional closed string states rather than the ones we need. 

Determinant operators and giant graviton branes may allow one to side-step this challenge in two related ways:
\begin{itemize}
    \item The cyclic cohomology of the endomorphisms of a giant graviton brane, or of a category of giant graviton branes, has the correct properties to recover the closed string states which appear as boundary-to-bulk holographic propagators.
    \item If we take $m \to \infty$ and expand a (possibly modified) determinant operator in inverse powers of $m$, the outcome is a sequence of BRST-closed multi-trace operators. In particular, the sub-leading term is single-trace. This is another way to relate closed string states to the cyclic homology of determinant modifications. 
\end{itemize}
We leave these problems to future work.

Our general strategy can be readily applied to other examples of twisted holography \cite{Costello:2020jbh,Fernandez:2024tue}. Low-hanging targets are gauged quantum-mechanical systems of large matrices \cite{Ishtiaque:2018str} and matrix models. 

\section*{Acknowledgements}
We would like to thank Kasia Budzik, Kevin Costello, Victor Ginzburg, Yan Soibelman and Ben Webster for useful discussions. This research was supported in part by a grant from the Krembil Foundation. DG, ALR and HS are supported by the NSERC Discovery Grant program and by the Perimeter Institute for Theoretical Physics. KZ is supported by Harvard University CMSA. Research at Perimeter Institute is supported in part by
the Government of Canada through the Department of Innovation, Science and Economic
Development Canada and by the Province of Ontario through the Ministry of Colleges and
Universities.

\appendix

\section{The Planar Mode Action: An Expression for the Contour Integral}
\label{app:planar action}
In Section~\ref{sec:planar_symmetry_algebra} we defined the expression for the planar-linear action of modes $O_{n;C}$ on a tensor product of fields \eqref{eq:ad_planar_mode_action_with_contour}, \eqref{eq:ad_planar_mode_action_without_contour_old}. In this appendix we provide a different expression for the planar mode action after one performs the contour integral. Recall that the planar action of a mode\footnote{In this section we'll use $l_i$ to refer to some derivative counters as it simplifies OPE expressions below.}
\begin{align}
O_{n;C} = \oint dz \, (C|\Phi(z+l_{|C|},) \cdots,\Phi(z+l_1))
\end{align}
is given by a polynomial in the 't Hooft coupling $\lambda$. 
\begin{align}
[O_{n;C}, \Phi(s_1)\otimes \cdots\Phi(s_k)] &= O^{(0)}_{n;C}(\Phi(s_1)\otimes \cdots\Phi(s_k)) + \lambda O^{(1)}_{n;C}(\Phi(s_1)\otimes \cdots\Phi(s_k)) +\cdots\nonumber\\
&+ \lambda^m O^{(m)}_{n;C}(\Phi(s_1)\otimes \cdots\Phi(s_k))\,,
\end{align}
where the $O^{(m)}_{n;C}$ contains $m+1$ wick contractions and is given by
\begin{align}
&O^{(m)}_{n;C},(\Phi(s_1) \otimes \cdots \otimes \Phi(s_{m+1})) = \nonumber\\
&\oint \frac{dz}{2 \pi i} z^n \left(C|\Phi(z+l_{|C|}),\cdots, \frac{\eta^{(1)}_{m+1}}{z+l_{m+1}-s_{m+1}},\cdots,\frac{\eta^{(1)}_{1}}{z+l_1-s_{1}}\right) \eta^{(2)}_1 \otimes \cdots \eta^{(2)}_{m+1}\,.\label{eq:planar_action_m_contractions_appendix}
\end{align}



To provide an expression for this contour integral we note that integrals of the form
\begin{align}
\oint dz \frac{f(z)}{(z-z_1)(z-z_2)\cdots(z-z_k)}\,,
\end{align}
where all poles are enclosed by the contour, can be expressed in terms of a ratio of Vandermonde-like determinants
\begin{align}
\int dz \frac{f(z)}{(z-z_1)(z-z_2)\cdots(z-z_k)} = \frac{D_k(f(z))}{D_k(z^{k-1})}\,,
\end{align}
with $D_{k}(f(z))$ the determinant
\begin{align}
D_{k}(f(z)) \equiv 
\begin{vmatrix}
     & f(z_k) & z_k^{k-2} & z_k^{k-3} & \cdots & z_k & 1\\
     & f(z_{k-1}) & z_{k-1}^{k-2} & z_{k-1}^{k-3} & \cdots & z_{k-1} & 1\\
     & \vdots &           &\ddots     &        &     & \vdots\\
     & f(z_1) & z_1^{k-2} & z_1^{k-3} & \cdots & z_1 & 1
\end{vmatrix}\,.
\end{align}

We will call this ratio $R_k(f(z)) \equiv \frac{D_{k}(f(z))}{D_{k}(z^{k-1})}$. 

In Equation~\eqref{eq:planar_action_m_contractions_appendix} the corresponding $f(z)$ is
\begin{align}
f(z) = z^n \Phi(z+l_{|C|})\otimes\cdots\otimes \Phi(z+l_{m+2})\,,
\end{align}
and the poles are at $z_i = s_i - l_i$ for $i=1,..,m+1$. Therefore, after performing the contour integral we get
\begin{align}
&O^{(m)}_{n;C},(\Phi(s_1) \otimes \cdots \otimes \Phi(s_{m+1})) =\nonumber\\
&\left(C|R_{m+1}\left(z^n \Phi(z+l_{|C|})\otimes\cdots\otimes \Phi(z+l_{m+2})\right), \eta^{(1)}_{m+1},\cdots,\eta^{(1)}_{1}\right) \eta^{(2)}_1 \otimes \cdots \eta^{(2)}_{m+1}\,.
\end{align}

To give a more concrete idea of the meaning of this expression, let's write it explicitly for the case with two wick contractions and a $C$ that takes three arguments, $|C| = 3$. In this case, the action is given by
\begin{align}
&O^{(1)}_{n;C},(\Phi(s_1) \otimes \Phi(s_{2})) =\nonumber\\
&\left(C|R_2\left(z^n \Phi(z+l_{3})\right), \eta^{(1)}_{2},\eta^{(1)}_{1}\right) \eta^{(2)}_1 \otimes \eta^{(2)}_{2}\,,
\end{align}
where $R_2(z^n\Phi(z+l_{3}))$ is explicitly given as
\begin{align}
R_2(z^n\Phi(z+l_{3})) &= \frac{
\begin{vmatrix}
(s_2-l_2)^n\Phi(s_2-l_2+l_{3}) & \quad 1\\
(s_1-l_1)^n\Phi(s_1-l_1+l_{3}) & \quad 1
\end{vmatrix}
}{
\begin{vmatrix}
s_2-l_2 & \quad 1\\
s_1-l_1 & \quad 1
\end{vmatrix}
}\\
&=\frac{(s_2-l_2)^n\Phi(s_2-l_2+l_{3}) - (s_1-l_1)^n\Phi(s_1-l_1+l_{3})}{(s_2-l_2)- (s_1-l_1)} \,.
\end{align}

The planar action of single trace modes on mesons, and of mesons on mesons can be similarly expressed in terms of $R_m$. We'd find correspondingly
\begin{align}
&O^{(m)}_{n;C}(I(s_0)\otimes\Phi(s_1) \otimes \cdots \otimes J(s_{m+2})) =\nonumber\\
&I(s_0)\otimes\left(C|R_{m+1}\left(z^n \Phi(z+l_{|C|})\otimes\cdots\otimes\Phi(z+l_{m+2})\right), \eta^{(1)}_{m+1},\cdots,\eta^{(1)}_{1}\right) \eta^{(2)}_1 \otimes \cdots \eta^{(2)}_{m+1} \otimes J(s_{m+2})\,,
\end{align}
and
\begin{align}
&O^{(m)}_{n;P}(I(s_1)\otimes\Phi(s_2) \otimes \cdots \otimes J(s_{m+2})) =\nonumber\\
&\left(C|R_{m+1}\left(z^n I(z+l_{|C|})\otimes\cdots\otimes\Phi(z+l_{m+2})\right), \eta^{(1)}_{m+1},\cdots,\mu^{(1)}_{1}\right) \mu^{(2)}_1 \otimes \cdots \eta^{(2)}_{m+1} \otimes J(s_{m+1})\,.
\end{align}

\section{Cyclic cohomology of a weighted algebra}
\label{appendix:vanish_S}
In this section, we discuss properties of cyclic cohomology of a weighted algebra. In this paper, the weight of the algebra is provided by their scaling dimension. In particular, we would like to prove the statement that the Connes' periodicity map $S$ vanishes on the positive weight part of cyclic cohomology. We follow the discussion in \cite{loday2013cyclic}, where the homology version of this statement is proved.

First recall that a derivation (of degree $0$) of an algebra $A$ into  itself is a map $D:A\to A$ such that $D(ab) = D(a)b + aD(b)$. One can extend the derivation on $A$ to Hochschild cochain $CC^{\bullet}(A) := CC^{\bullet}(A,A^*)$ via the formula
\begin{equation}
	(L_Df)(a_0,a_1,\dots,a_n) = \sum_{i \geq 0} f(a_0,\dots,Da_i,\dots,a_n).
\end{equation}  
We can check that the map $L_D$ commutes with the Hochschild differential and $t$. Therefore, it induce a map on the Hochschild cohomology and cyclic cohomology.
\begin{equation}
	L_D:HC^{\bullet}(A) \to HC^{\bullet}(A)\,.
\end{equation}
We further introduce operators $e_D:CC^{n}(A) \to CC^{n+1}(A)$ and  $E_D:CC^{n}(A) \to CC^{n-1}(A)$ as follows
\begin{equation}
	\begin{aligned}
			(e_Df)(a_0,\dots,a_n) &= (-1)^{n+1}f(D(a_n)a_0,a_1,\dots,a_{n-1})\\
			(E_Df)(a_0,\dots,a_{n-2}) & = \sum_{1\leq i\leq j\leq n-2} (-1)^{in+1}f(1,a_i,\dots,a_{j-1},D(a_j),a_{j+1},\dots,a_{n-2},a_0,\dots,a_{i-1})\,.
	\end{aligned}
\end{equation}
Then one can check the following identity
\begin{theorem}
	\begin{equation}
		\begin{aligned}
			&[e_D,Q_{\CH}] = 0\,,\\
			&[e_D,B] + [E_D,Q_{\CH}] = L_D\,,\\
			&[E_D,B] = 0\,.
		\end{aligned}
	\end{equation}
\end{theorem}
From the above identities, we have the following results
\begin{theorem}
	The map $L_D\circ S = 0: HC^{\bullet}(A) \to HC^{\bullet+2}(A)$.
\end{theorem}
\begin{proof}
	We use the bicomplex $\mathcal{B}(A) = (CC^{\bullet}(A)[u],Q_{\CH}+uB)$ for the cyclic cohomology. The periodicity map $S = u$ in this case. Now we construct a homotopy map $h: \mathcal{B}(A) \to \mathcal{B}(A)[1]$ defined as
	\begin{equation}
		h = e_D+uE_D\,.
	\end{equation}
	Then we find
	\begin{equation}
		[h,Q_{\CH}+uB] = [e_D,Q_{\CH}] + u([e_D,B]+[E_D,Q_{\CH}]) + u^2[E_D,B] = uL_D\,.
	\end{equation}
	Therefore, $h$ is the homotopy between $L_D\circ S$ and $0$. Thus $L_D\circ S = 0$ on cohomology.
\end{proof}

Now for a weighted algebra $A$, we define a derivation $D$ on $A$ by $D(a) = wt(a)a$. Then $L_D$ acting on a Hochschild cochain $f$ is simply $L_Df = wt(f)f$. The above results tells us that $S$ must be $0$ on positive weight part of the cyclic cohomology. As a corollary, we have
\begin{theorem}
	Let $A$ be a unital weighted algebra. Define
	\begin{equation}
		HC^{n}(A)^{(\geq 1)} = \bigoplus_{w\geq1}HC^{n}(A)^{(w)} \cong HC^{n}(A)/HC^{n}(A_0)\,.
	\end{equation}
	Then for $HC^{\bullet}(A)^{(\geq 1)}$, the Connes' long exact sequence reduces into a collection of short exact sequences
	\begin{equation}
		0\longrightarrow HC^{n}(A)^{(\geq 1)} \overset{I}{\longrightarrow} HH^n(A)^{(\geq 1)} \overset{B}{\longrightarrow} HC^{n-1}(A)^{(\geq 1)} \longrightarrow 0\,.
	\end{equation}
\end{theorem}

\section{Global symmetry algebra}
    \label{sec:global_sym}
In this section, we discuss the mathematical construction related to the global symmetry algebra of a vertex algebra. To simplify the discussion, we consider the case when our vertex algebra arises from the chiral envelope of a vertex Lie algebra. In our example of the large $N$ chiral algebra, the tree-level planar limit defines a vertex Lie algebra structure on $L = \mathrm{HC}^{\bullet}(A[[s]])$. However, the discussion in this appendix applies to any (conformal) vertex Lie algebra.

Let us first recall some construction of Lie algebra attached to a vertex Lie algebra. Given a vertex Lie algebra $(L,T,Y_-)$, we first have its Lie algebra of modes \cite{frenkel2001vertex}
\begin{equation}
	\oint L := L\otimes \mathbb{C}((t))/\mathrm{Im}(T\otimes\mathrm{Id}+ \mathrm{Id}\otimes \partial_t)\,.
\end{equation}
Physically, $\oint L$ consist of the modes $O_{n}$ of the fields $O(z) = \sum_{n\in\mathbb{Z}}O_n z^{-n-1}$ in the vertex algebra. We are particularly interested in the non negative modes in this Lie algebra. We denote
\begin{equation}
	(\oint L)_{+} := \{O_{n} \in \oint L\mid n\geq 0 \}.
\end{equation}
This forms a Lie sub-algebra of the mode algebra $\oint L$. In the example $L = \mathrm{HC}^{\bullet}(A[[s]])$, $(\oint L)_{+}$ is the extended symmetry algebra $\mathfrak{L}^{\mathbb{C}}$ defined in Section \ref{sec:GCA_standard}, \ref{sec:general_GCA}.

We are interested in the case when our vertex algebra has a stress-energy tensor. In such cases, there is a general procedure that allows us to pass from a conformal vertex Lie algebra to a vertex algebra bundle on a curve, which is a $\mathcal{D}$-module 
	\begin{equation}
		L \rightsquigarrow \mathcal{L}\,.
	\end{equation}
	Such a $\mathcal{D}$-module $\mathcal{L}$ is called a (chiral) Lie* algebra. Recall that we have the de Rham functor $h$ that send a $\mathcal{D}$-module to $\mathcal{O}$-module, given by $h(\mathcal{L}) = \mathcal{L}\otimes_{\mathcal{D}_X}\mathcal{O}_X$. Then, according to \cite{frenkel2001vertex,beilinson2004chiral}, we have the following geometric description of the Lie algebra $\oint L$ and $(\oint L)_{+}$.
	\begin{theorem}
		Let $D$ be the standard disc and $D^{\times} = D - \{0\}$ the punctured disc. We have
		\begin{equation}
			\oint L \cong \Gamma(D^{\times},h(\mathcal{L})),\quad\quad		(\oint L)_{+} = \Gamma(D,h(\mathcal{L})).
		\end{equation}
	\end{theorem}
	Given the stress-energy tensor, we have, in particular, an $\mathfrak{sl}_2 = \{L_{-1},L_0,L_1\}$ action on the Lie algebra $\oint L $ and $(\oint L)_{+}$. Given a $\mathfrak{sl}_2$ module $M$ with bounded integer weights, we define a submodule $\mathrm{Core}(M)$
	\begin{equation}
		\mathrm{Core}(M) = \{m \in M \mid L_1^N(m)  = 0 \text{ for some } N\}.
	\end{equation}
	In our case $\mathrm{Core}((\oint L)_{+})$ is a Lie subalgebra. This Lie algebra is defined to be the global symmetry algebra of the vertex algebra. We have the following geometric description of $\mathrm{Core}((\oint L)_{+})$	
	\begin{theorem}
		\begin{equation}
					\mathrm{Core}((\oint L)_{+}) = \Gamma(\bC\mathbb{P}^1,h(\mathcal{L}))\,.
		\end{equation}
	\end{theorem}
This Lie algebra $\Gamma(\bC\mathbb{P}^1,h(\mathcal{L}))$ is defined to be the global symmetry algebra of the vertex Lie algebra $L$. It coincide with the physical definition as the Lie algebra of modes that annihilate both the vacuum at $0$ and $\infty$. In the example $L = \mathrm{HC}^{\bullet}(A[[s]])$, $\Gamma(\bC\mathbb{P}^1,h(\mathcal{L}))$ is the global symmetry algebra $\mathfrak{L}$ we studied in this paper.

\section{Calabi-Yau algebra}
\label{appendix:CY}
In this appendix, we briefly review the mathematical definition of a Calabi-Yau algebra. There are two distinct notions of Calabi-Yau structures, called smooth Calabi-Yau \cite{Ginzburg:2006fu} and compact Calabi-Yau \cite{Kontsevich2008NotesOA}, which are related by Koszul duality \cite{bergh2010calabi,cohen2015}. Both of these notions give rise to (partially defined) $2d$ topological quantum field theories in the sense of \cite{costello2007topological,lurie2008classification}. Typically, a compact Calabi-Yau structure defines a TQFT that assigns values to cobordisms with at least one input, while a smooth Calabi-Yau structure gives rise to a TQFT defined on cobordisms with at least one output \cite{kontsevich2023smooth}. Geometrically, these two notions of Calabi-Yau algebras correspond to a space-filling brane and a brane supported at a point, respectively.

We introduce some notation. For $B$ an associative algebra, we denote $B^{op}$ the opposite algebra. It has the same elements as $B$ but equipped with a opposite multiplication $a\cdot^{op} b = b\cdot a$. We denote $B^{e} = B\otimes B^{op}$ so that a $B$ bimodule is the same as a $B^{e}$ left module. 

We first introduce the notion of a smooth Calabi-Yau algebra \cite{Ginzburg:2006fu}, as it is more commonly used in the literature.

Let $B$ be a homologically smooth algebra. We denote $B^D$ its (derived) dual bimodule, defined as
\begin{equation}
	B^D := R\mathrm{Hom}_{B^e}(B,B^e)\,.
\end{equation}

Recall that we have the following isomorphism in the derived category
\begin{equation*}
	R\mathrm{Hom}_{B^e}(B^D,B) \cong B\otimes_{B^e}^{\mathbb{L}}B\,.
\end{equation*}

We call $B$ a (weak) smooth $d$-Calabi-Yau algebra if there exists a map $\varphi \in \mathbb{C}[d] \to \CH_{\bullet}(B)$ that induces an isomorphism 
\begin{equation}
	\varphi: B^D \cong B[-d]
\end{equation}
in the derived category of $B^e$ module. 

The Hochschild complex $\CH_{\bullet}(B)$ carries a circle
action, whose homotopy fixed points $\CH_{\bullet}^{hS^1}$ are calculated by the negative cyclic complex $\CC^{-}_{\bullet}(B) $. A smooth Calabi-Yau structure is a lift of the Hochschild homology element $\varphi \in \CH_{d}(B)$ to the negative cyclic homology element $\tilde{\varphi} \in \CC_{d}^{-}(B)$. In the geometric context, the class $[\tilde{\varphi}]$ corresponds to a choice of a (closed) Calabi-Yau volume form.

Now we introduce the dual notion of compact Calabi-Yau structure. Let $A$ be a proper dg algebra. Recall that we have the following isomorphism
\begin{equation*}
	\mathrm{Hom}_\mathbb{C}(A\otimes_{A^e}^{\mathbb{L}}A,\mathbb{C}) \cong R\mathrm{Hom}_{A^e}(A,A^{*})\,.
\end{equation*}

We define a (weak) compact $d$-Calabi-Yau structure on $A$ as a map $\eta : \CH_{\bullet}(A) \to \mathbb{C}[-d]$, such that it induces an isomorphism
\begin{equation*}
	\eta: A[d] \to A^*\,.
\end{equation*}

The homotopy orbit of the Hochschild complex $\CH_{\bullet}(A)_{hS^1}$ with respect to the circle action is the cyclic complex $\CC_{\bullet}(A)$. A compact Calabi-Yau algebra is a lift of $\eta$ into a cyclic chain $\tilde{\eta}: \CC_{\bullet}(A) \to \mathbb{C}[-d]$.

The concept of compact Calabi-Yau algebra is closely related to cyclic $A_\infty$ algebra. A cyclic $A_\infty$ algebra is a  $A_\infty$ algebra $(A,\{m_n\})$ equipped with a non-degenerate symmetric pairing $\langle \cdot,\cdot\rangle:A\otimes A \to \mathbb{C}$, such that the expressions $\langle a_0,m_n(a_1,\dots,a_n)\rangle$ are graded cyclically symmetric. Due to Kontsevich and Soibelman \cite{Kontsevich2008NotesOA}, these two concepts are essentially equivalent. Any cyclic  $A_\infty$ algebra has a canonically defined compact Calabi-Yau strucutre, and any compact Calabi-Yau algebra is quasi-isomorphic to a cyclic $A_\infty$ algebra.

In this paper, we adopt the notion of a "cyclic $A_\infty$ algebra" as our definition of a Calabi-Yau algebra. This is exactly the structure we need to define the large $N$ chiral gauge theory in Section \ref{sec:HAcalculations}. Additionally, we simplify by considering the $A_\infty$ algebra to be a dg associative algebra.

Though we mainly used the compact CY algebra in our construction, the smooth CY algebra also appears in the chiral algebra, after we add the (anti-)fundamental matter fields. The mesonic operators naturally give rise to smooth Calabi-Yau algebra as they correspond to space filling branes in our system. To make the connection more clear, we recall that compact and smooth CY structure are related by Koszul duality \cite{cohen2015}
\begin{equation}
    \mathrm{RHom}_{A^{e}}(\mathbb{K},\mathbb{K}) = B, \quad \mathrm{RHom}_{B^{e}}(\mathbb{K},\mathbb{K}) = A\,.
\end{equation}
Here $\mathbb{K}$ is not always the base field $\bC$, but could also be copies of $\bC$, each corresponds to an idempotent.

As commented in Section \ref{sec:generla_flavor}, in the mesonic example, its natural to choose fundamental matter that correspond to $\mathbb{K}$ and anti-fundamental to its dual. Then quasi-primary mesonic operators are given by the Koszul dual $B = \mathrm{RHom}_{A^{e}}(\mathbb{K},\mathbb{K}) $ of $A$, which is a smooth Calabi-Yau algebra. 

It would be important to show that the global symmetry algebra $\fP_\lambda$ defined in the main text is smooth 3d-CY, possibly using structures which arise in the chiral algebra. We do not know how to do so.

A great source of Calai-Yau algebra comes from quiver constructions. If we consider our $A$ to be the compact CY algebra that correspond to a quiver $Q$, then its Koszul dual smooth CY algebra is the $2d$-Ginzburg dg algebra \cite{Ginzburg:2006fu,etgu2017koszul} that correspond to the same quiver. It will be interesting to further explore the relation between $\fP_\lambda$ and the Ginzburg dg algebra in the quiver case.

\section{A lightning review of Maurer-Cartan equations and BRST anomalies}\label{app:MC_equation_and_BRST_anomalies}
We refer the reader to \cite{Gaiotto:2024gii} for a detailed discussion. 

In many QFT calculations, it is useful to define perturbatively a 
formal deformation of a reference theory, by regularizing in some scheme the exponential of an integrated interaction. We can denote that very schematically as 
\begin{equation}
    \left[e^{\int\Phi}\right]\,.
\end{equation}
Either the choice of interactions or the regularization scheme may introduce BRST anomalies. We can write schematically
\begin{equation}
    e^{\epsilon Q} \left[e^{\int\Phi}\right] = \left[e^{\int\Phi +\epsilon \text{MC}(\Phi)} \right]
\end{equation}
by rewriting the effect of a BRST transformation as a further deformation of the theory. The expression $\text{MC}(\Phi)$ is a multi-linear map from the space of interactions to itself, or better an odd nilpotent vectorfield on the space of formal couplings. The condition for BRST anomalies to vanish is 
the {\it Maurer-Cartan equation} 
\begin{equation}
    \text{MC}(\Phi)=0\,.
\end{equation}
By definition, this data equips the space of interactions with the structure of an $L_\infty$ algebra. 

Other structures can be produced by coupling the original theory to auxiliary fields, adding interactions to defects in the QFT, etc. 
For example, a Wilson-like line defect involving an auxiliary finite-dimensional quantum-mechanical system involves a path-ordered interaction 
and a matrix-valued Maurer-Cartan equation, which by definition leads to an $A_\infty$ algebra. 

\section{Non-commutative Algebra From OPEs}
\label{app:non_com_OPE}

The product rule \eqref{eq:ad_full_right_proudct} for the algebra of meson modes in the  chiral algebra with $OSp$ symmetry, was originally gleaned by studying in detail the commutators between the meson modes. In this appendix we showcase said calculations which now serve as additional evidence of the final product rule \eqref{eq:ad_full_right_proudct}.

\subsection{A useful Q-exact relation}
Given that our mesons have ghost number 0, $Q$-exact relations between them arise from taking $Q$ of mesons with ghost number -1, that is, with at least one $b$ field. Take for example the operator $IbZJ$, where by $Z$ we mean a field valued in the ghost number one part of the algebra $Z = Z^\alpha\theta_\alpha$, and similarly, $I=I^i\tilde{m}_i$ and $J = J_i m^i$ are right and left module valued respectively, with the meson $IbZJ$ living in the corresponding tensor product. Its $Q$ is given by
\begin{align}
Q IbZJ = -I(Z,Z)ZJ + \lambda  \partial IZJ -\lambda I\partial ZJ + \text{multi-meson terms}\,,
\end{align}
with $(\cdot,\cdot)$ the usual algebra pairing $(Z,Z)=\omega_{ab}Z^aZ^b$. Neglecting the multi-meson terms, we can read this equation to mean that
\begin{align}
\lambda  \partial IZJ\sim I(Z,Z)ZJ + I\partial ZJ\,.
\end{align}

Note this allowed us to move the partial derivative one symbol to the right at the price of adding a new term with an extra $(Z,Z)$ in between the two symbols. Studying the $Q$ of further mesons with one $b$ we would find that this pattern holds more generally, so that $Q$-exact relations allow to move derivatives one symbol to the right, at the price of a new term with an extra $(Z,Z)$. We may write this as
\begin{align}\label{eq:ad_q_exact_relation}
\lambda \overset{\leftarrow}{\partial} \sim \lambda \overset{\rightarrow}{\partial} + (Z,Z) \,.
\end{align}



\subsection{\texorpdfstring{$u^{a_1}_{n_1}\cdot u^{a_2}_{n_2}$}{u1 x u2} from OPEs}
The first mesons in the cohomology are
\begin{align}
I^i Z J_j\,.
\end{align}
We denote their modes by 
\begin{align}
(u^a_n)^i_j = \oint z^n I^i Z^a J_j\,.
\end{align}
As explained in Section \ref{subsec:open_symmetry_algebra}, to determine their product we must study their commutator:
\begin{align}
[(u^{a_1}_{n_1})^i_j,(u^{a_2}_{n_2})^k_l ] = \delta^k_j (u^{a_1}_{n_1} \cdot u^{a_2}_{n_2})^i_l \pm \delta^i_l (u^{a_1}_{n_1} \cdot u^{a_2}_{n_2})^k_j\,.
\end{align}
We set $i>j=k>l$ to narrow in on the product
\begin{align}
[(u^{a_1}_{n_1})^i_j,(u^{a_2}_{n_2})^j_l ] = (u^{a_1}_{n_1} \cdot u^{a_2}_{n_2})^i_l \,.
\end{align}
To determine this commutator, let's first go over the OPE between the corresponding mesons 
\begin{align}\label{eq:ad_1x1_OPE}
I^i Z^{a_1} J_j (z)\, I^j Z^{a_2} J_k (w) \sim &\frac{1}{z-w}(I^i Z^{a_1} Z^{a_2} J_k (w) + \lambda \omega^{a_1a_2} \partial I^iJ_k (w)) + \nonumber\\
&\frac{1}{(z-w)^2} \lambda \omega^{a_1a_2} I^iJ_k (w))\,.
\end{align}

The operators on the right hand side are guaranteed to be $Q$-closed. However, they are not expressed in terms of the more natural representatives \eqref{eqn_towerf_nc}
\begin{align}\label{eq:ad_canonical_meson_representatives}
IZ^{[a_1} \cdots Z^{a_k]} J\,.
\end{align}

In particular, these representatives contain no derivatives, so one might suspect that if one could ``get rid" of the meson with a derivative in \eqref{eq:ad_1x1_OPE} using $Q$-exact relations, one would find an expression in terms of our canonical representatives \eqref{eq:ad_canonical_meson_representatives}. Indeed, if for $\lambda\partial I^iJ_k$ we move ``half of its derivative" to the right
\begin{align}
\lambda \partial I^iJ_k  &\sim \lambda \frac{1}{2}\partial I^iJ_k +\lambda\frac{1}{2}I^i \partial J_k+\frac{1}{2}I^i (Z, Z) J_k \nonumber\\
&=\lambda \frac{1}{2}\partial (I^iJ_k) +\frac{1}{2}I^i (Z, Z) J_k\,.
\end{align}

(Note $\sim$ here means $Q$-equivalence class and not OPE; we'll freely alternate between the two meanings) one is left with a term with no derivatives plus a total derivative term. However, a mode of a total derivative can be expressed in terms of modes of mesons with no derivatives through integration by parts, so at this point one effectively has ``removed" the derivative terms and indeed plugging back into the OPE
\begin{align}
I^i Z^{a_1} J_j (z) I^j Z^{a_2} J_k (w) \sim &\frac{1}{z-w}(I^i Z^{a_1} Z^{a_2} J_k (w) + \frac{1}{2}\omega^{a_1 a_2}I^i (Z, Z) J_k (w) + \frac{1}{2}\lambda \omega^{a_1 a_2} \partial (I^iJ_k) (w))  \nonumber\\
&\frac{1}{(z-w)^2} \lambda \omega^{a_1 a_2} I^iJ_k (w))\,.
\end{align}

We find the traceless linear combination from \eqref{eq_ad:Z2_traceless} 
\begin{align}
I^i Z^{[a_1} Z^{a_2]}  J_k = I^i Z^{a_1} Z^{a_2} J_k  + \frac{1}{2}\omega^{a_1 a_2}I^i (Z, Z) J_k 
\end{align}
in the OPE, which is our canonical representative for the $Q$-closed meson with two $Z$'s.
Moving on to study the modes, we take the integral $\oint z^{n_1}$ of the OPE
\begin{align}
[(u^{a_1}_{n_1})^i_j,I^j Z^{a_2} J_k (w)] = w^{n_1}&(I^i Z^{[a_1} Z^{a_2]} J_k (w)  + \frac{1}{2}\lambda \omega^{a_1 a_2} \partial (I^iJ_k) (w))  +\nonumber\\
&n_1w^{n_1-1} \lambda \omega^{a_1 a_2} I^iJ_k (w)\,.
\end{align}
Taking now the integral $\oint w^{n_2}$ of what is left and expressing in terms of the modes
\begin{align}
u^{a_1a_2}_{n_1n_2} &= \oint z^{n_1+n_2} I^i Z^{[a} Z^{b]}  J_k\\
u_{n_1+n_2-1} &= \oint z^{n_1+n_2-1} I^i J_k\,,
\end{align}
we find the product
\begin{align}
u^{a_1}_{n_1} \cdot u^{a_2}_{n_2} &= u^{a_1a_2}_{n_1n_2} + \lambda\frac{n_1 -n_2}{2}\omega^{a_1a_2}u^{a_1a_2}_{n_1+n_2-1}\nonumber\\
&= u^{a_1a_2}_{n_1n_2} - \lambda\frac{1}{2}\epsilon_{n_1n_2}\omega^{a_1a_2}u^{a_1a_2}_{n_1+n_2-1}\,,
\end{align}
with $\epsilon_{n_1n_2} = \begin{pmatrix}0&1\\-1&0\end{pmatrix}_{n_1n_2} =n_2-n_1$ for $n_1, n_2 \in \{0,1\}$. One can further simplify the expression using the fact that 
\begin{align}
\epsilon_{n_1n_2}u_{n_1+n_2-1} = \epsilon_{n_1n_2}u_{0}\,,
\end{align}
since $n_1+n_2$ must add up to $1$ to have a non zero answer. Moreover, easy OPE calculations show $u_{0}$ behaves as a unit for the algebra, so we may set it to one henceforth. This leads to the final expression for product
\begin{align}
u^{a_1}_{n_1} \cdot u^{a_2}_{n_2} = u^{a_1a_2}_{n_1n_2} - \lambda\frac{1}{2}\epsilon_{n_1n_2}\omega^{a_1a_2}\,.
\end{align}

\subsection{\texorpdfstring{$u^{a_1a_2}_{n_1n_2}\cdot u^{a_3}_{n_3}$}{u12 x u3} from OPEs}We move on to study the product
\begin{align}
u^{a_1a_2}_{n_1n_2}\cdot u^{a_3}_{n_3}
\end{align}
by means of 2d CFT calculations. To do so we will need to study the OPE between
\begin{align}
I^iZ^{[a_1}Z^{a_2]}J_j = I^iZ^{a_1}Z^{a_2}J_j + \frac{\omega^{a_1a_2}}{2}I^i(Z,Z)J_j\,,
\end{align}
and 
\begin{align}
I^jZ^{a_3}J_k\,,
\end{align}

Plowing forward:
\begin{align}\label{eq:ad_2x1_OPE}
I^iZ^{[a_1}&Z^{a_2]}J_j(z)\,I^jZ^{a_3}J_k(w)\sim \nonumber\\
&\frac{1}{z-w}\left(I^iZ^{a_1}Z^{a_2}Z^{a_3}J_k + \frac{1}{2}\omega^{a_1a_2}I^i(Z,Z)Z^{a_3}J_k + \right.\nonumber\\ 
 &\quad\qquad\qquad\left.\lambda\omega^{a_2a_3} \partial(I^iZ^{a_1})J_k + \frac{\lambda}{2}\omega^{a_1a_2}\partial(I^iZ^{a_3})J_k\right) +\nonumber\\
 &+\frac{\lambda}{(z-w)^2}\left(\omega^{a_2a_3}I^iZ^{a_1}J_k + \frac{1}{2}\omega^{a_1a_2}I^iZ^{a_3}J_k\right)\,.
\end{align}

Proceeding similarly as in the $u^{a_1}_{n_1} \cdot u^{a_2}_{n_2}$ product, one may turn the operators with derivatives into operators without derivatives plus total derivatives using the $Q$-exact relation \eqref{eq:ad_q_exact_relation}
\begin{align}
\lambda\partial(IZ)J &=\lambda\partial IZJ+\lambda I \partial ZJ
\\&\sim \lambda\frac{2}{3} \partial(IZJ) + \frac{1}{3}I(Z,Z)ZJ + \frac{2}{3}IZ(Z,Z)J\,.
\end{align}

Substituting this back into the two terms with derivatives in the OPE \eqref{eq:ad_2x1_OPE}, we obtain
\begin{align}
I^iZ^{[a_1}&Z^{a_2]}J_j(z)\,I^jZ^{a_3}J_k(w)\sim \nonumber\\
&\frac{1}{z-w}\biggl(I^iZ^{[a_1}Z^{a_2}Z^{a_3]}J_k +\nonumber\\ 
 &\quad\qquad\qquad\left.\lambda\frac{2}{3}\omega^{a_2a_3} \partial(I^iZ^{a_1}J_k) + \lambda\frac{1}{3} \omega^{a_1a_2}\partial(I^iZ^{a_3}J_k)\right) +\nonumber\\
 &+\frac{\lambda}{(z-w)^2}\left(\omega^{a_2a_3}I^iZ^{a_1}J_k + \frac{1}{2}\omega^{a_1a_2}I^iZ^{a_3}J_k\right)\,,
\end{align}
where it is a pleasant sanity check to find that this whole derivative business has generated precisely the traceless linear combination \eqref{eq_ad:Z3_traceless} which we rewrite here for convenience:
\begin{align}
I^iZ^{[a_1}Z^{a_2}Z^{a_3]}J_k &= I^iZ^{a_1} Z^{a_2} Z^{a_3} J_k + \frac23 \omega^{a_2 a_3} I^iZ^{a_1} (Z, Z) J_k + \frac23 \omega^{a_1 a_2}I^i(Z,Z) Z^{a_3}J_k  \cr &+\frac13 \omega^{a_1 a_2} I^iZ^{a_3} (Z,Z)J_k +\frac13 \omega^{a_2 a_3} I^i (Z,Z) Z^{a_1} J_k\,.  
\end{align}

Integrating over $\oint z^{n_1+n_2}$ to get the first mode
\begin{align}
[(u^{a_1a_2}_{n_1n_2})^i_j&, I^jZ^{a_3}J_k(w)]=\nonumber\\ 
&z^{n_1+n_2}\biggl(I^iZ^{[a_1}Z^{a_2}Z^{a_3]}J_k + \nonumber\\ 
 &\quad\qquad\qquad\left.\lambda\frac{2}{3}\omega^{a_2a_3} \partial(I^iZ^{a_1}J_k) + \lambda\frac{1}{3} \omega^{a_1a_2}\partial(I^iZ^{a_3}J_k)\right) +\nonumber\\
 &+\lambda(n_1+n_2)z^{n_1+n_2-1}\left(\omega^{a_2a_3}I^iZ^{a_1}J_k + \frac{1}{2}\omega^{a_1a_2}I^iZ^{a_3}J_k\right)\,.
\end{align}

Finally, integrating over $\oint w^{n_3}$
\begin{align}
u^{a_1a_2}_{n_1n_2} \cdot u^{a_3}_{n_3} = u^{a_1a_2a_3}_{n_1n_2n_3} + \lambda \biggl[ (n_1+n_2+n_3)&\biggl(-\frac{2}{3}\omega^{a_2a_3}u^{a_1}_{n_1+n_2+n_3-1} -\frac{1}{3}\omega^{a_1a_2}u^{a_3}_{n_1+n_2+n_3-1}\biggr)\nonumber\\
(n_1+n_2)&\biggl(\omega^{a_2a_3}u^{a_1}_{n_1+n_2+n_3-1} +\frac{1}{2}\omega^{a_1a_2}u^{a_3}_{n_1+n_2+n_3-1}\biggr)\nonumber\,,
\end{align}
which simplifies to
\begin{align}
u^{a_1a_2}_{n_1n_2} \cdot u^{a_3}_{n_3} = u^{a_1a_2a_3}_{n_1n_2n_3} + \frac{\lambda}{6}(n_1 + n_2 - 2n_3)(\omega^{a_1a_2}u^{a_3}_{n_1+n_2+n_3-1} + 2\omega^{a_2a_3}u^{a_1}_{n_1+n_2+n_3-1})\,.
\end{align}

To bring it to the form of \eqref{eq:ad_2x1_prod}, we just need to note that $n_1 + n_2 - 2n_3 = -\epsilon_{n_1n_3}-\epsilon_{n_2n_3}$ in terms of the levi-civita symbols. Moreover
\begin{align}\label{eq:ad_n_1+n_2=1}
\epsilon_{n_1n_3}u^{a}_{n_1+n_2+n_3-1} = \epsilon_{n_1n_3}u^{a}_{n_2}\\
\epsilon_{n_2n_3}u^{a}_{n_1+n_2+n_3-1} = \epsilon_{n_2n_3}u^{a}_{n_1}
\end{align}
by the same argument as the one used for \eqref{eq:ad_n_1+n_2=1}.

This leads to our final expression which matches the one obtained through algebraic means:
\begin{align}
u^{a_1a_2}_{n_1n_2} \cdot u^{a_3}_{n_3} = u^{a_1a_2a_3}_{n_1n_2n_3} - \frac{\lambda}{6}\biggl(&\epsilon_{n_1n_3}(\omega^{a_1a_2}u^{a_3}_{n_2} + 2\omega^{a_2a_3}u^{a_1}_{n_2})+\nonumber\\
&\epsilon_{n_2n_3}(\omega^{a_1a_2}u^{a_3}_{n_1} + 2\omega^{a_2a_3}u^{a_1}_{n_1})\biggr)\,.
\end{align}



\section{Generalizing to Symmetric and Anti-symmetric matrices}
\label{appendix:ortho_symp}
We wish to generalize the chiral algebra construction to the case where our fields are now valued in symmetric or anti-symmetric matrices.

It will prove useful to construct a $\bZ_2 $ action that commutes with the tree level BRST charge $Q_0$. A consistent $Q_0$ action amounts to a Lie subalgebra of the matrix Lie algebra $\mathfrak{gl}_N\otimes A$ under the $\bZ_2$ action. Such a $\bZ_2$ action, consists of an involution $a \to \bar{a}$ on the algebra $A$ and a $\bZ_2$ action on the matrices. Here are two possible choices of Lie subalgebras:
\begin{itemize}
	\item Skew symmetric matrices $\mathfrak{so}_N(A)$:
	\begin{equation}\label{eqn:skew_sym_mat}
		\mathfrak{so}_N(A) :\{\Phi \in \mathfrak{gl}(N)\otimes A\mid \bar{\Phi}_{ji} = -\Phi_{ij}\}\,.
	\end{equation}
	\item Symplectic matrices $\mathfrak{sp}_N(A)$:
		\begin{equation}\label{eqn:symp_mat}
		\mathfrak{sp}_N(A) :\{\Phi \in \mathfrak{gl}(2N)\otimes A\mid \bar{\Phi}_{ji}  = -(\Omega^{-1}\Phi\Omega)_{ij}\}\,.
	\end{equation}
\end{itemize}
Depending on the symmetry properties of the matrices, the $\bZ_2$ action in question will be different. Let's explore the different cases.

\begin{itemize}
\item  $\Phi(z) \in A \otimes \bigwedge^2\bC^N \cong A \otimes \mathfrak{so}(N)$ \\
This corresponds to \eqref{eqn:skew_sym_mat} with trivial involution $\bar{a} = a$.
We interpret these fields as living in the adjoint representation of an $SO(N)$ symmetry that has been gauged.
Here, the transposition operation commutes with $Q_0$ for free:
\begin{equation}
    Q_0 \Phi^T = Q_0 (-\Phi) = - \Phi \Phi\,,
\end{equation}
while
\begin{equation}
    (Q_0 \Phi)^T = (\Phi \Phi)^T = - (\Phi)^T (\Phi)^T = - \Phi \Phi \,,\label{eq:ad_commutation_of_Phis}
\end{equation}
where the minus sign in equation \eqref{eq:ad_commutation_of_Phis} arose from commuting one $\Phi$ past the other.

So we see:
\begin{equation}
    Q_0(\Phi)^T = (Q_0\Phi)^T \,.
\end{equation}

\item $\Phi(z) \in A \otimes \Sym^2\bC^N \cong A \otimes \mathfrak{sp}(N)$ \\

This corresponds to \eqref{eqn:symp_mat} with trivial involution $\bar{a} = a$. We interpret these fields as living in the adjoint representation of an $Sp(N)$ symmetry that has been gauged. If we define our field with one index up and one down, $\Phi^a_b$, where the symplectic form $\Omega_{ab}$ was used to lower what naturally would've been two upper indices, then its transposition symmetry reads:
\begin{equation}
    \Phi^T = - \Omega^{-1} \Phi \Omega\,.
\end{equation}


In this case, $Q_0$ also commutes with transposition. The proof of this is completely analogous to the previous one but for some extra $\Omega$'s.




\item Matter fields valued in symmetric matrices and ghosts in antisymmetric: $\Phi(z) \in \underbrace{A_0 \otimes \mathfrak{so}(N)}_{c \text{ ghosts}} \oplus \underbrace{A_1 \otimes \Sym^2\bC^N}_{Z\text{ fields}} \oplus \underbrace{A_2 \otimes \mathfrak{so}(N)}_{b\text{ ghosts}}$.\\

This corresponds to \eqref{eqn:skew_sym_mat} with involution $\bar{a} = (-1)^{g}a$. Here we interpret the $Z$ fields as living in the symmetric representation of an $SO(N)$ gauge symmetry. In this case transposition is equivalent to multiplying ghost even elements by minus one:

\begin{equation}
    \Phi^T = (-1)^{g+1} \Phi = -c + Z - b\,.
\end{equation}

Again, transposition commutes with $Q_0$:
\begin{align}
    Q_0 (\Phi)^T &= Q_0(-c+Z-b)\\
    &= -cc+[c,Z]-[c,b]-ZZ \\
    &= -(-c+Z-b)(-c+Z-b)\\
    &= -(c+Z+b)^T(c+Z+b)^T\\
    &= ((c+Z+b)(c+Z+b))^T\\
    &= (Q_0\Phi)^T\,.
\end{align}






\item Matter fields valued in antisymmetric matrices and ghosts in $\mathfrak{sp}(N)$: $\Phi(z) \in \underbrace{A_0 \otimes \mathfrak{sp}(N)}_{c \text{ ghosts}} \oplus \underbrace{A_1 \otimes \wedge^2\bC^{2N}}_{Z\text{ fields}} \oplus \underbrace{A_2 \otimes \mathfrak{sp}(N)}_{b\text{ ghosts}}$

This corresponds to \eqref{eqn:symp_mat} with involution $\bar{a} = (-1)^{g}a$. Using the same convention to raise and lower indices with the symplectic form, $\Omega$ as in the $Sp$ case, the transposition operation reads:

\begin{equation}
    \Phi^T = (-1)^{g+1} \Omega^{-1}\Phi\Omega\,.
\end{equation}

$Q_0$ also commutes with this operation and the proof is identical to the previous one except for the extra $\Omega$'s.
\end{itemize}

In all of the examples above, single trace operators can be described by the so called Dihedral cohomology $HD^{\bullet}(A[[s]])$ \cite{LODAY198893}, depending on the choice of the involution. We now briefly describe this connection. The involution on $A$ induces a $\bZ_2$ action on the Hochschild complex $A^{\otimes n}$. 
\begin{equation}
	\omega(a_1,\dots,a_n) = (\pm1)(\bar{a}_n,\bar{a}_{n-1},\dots, \bar{a_1})\,.
\end{equation}
This action, together with the cyclic $\bZ_n$ action on $A^{\otimes n}$, forms an action of the Dihedral group $D_n = \langle t,\omega \mid t^{n} = \omega^2 =  1, \omega t \omega^{-1} = t^{-1} \rangle$ on $A^{\otimes n}$. In both the $\mathfrak{so}_N(A)$ and $\mathfrak{sp}_N(A)$ cases, we consider $D_n$ invariant multilinear maps to build single trace operators
\begin{equation}
	(C|\bullet,\dots,\bullet): (A[[s]]\otimes \dots\otimes A[[s]])^{D_n} \to \bC\,.
\end{equation}
These identify the single trace operators with the Dihedral cohomology $HD^{\bullet}(A[[s]])$.

We can also consider a mixture of $SO$ and $Sp$ fields, half-hypermultiplet matter. This naturally leads to the generalized notion of a Calabi-Yau category. In fact, our chiral algebra construction for a $2$ Calabi-Yau algebra $A$ can be generalized to a (compact) $2$ Calabi-Yau category $\mathcal{C}$. Similar to our discussion in Appendix \ref{appendix:CY}, a compact $d$-dimensional Calabi-Yau category $\mathcal{C}$ is the same as a cyclic $A_\infty$ category. Here, to simplify the discussion we consider an ordinary category $\mathcal{C}$ equipped with a trace map:
\begin{equation}
	\Tr_a: \mathrm{Hom}_{\mathcal{A}}(a,a) \to \bC[d]\,,
\end{equation}
for each object $a \in \mathcal{A}$. The associated pairing
\begin{equation}
	(\bullet,\bullet): \mathrm{Hom}_{\mathcal{A}}(a,b)\otimes \mathrm{Hom}_{\mathcal{A}}(b,a) \to \bC[d]
\end{equation}
given by $(\alpha\beta) = \Tr(\alpha\beta)$ is required to be symmetric and non-degenerate. For our construction, we require that there is only a finite number of objects in the category and each $\mathrm{Hom}_{\mathcal{A}}(a,b)$ is finite dimensional.

More explicitly, a $d$-dimensional Calabi-Yau category $\mathcal{C}$ is the same as the following data
\begin{itemize}
	\item A super algebra $A_a = A_{aa} = \mathrm{End}_{\mathcal{C}}(a)$ for each object $a$.
	\item A $A_a-A_b$ bi-module $A_{ab} = \mathrm{Hom}_{\mathcal{C}}(a,b)$ for each pair of objects $a,b$.
	\item A collections of maps $A_{ab}\otimes A_{bc} \to A_{ac}$ satisfying some compatibility conditions.
	\item A trace map $\Tr: A_a \to \bC[d]$ for each object $a$, such that the pairing $(\bullet,\bullet): A_{ab}\otimes A_{ba} \to \bC[d]$ is symmetric and non-degenerate.
\end{itemize}

As before, we assign elements in $\mathrm{Hom}_{\mathcal{C}}(a,b)$ both a super degree $|\alpha|$ and a ghost degree $\text{gh}[\alpha]$. Then one can consider matrix fields $\Phi(z)$ valued in $\bigoplus_{a,b\in\mathcal{C}} A_{ab}$. We can define the OPE and BRST current as we did for a CY algebra. However, for $SU(N)$ or $GL(N)$ matrices, this ``generalization" does not gives us anything new, as we can define the algebra
\begin{equation}
	A[\mathcal{C}] = \bigoplus_{a,b}\mathrm{Hom}_{\mathcal{C}}(a,b)\,.
\end{equation}
We can see that the chiral algebra construction for the Calabi-Yau category $\mathcal{C}$ is the same as the chiral algebra associated to the Calabi-Yau algebra $A[\mathcal{C}]$.

On the contrary, the generalization to a Calabi-Yau category does give us a non-trivial generalization when we consider a mixture of $SO$ and $Sp$ fields. 

We first introduce the notion of an involution on a category. An involution on a category $\mathcal{C}$ is a collection of maps
\begin{equation}
	\alpha \in \mathrm{Hom}_{\mathcal{C}}(a,b) \to \bar{\alpha} \in  \mathrm{Hom}_{\mathcal{C}}(b,a) 
\end{equation} 
for each pair of objects $a,b$. It satisfies the condition that 
\begin{equation}
	\bar{\bar{\alpha}} = \alpha,\quad \overline{\alpha\beta} = (-1)^{|\alpha||\beta|}\bar{\beta}\bar{\alpha}\,.
\end{equation}

Now we consider a chiral algebra with a mixture of $SO$ and $Sp$ fields. We divide the set of objects $\mathrm{Obj}\,\mathcal{C}$ into two subsets $\mathrm{Obj}\,\mathcal{C} =\mathbf{O}\cup \mathbf{P}$, where an object in $\mathbf{O}$ corresponds to an $SO$ field and an object in $\mathbf{P}$ corresponds to an $Sp$ field. 

We follow \cite{Kac1977} in their definition of the Lie superalgebra $\mathfrak{osp}(n|m)$, and define the Lie algebra $\mathfrak{osp}(\mathcal{C},\mathbf{O},\mathbf{P})$ as a Lie subalgebra of
\begin{equation}
    \mathfrak{gl}(\mathcal{C}) := \mathrm{End}_{\mathcal{C}}((\bigoplus_{o\in \mathbf{O}}\bC^No)\oplus(\bigoplus_{p\in \mathbf{P}}\bC^{2N}p))\,,
\end{equation}
given by
\begin{equation}
    \mathfrak{osp}(\mathcal{C},\mathbf{O},\mathbf{P}) = \{\Phi \in \mathfrak{gl}(\mathcal{C})\mid \bar{\Phi}^T = -i^{\deg(\Phi)}B^{-1}\Phi B\}\,.
\end{equation}
Here $B$ is the matrix 
\begin{equation}
    B = (\bigoplus_{o\in \mathbf{O}}iI_N)\oplus(\bigoplus_{p\in \mathbf{P}}\Omega_N)\,.
\end{equation}
The $\deg(\Phi)$ here can be chosen as any degree compatible with the multiplication, e.g., $|\Phi|$, $\text{gh}(\Phi)$, or $|\Phi| - \text{gh}(\Phi)$. Different choices, of course, give rise to different theories. We can expand the definition explicitly. A field $\Phi$ valued in $\mathfrak{osp}(\mathcal{C},\mathbf{O},\mathbf{P})$ consists of several components. 
\begin{itemize}
    	\item We have a field $\Phi_{\mathbf{O}}(z) \in \mathfrak{so}_N(A_{\mathbf{O}})$ in $\deg = 0$. Here $A_{\mathbf{O}} = \bigoplus_{a,b \in \mathbf{O}}\mathrm{Hom}_{\mathcal{C}}(a,b)$, and $\mathfrak{so}_N(A_{\mathbf{O}})$ is given by the condition
	\begin{equation}
		\bar{\Phi}_{\mathbf{O}}(z)^T = -\Phi_{\mathbf{O}}(z)\,.
	\end{equation}
		\item We have a field $\Phi_{\mathbf{P}}(z) \in \mathfrak{sp}_N(A_{\mathbf{P}})$ in $\deg = 0$. Here $A_{\mathbf{P}} = \bigoplus_{a,b \in \mathbf{P}}\mathrm{Hom}_{\mathcal{C}}(a,b)$, and $\mathfrak{sp}_N(A_{\mathbf{P}})$ is given by the condition
	\begin{equation}
		\bar{\Phi}_{\mathbf{P}}(z)^T = -\Omega^{-1}\Phi_{\mathbf{P}}(z)\Omega\,.
	\end{equation}
 \item  We have a pair of fields $\Phi_{\mathbf{O}\mathbf{P}} \in \mathrm{Hom}(\bC^N,\bC^{2N})\otimes A_{\mathbf{O}\mathbf{P}}$ and $\Phi_{\mathbf{P}\mathbf{O}} \in \mathrm{Hom}(\bC^{2N},\bC^{N})\otimes A_{\mathbf{P}\mathbf{O}}$ in $\deg = 1$. This pair of fields is constrained by
 \begin{equation}
     \bar{\Phi}_{\mathbf{O}\mathbf{P}}^T = \Omega^{-1}\Phi_{\mathbf{P}\mathbf{O}},\quad\quad 
     \bar{\Phi}_{\mathbf{P}\mathbf{O}}^T = -\Phi_{\mathbf{O}\mathbf{P}}\Omega\,.
 \end{equation}
\end{itemize}

It is straight forward to generalize the above construction to the case when each node (object) $a$ corresponds to a distinct rank $k_aN$. 

\bibliographystyle{JHEP}

\bibliography{mono}

\end{document}